%% file: main.tex
\documentclass[a4paper,UKenglish,cleveref, autoref, thm-restate,dvipsnames]{lipics-v2021}

\usepackage{todonotes}
\usepackage{xspace}
\usepackage{mathtools}
\usepackage{multicol}
\usepackage{stmaryrd}
\usepackage{amsthm}
\usepackage{amssymb}
\usepackage{nicefrac}
\usepackage{xparse}
\usepackage{array}
\usepackage{makecell}
\usepackage{wasysym}
\usepackage{cleveref}

\hideLIPIcs  

\title{Reaching as Cheap as Possible in 1-clock Robust Weighted Timed Games} 

\author{Nathalie Bertrand}{Univ Rennes, Inria, CNRS, IRISA}{nathalie.bertrand@inria.fr}{http://orcid.org/0000-0002-9957-5394}{}
\author{Maëlle Gautrin}{ENS Paris-Saclay}{maelle.gautrin@ens-paris-saclay.fr}{}{}
\author{Julie Parreaux}{Univ Rennes, Inria, CNRS, IRISA}{julie.parreaux@irisa.fr}{https://orcid.org/0009-0009-2744-780X}{}

\authorrunning{N. Bertrand, M. Gautrin, J. Parreaux} 

\Copyright{Nathalie Bertrand, Maëlle Gautrin and Julie Parreaux}



\category{}
\funding{This work has been partially funded by the project ANR BisoUS
  (ANR-22-CE48-0012).}

\nolinenumbers

\EventEditors{~}
\EventNoEds{1}
\EventLongTitle{}
\EventShortTitle{}
\EventAcronym{}
\EventYear{2026}
\EventDate{September 2--5, 2026}
\EventLocation{Liverpool, United Kingdom}
\EventLogo{}
\SeriesVolume{42}
\ArticleNo{1}

\input{macros}

\begin{document}

\maketitle

\begin{abstract}
  The value problem for 2-player games on graph generally consists in
  determining the minimal value $\MinPl$ can ensure against any
  possible strategy for $\MaxPl$. We consider here the value problem
  for reachability objectives in weighted timed games (WTGs) under a
  robust semantics. WTGs are a modelling formalism combining real-time
  constraints and integer weights on transitions and locations in an
  adversarial setting. Robustness allows for representing timing
  imprecisions in the measurement of delays and clock values. Robust
  weighted timed games have been introduced more than a decade ago:
  they are undecidable in general, and were quite recently shown
  decidable for the subclasses of acyclic or divergent robust
  WTGs. This paper pursues the goal of identifying decidable
  subclasses and establishes the decidability of the robust value
  problem for 1-clock WTGs.
\end{abstract}

\section{Introduction}
\label{sec:intro}
\input{introduction}

\section{Robust Weighted Timed Games}
\label{sec:prelim}
\input{prelim}

\section{Characterization of Robust Deadlocks}
\label{sec:deadlocks}
\input{deadlocks}

\section{Transformation into the Copy Game}
\label{sec:copy}
\input{copy}

\section{Recovering the Original Robust Value in the Copy Game}
\label{sec:value}
\input{values}

\section{Conclusion}\label{sec:conc}
\input{conclusion}

\bibliography{bibliography}

	\newpage
	\appendix
        The following appendices contain all proofs of our results.
	\input{app_asap}
\newpage
\input{app_blocage}

\newpage
\input{app_copy}
\newpage
\input{app_copy-rval}
\newpage
\input{app_copy-well-formed}
\newpage
\input{app_copy_e-to-wf}
\newpage
\input{app_copy-rob-eq-val}

\end{document}

%% file: macros.tex
\usepackage{tikz}
\usetikzlibrary{automata,positioning,shapes,arrows,decorations,decorations.pathmorphing,decorations.markings,bending}
\usetikzlibrary{shapes.misc}
\usetikzlibrary{backgrounds}

\tikzset{every loop/.style={looseness=7}, >=latex}
\tikzset{every picture/.style={>=latex}}
\tikzset{min/.style={circle, draw, minimum size=7mm,inner sep=1.5pt,font=\small},
	max/.style={rectangle, draw, minimum size=7mm,inner sep=1.5pt,font=\small},
	trans/.style={->, >=stealth, shorten >=1pt,thick},
	guard/.style={font=\small, above, sloped},
	reset/.style={font=\small, below, sloped},
	weightlabel/.style={font=\small, below, sloped},
	target/.style={circle, font=\LARGE, minimum size=1mm,inner sep=-2pt}}

\tikzstyle{strat} =[minimum width=0.1cm,line width=0.01mm,draw=none]
\tikzstyle{vecArrow} = [decoration={markings,mark=at position
1 with {\arrow[scale=.5,>=latex]{>}}}]

\DeclareMathOperator*{\arginf}{\arg\!\inf}
\DeclareMathOperator*{\argsup}{\arg\!\sup}

\newcommand{\arginfepsilon}[1][]{\ifthenelse{\equal{#1}{}}{{\arginf}^{\varepsilon}}{{\arginf}^{#1}}}
\newcommand{\argsupepsilon}[1][]{\ifthenelse{\equal{#1}{}}{{\argsup}^{\varepsilon}}{{\argsup}^{#1}}}
\newcommand{\tuple}[1]{\ensuremath{\langle #1 \rangle}\xspace}

\newcommand{\last}{\textsf{last}}

\newcommand{\convert}{\ensuremath{\mathsf{convert}}\xspace}
\newcommand{\wt}{\ensuremath{\mathsf{wt}}\xspace}
\newcommand{\wf}{\ensuremath{\mathrm{wf}}\xspace}
\newcommand{\emath}{\ensuremath{\mathrm{e}}\xspace}
\newcommand{\rmath}{\ensuremath{\mathrm{r}}\xspace}

\newcommand{\reals}{\ensuremath{\mathbb{R}}\xspace}
\newcommand{\nats}{\ensuremath{\mathbb{N}}\xspace}
\newcommand{\R}{\ensuremath{\mathbb{R}}\xspace}
\newcommand{\Rplus}{\ensuremath{\R_{\geq 0}}\xspace}

\newcommand{\Z}{\ensuremath{\mathbb{Z}}\xspace}
\newcommand{\N}{\ensuremath{\mathbb{N}}\xspace}
\newcommand{\Q}{\ensuremath{\mathbb{Q}}\xspace}

\newcommand{\C}{\mathcal C}

\newcommand{\MinPl}{\ensuremath{\mathsf{Min}}\xspace}
\newcommand{\MaxPl}{\ensuremath{\mathsf{Max}}\xspace}

\newcommand{\WTG}{WTG\xspace}

\newcommand{\game}{\ensuremath{\mathcal G}\xspace}
\newcommand{\gameEx}{\ensuremath{\tuple{\LocsMin, \allowbreak\LocsMax, \allowbreak\LocsT,
	 \allowbreak\Trans,\weight}}\xspace}

\newcommand{\clockx}{x}

\newcommand{\delay}{\ensuremath{t}\xspace}

\newcommand{\clockbound}{\ensuremath{\mathsf{M}}\xspace}
\newcommand{\val}{\ensuremath{\nu}\xspace}
\newcommand{\guard}{\ensuremath{I}\xspace}
\newcommand{\leftg}[1]{\textsf{left}(#1)}
\newcommand{\rightg}[1]{\textsf{right}(#1)}
\newcommand{\closureg}[1]{\overline{#1}}
\newcommand{\Guards}{\ensuremath{\mathsf{Guards}}\xspace}
\newcommand{\reset}{\ensuremath{Y}\xspace}

\newcommand{\loc}{\ensuremath{\ell}\xspace}
\newcommand{\Locs}{\ensuremath{L}\xspace}

\newcommand{\LocsMin}{\ensuremath{\Locs_{\MinPl}}\xspace}
\newcommand{\LocsMax}{\ensuremath{\Locs_{\MaxPl}}\xspace}

\newcommand{\LocsT}{\ensuremath{\Locs_T}\xspace}
\newcommand{\Trans}{\ensuremath{\Delta}\xspace}
\newcommand{\trans}{\ensuremath{\delta}\xspace}

\newcommand{\maxBound}{\ensuremath{M}}

\newcommand{\RTgame}{\ensuremath{\mathcal{A}}\xspace}

\newcommand{\rRTgame}[1][]{\ifthenelse{\equal{#1}{}}
{\ensuremath{\mathcal{R}(\RTgame)}\xspace}
{\ensuremath{\mathcal{R}^{#1}(\RTgame)}\xspace}}

\newcommand{\weight}{\ensuremath{\mathsf{wt}}\xspace}

\newcommand{\weightC}{\ensuremath{\weight_\Sigma}\xspace}

\newcommand{\maxWeightLoc}{W_{\mathsf{loc}}}

\newcommand{\minstrategy}{\ensuremath{\sigma}\xspace}
\newcommand{\maxstrategy}{\ensuremath{\tau}\xspace}

\newcommand{\robminstrategy}[1][]{\ifthenelse{\equal{#1}{}}
{\ensuremath{\minstrategy^{\perturbation}\xspace}}{\ensuremath{\minstrategy^{\perturbation,#1}\xspace}}}
\newcommand{\robmaxstrategy}[1][]{\ifthenelse{\equal{#1}{}}
{\ensuremath{\maxstrategy^{\perturbation}\xspace}}{\ensuremath{\maxstrategy^{\perturbation,#1}\xspace}}}

\newcommand{\Strat}[1][]{\ifthenelse{\equal{#1}{}}
{\ensuremath{\ensuremath{\mathsf{Strat^0}}\xspace}}{\ensuremath{\ensuremath{\mathsf{Strat}^0}_{#1}}}\xspace}
\newcommand{\StratMin}[1][]{\ifthenelse{\equal{#1}{}}
{\ensuremath{\Strat_{\MinPl}}}{\Strat_{#1, \MinPl}}\xspace}
\newcommand{\StratMax}[1][]{\ifthenelse{\equal{#1}{}}
{\ensuremath{\Strat_{\MaxPl}}}{\Strat_{#1, \MaxPl}}\xspace}

\newcommand{\rStrat}[1][]{\ifthenelse{\equal{#1}{}}
{\ensuremath{\ensuremath{\mathsf{Strat}^\perturbation}\xspace}}
{\ensuremath{\ensuremath{\mathsf{Strat}}^\perturbation_{#1}}}\xspace}
\newcommand{\rStratMin}[1][]{\ifthenelse{\equal{#1}{}}
{\ensuremath{\rStrat_{\MinPl}}}{\rStrat_{#1, \MinPl}\xspace}}
\newcommand{\rStratMax}[1][]{\ifthenelse{\equal{#1}{}}
{\ensuremath{\rStrat_{\MaxPl}}}{\rStrat_{#1, \MaxPl}\xspace}}

\newcommand{\Value}{\ensuremath{\mathsf{Val}}\xspace}
\newcommand{\dValue}{\ensuremath{\mathsf{rVal}^0}\xspace}

\newcommand{\rValue}[1][]{\ifthenelse{\equal{#1}{}}
{\ensuremath{\mathsf{rVal}^\perturbation}}
{\ensuremath{\mathsf{rVal}^{#1}}}}
\newcommand{\rupperValue}[1][]{\ifthenelse{\equal{#1}{}}
{\ensuremath{\overline{\mathsf{rVal}}^{\perturbation}}}
{\ensuremath{\overline{\mathsf{rVal}}^{\perturbation,#1}}}}
\newcommand{\rlowerValue}[1][]{\ifthenelse{\equal{#1}{}}
{\ensuremath{\underline{\mathsf{rVal}}^{\perturbation}}}
{\ensuremath{\underline{\mathsf{rVal}}^{\perturbation,#1}}}}

\newcommand{\rValueL}[1][]{\ifthenelse{\equal{#1}{}}
{\ensuremath{\overline{\mathsf{rVal}}}}
{\ensuremath{\overline{\mathsf{rVal}}^{#1}}}}
\newcommand{\bValue}[1][]{\ifthenelse{\equal{#1}{}}
{\ensuremath{\overline{\frac{1}{2}\mathsf{rVal}}^{\perturbation}}}
{\ensuremath{\overline{\frac{1}{2}\mathsf{rVal}}^{\perturbation,#1}}}}
\newcommand{\bValueL}[1][]{\ifthenelse{\equal{#1}{}}
{\ensuremath{\overline{\frac{1}{2}\mathsf{rVal}}}}
{\ensuremath{\overline{\frac{1}{2}\mathsf{rVal}}^{#1}}}}
\newcommand{\rValueS}[1][]{\ifthenelse{\equal{#1}{}}
{\ensuremath{\mathsf{rVal}^{\perturbation}}}
{\ensuremath{\mathsf{rVal}^{\perturbation,#1}}}}

\newcommand{\rgame}[1][]{\ifthenelse{\equal{#1}{}}
{\ensuremath{\mathcal{R}(\game)}\xspace}
{\ensuremath{\mathcal{R}^{#1}(\game)}\xspace}}

\newcommand{\stateI}{\ensuremath{(\loc_0, 0)}\xspace}

\newcommand{\moveto}[1]{\ensuremath{\xrightarrow{#1}}\xspace}
\newcommand{\move}[1][]{\ifthenelse{\equal{#1}{}}
{\ensuremath{\trans, \delay}\xspace}
{\ensuremath{\trans_{#1}, \delay_{#1}}\xspace}}
\newcommand{\moveP}[1][]{\ifthenelse{\equal{#1}{}}
{\ensuremath{\trans', \delay'}\xspace}
{\ensuremath{\trans'_{#1}, \delay'_{#1}}\xspace}}
\newcommand{\umove}[1][]{\ifthenelse{\equal{#1}{}}
{\ensuremath{\utrans, \delay}\xspace}
{\ensuremath{\utrans_{#1}, \delay_{#1}}\xspace}}

\newcommand{\play}{\ensuremath{\rho}\xspace}

\newcommand{\Play}{\ensuremath{\mathsf{Play}}\xspace}
\newcommand{\rPlay}[1][]{\ifthenelse{\equal{#1}{}}
{\ensuremath{\mathsf{rPlay}^{\perturbation}}\xspace}{\ensuremath{\mathsf{rPlay}^{#1}}\xspace}}
\newcommand{\rPlays}[1][]{\ifthenelse{\equal{#1}{}}
		{\ensuremath{\mathsf{rPlays}^{\perturbation}}\xspace}{\ensuremath{\mathsf{rPlays}^{#1}}\xspace}}

\newcommand{\FPlays}[1][]{\ifthenelse{\equal{#1}{}}
{\ensuremath{\mathsf{FPlays}}\xspace}{\ensuremath{\mathsf{FPlays}^{#1}}\xspace}}
\newcommand{\FPlaysMin}[1][]{\ensuremath{\FPlays[#1]_\MinPl}\xspace}
\newcommand{\FPlaysMax}[1][]{\ensuremath{\FPlays[#1]_\MaxPl}\xspace}

\newcommand{\outcomes}{\ensuremath{\mathsf{Play}}\xspace}
\newcommand{\rPlaysMin}[1][]{\ensuremath{\rPlays[#1]_\MinPl}\xspace}

\newcommand{\FPlaysG}[1][]{\ensuremath{\FPlays[#1]_{\game}}}
\newcommand{\FPlaysGMin}[1][]{\FPlays[#1]_{\game, \MinPl}}
\newcommand{\FPlaysGMax}[1][]{\FPlays[#1]_{\game, \MaxPl}}

\newcommand{\rFPlays}[1][]{\ifthenelse{\equal{#1}{}}
{\ensuremath{\mathsf{FPlays}^{\perturbation}}\xspace}{\ensuremath{\mathsf{FPlays}^{\perturbation, #1}}\xspace}}

\newcommand{\ppath}{\ensuremath{\pi}\xspace}

\newcommand{\Paths}{\ensuremath{\mathsf{Paths}}\xspace}
\newcommand{\FPaths}{\ensuremath{\mathsf{FPaths}}\xspace}

\newcommand{\PPaths}{\ensuremath{\mathsf{FPaths}}\xspace}
\newcommand{\PPathsMin}{\PPaths_\MinPl}
\newcommand{\PPathsMax}{\PPaths_\MaxPl}

\newcommand{\PSPACE}{\ensuremath{\mathsf{PSPACE}}\xspace}
\newcommand{\EXPTIME}{\ensuremath{\mathsf{EXPTIME}}\xspace}

\newcommand{\perturbation}{\ensuremath{p}\xspace}
\newcommand{\perturbationBound}{\ensuremath{\eta}\xspace}

\newcommand{\pertInt}[1][]{\ifthenelse{\equal{#1}{}}
{\ensuremath{[0, 2\perturbationBound]}}
{\ensuremath{[0, 2#1]}\xspace}}
\newcommand{\sem}[2][]{\ifthenelse{\equal{#1}{}}
{\ensuremath{\llbracket#2\rrbracket}}
{\ensuremath{\llbracket#2\rrbracket^{#1}}}\xspace}

\newcommand{\States}{\ensuremath{S}}
\newcommand{\StatesMin}{\ensuremath{\States_{\MinPl}}\xspace}
\newcommand{\StatesMax}{\ensuremath{\States_{\MaxPl}}\xspace}
\newcommand{\StatesT}{\ensuremath{\States_{T}}\xspace}
\newcommand{\states}{\ensuremath{s}}
\newcommand{\Edges}[1][]{\ifthenelse{\equal{#1}{}}
{\ensuremath{E}^{\perturbation}\xspace}
{\ensuremath{\mathsf{Conf}}^{#1}\xspace}}
\newcommand{\EdgesMin}[1][]{\ensuremath{\Edges[#1]_{\MinPl}}\xspace}
\newcommand{\EdgesMax}[1][]{\ensuremath{\Edges[#1]_{\MaxPl}}\xspace}
\newcommand{\EdgesRob}[1][]{\ensuremath{\Edges[#1]_{\textsf{prtb}}}\xspace}
\newcommand{\edge}{\ensuremath{e}}

\newcommand{\Conf}[1][]{\ifthenelse{\equal{#1}{}}
{\ensuremath{\mathsf{Conf}}\xspace}
{\ensuremath{\mathsf{Conf}}^{#1}\xspace}}
\newcommand{\ConfMin}[1][]{\ensuremath{\Conf[#1]_{\MinPl}}\xspace}
\newcommand{\ConfT}[1][]{\ensuremath{\Conf[#1]_{T}}\xspace}
\newcommand{\ConfMax}[1][]{\ensuremath{\Conf[#1]_{\MaxPl}}\xspace}

%% file: introduction.tex
\paragraph*{Timed Automata and Weighted Timed Games}
imed automata form an elegant and well-studied model for real-time systems~\cite{AD-tcs94}.
To check for the existence of an error execution, which can be phrased as a reachability question,
efficient algorithms and tools have been developed since then.

To represent real-time controller synthesis problems, extensions of timed automata with weights
and two antagonistic players have been introduced under the terminology of \emph{weighted timed games}
(WTGs, for short)~\cite{Alur2004, Bouyer2004}.
In WTGs, locations and transitions are equipped with integer weights.
During an execution,transition weights are summed and added to the location weights scaled by the elapsed time.
The objective of Player \MinPl is to reach a target set of locations while minimising the cumulative weight.
The value problem asks whether \MinPl can ensure the cumulative weight is below a given threshold.

Interestingly, the presence of both negative and positive weights
makes the value problem for WTGs challenging: it is undecidable in
general~\cite{Brihaye2014}. This negative result for WTGs with at
least three clocks has been recently improved to show the
undecidability of the value problem for WTGs with two
clocks~\cite{GOV-fossacs26}.  When restricting to 1-clock WTGs, i.e.\
WTGs for which the underlying timed automaton has a single clock, it
has been shown that the value can be computed by various
algorithms~\cite{bouyer2006almost,HIM-concur13}, yet the value problem
is \PSPACE-hard~\cite{Fearnley2020}.

The decidability frontier of the related problem of approximating the
value in WTGs has also been studied. It is undecidable in
general~\cite{GO-concur24} and decidable for restricted classes for
non-negative weights only~\cite{bouyer2015value} and also including
negative weights~\cite{BusattoGaston2018}.

\paragraph*{Robustness}
Timed automata are an elegant model, but they expose some undesirable
artefacts from mathematical modelling.  Specifically, the precision in
time measurement is supposed to be arbitrary, which is unrealistic for
practical real-time systems.  Towards implementability, robustness has
been considered in the literature.  A parametric semantics that allows
some perturbation of delays (up to a parameter bound) was introduced
in~\cite{Puri-deds00}.  The early works consider a robust semantics
for timed automata, in which perturbations are controlled by an
opponent, while the controller has a reachability
objective~\cite{Sankur2013, Wulf2004,bouyer2015robust,
bouyer2013robust}.  It thus amounts to deciding whether there exists a
robust execution reaching a target set of locations.  To solve this
problem, the data structure of Difference Bound Matrix (DBM), which is
standard in algorithms for timed automata, has been refined to shrunk
DBMs~\cite{Sankur2014}.  While several variants of robustness have
been considered in the literature, we will focus on the conservative
semantics~\cite{Sankur2013,Oualhadj2014, BusattoGaston2019}.  There,
the delays proposed by the controller before firing a transition must
remain feasible under any perturbation (within the allowed
perturbation bound) proposed by its opponent.

In the context of weighted timed games, the perturbations are naturally in the hands of the opponent \MaxPl,
who is antagonistic to \MinPl.
The robust value problem asks whether \MinPl can guarantee reaching the target set of locations
with the cumulative weight bounded by a threshold.
Robust strategies must therefore anticipate any delay perturbation \MaxPl could make.
The robust value is then defined for infinitesimal perturbations,
as in the timed automata’s case~\cite{Puri-deds00, Sankur2013}.
The decidability landscape for the robust value problem is currently as follows.
Perhaps unsurprisingly, like for non-robust WTGs, it was shown a decade ago that
the robust value problem is undecidable in general~\cite{bouyer2015robust, Guha2015}.
More recently, classes of WTGs for which the robust value problem is decidable have been identified,
notably divergent WTGs (introduced in~\cite{bouyer2015value, BusattoGaston2017})
and acyclic WTGs~\cite{monmege2024synthesis}.
For acyclic WTGs, the decidability proof can be seen as an involved parametric extension of the fixpoint
computation of the (exact) value~\cite{Alur2004}.
For divergent WTGs, a finite unfolding of the game is performed to return to acyclic WTGs.
In the special case of 1-clock WTGs, to our knowledge, the only decidability result concerns
restricted 1-clock WTGs in which all weights are non-negative, and with an assumption about the cycles~\cite{Guha2015}.
The proof technique generalises the construction for 1-clock WTGs under the classical (i.e.\
non-robust) semantics~\cite{bouyer2013robust}.

\paragraph*{Contribution}
In this paper, we establish the decidability of the robust value
problem for a wide class of 1-clock WTGs with arbitrary integer
weights. The only assumptions we make are common when studying robust
WTGs, see~\cite{monmege2024synthesis} for instance: first the absence
of states with exact value $-\infty$, and second the clock valuation
is bounded by a maximal constant $\clockbound$.

To prove the decidability of the robust value problem for 1-clock
WTGs, we transform the original WTG $\game$ into another WTG $\C$
formed of several copies of $\game$. Intuitively, a copy indexed by
some integer $a$ represents that the clock valuation is bounded below
by $a$. This accounts for the fact that \MaxPl can use delay
perturbations to enforce lower-bounds on the clock valuation (until
the potential next reset). We prove that exact value in the copy game
$\C$ coincides with the robust value in the original game
$\game$. Since the value problem for 1-clock WTGs is
decidable~\cite{monmege2025decidability}, we obtain the desired
decidability result. A fundamental property of the copy game is that
it exhibits no deadlocks caused by robustness. Therefore the exact
value and robust value coincide in $\C$. A key in our development is
thus to identify deadlocks that are not present in the classical exact
semantics and appear in the robust semantics.

%% file: prelim.tex
We consider weighted timed games with a single \emph{clock}, denoted by $\clockx$.
The \emph{valuation} of $\clockx$ is a non-negative real number $\val \in \Rplus$.
For a valuation $\val$, a delay $t \in \Rplus$ and a set $Y \subseteq \{x\}$, we uniformly 
write $\val[Y \coloneqq 0]$ for the valuation obtained by a reset of $Y$, that is
$\val[Y \coloneqq 0] = \val$ if $Y = \emptyset$ and $\val[Y \coloneqq 0] = 0$ if $Y = \{x\}$.
A \emph{guard} is denoted by an interval $I \subseteq \reals^+$ of non-negative real values 
with bounds in $\nats \cup \{+\infty\}$.
It constrains the valuations of the clock:
a valuation $\val$ satisfies a guard $\guard$, denoted $\nu \models \guard$, if  $\val \in \guard$.
Given a guard $\guard$, we denote by $\closureg{\guard}$ (resp. $\leftg{\guard}$, and $\rightg{\guard}$)
the closure (resp. lower, and upper bounds) of $\guard$.
We let $\Guards$ denote the set of guards over the clock $\clockx$.

We now introduce the central notion of this paper, namely weighted timed games. 
\begin{definition}
  A \emph{weighted timed game} (\WTG) is a tuple $\game=\gameEx$ where
  $\Locs=\LocsMin\uplus\LocsMax\uplus\LocsT$ is a finite set of
  \emph{locations}, partitioned into \MinPl locations, \MaxPl
  locations, and \emph{target} locations,
  $\Trans\subseteq \Locs\times\Guards\times \{\emptyset, \{x\}\}
  \times \Locs$ is a finite set of \emph{transitions}, and
  $\weight\colon \Trans\uplus\Locs \to \Z$ is a
  \emph{weight~function}.
\end{definition}
An example of \WTG is depicted in Figure~\ref{fig:rob-valIt_ex-3}.

\begin{figure}[htbp]
	\centering
	\begin{tikzpicture}

		\node[min, label={left:$\loc_0$}] (l0) {$0$};
		\node[max, right=3.75cm of l0, label={above:$\loc_2$}] (l2) {$0$};
		\node[min, below=1cm of l2, label={below:$\loc_1$}] (l1) {$1$};
		\node[min, right=8cm of l0, label={[xshift=-.6cm, yshift=-0.3cm]$\loc_3$}] (l3) {$1$};
		\node[target, right=3cm of l3] (target) {\smiley};

		\node[above=1.25cm of l0] (p1) {};
		\node[right=8cm of p1] (p2) {};

		\draw[trans,Brown] (l0) |- node[guard,xshift=2.25cm] {$\delta^2 \colon \clockx \leq 2$} node[weightlabel, below,xshift=2.25cm] {1} (l1);
		\draw[trans,RoyalBlue] (l0) to[bend right=20] node[guard] {$\delta^0 \colon \clockx \leq 1$} node[weightlabel, below] {1} (l2);

		\draw (l0) -- (p1);
		\draw (l0) |- (p2);
		\draw[trans,BrickRed] (p1)  node[guard,xshift=4.5cm] {$\delta^3 \colon \clockx \geq 2$} node[weightlabel,xshift=4.5cm] {1} -| (l3);

		\draw[trans,OrangeRed] (l1) node[guard,xshift=2.25cm] {$\delta^4 \colon \clockx \leq 1$} -| (l3);

		\draw[trans,Purple] (l2) to[bend right=20] node[guard] {$\delta^1 \colon \clockx \leq 1$} node[weightlabel, below] {1} (l0);

		\draw[trans,ForestGreen] (l2) to node[guard] {$\delta^5 \colon \clockx \leq 2$} node[reset]  {$\clockx \coloneqq 0$} (l3);

		\draw[trans,CadetBlue] (l3) to[bend right=20] node[guard] {$\delta^6 \colon \clockx \leq 2$} node[weightlabel, below] {1} (target);

		\draw[trans] (l3) to[bend left=20] node[guard] {$\delta^7 \colon \clockx \leq 1$} (target);

	\end{tikzpicture}
	\caption{A \WTG where circle locations are controlled by
          \MinPl, the square location is controlled by \MaxPl and the
          target set is $T = \{\smiley\}$. Interval guards are denoted
          by constraints, for instance $x \geq 2$ represents
          $[2,+\infty]$. A reset of $Y = \{x\}$ is indicated by
          $x \coloneqq 0$, and when $Y = \emptyset$ we ommit the reset
          on the transition.  The weight function is indicated by
          labels on locations and transitions, \emph{e.g.}
          $\weight(\ell_3) =1$, $\weight(\delta^2) =1$, and omitted
          weights on transitions are $0$, \emph{e.g.}
          $\weight(\delta^4) =0$.}
	\label{fig:rob-valIt_ex-3}
\end{figure}
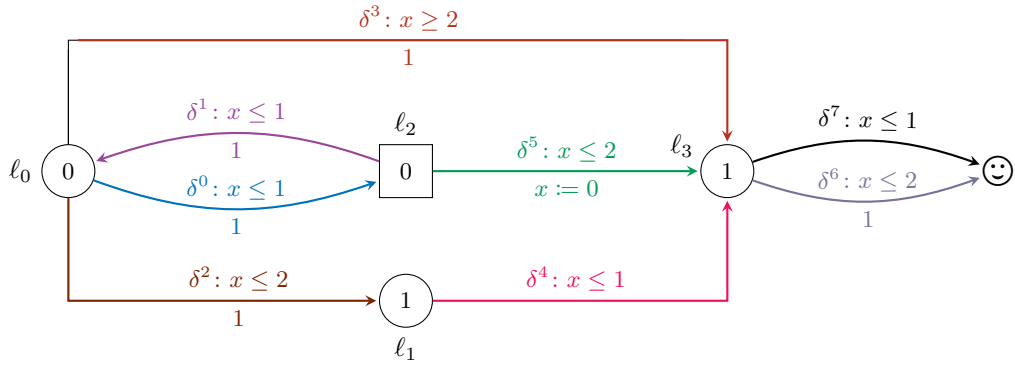
Before formalizing definitions, let us first give intuitions on how
weighted timed games are played.  We consider exact and robust
semantics for weighted timed games, each given by a 2-player game
played on an infinite transition system.  In these games, each player
chooses delays and transitions from configurations (i.e.\ a pair from
a location and a valuation) associated with their respective
locations.  Of particular interest to us is the robust semantics,
which intuitively, in contrast to the exact semantics, allows player
\MaxPl to perturb the delays chosen by \MinPl.  More precisely, in the
$\perturbation$-robust semantics, for a fixed parameter
$\perturbation >0$, \MaxPl will be allowed to perturb \MinPl's delays
by at most $\perturbation$.  We adopt here the conservative notion of
robustness (see~\emph{e.g.}~\cite{Oualhadj2014}) in which the delays
proposed by \MinPl must remain feasible after \MaxPl has applied any
possible perturbation. The exact semantics can be seen as the
instantiation of the $p$-robust semantics with $\perturbation=0$, and
for $\perturbation > 0$, the $\perturbation$-robust conservative
semantics may only reduce the choices of \MinPl.

To define the exact and robust semantics, we introduce the following notations.
\emph{Configurations} of $\game$ are pairs $(\loc,\val) \in \Locs \times \Rplus$
formed of a location and a valuation.
$\Conf = \ConfMin \uplus \ConfMax \uplus \ConfT$ denotes the set of all configurations, 
partitioned into configurations belonging to $\MinPl$, $\MaxPl$
and final configurations, according to the type of locations.

\begin{definition}
	\label{def:sem-squelette}
	A \WTG $\game = \gameEx$ and $p\geq 0$ induce the \emph{$p$-robust conservative semantics}
	(for $p >0$) and \emph{exact semantics} (for $p=0$)
        $\sem[\perturbation]{\game} = \tuple{\States, \Edges,
          \weight}$ where
\begin{itemize}
	\item $\States = \StatesMin \uplus \StatesMax \uplus \StatesT$ is the
        set of \emph{states} with $\StatesMin = \ConfMin$,
        $\StatesMax = \ConfMax \cup (\ConfMin \times \Trans \times \Rplus )$ and
		$\StatesT = \ConfT$;
		\item 
        $\Edges = \EdgesMin \uplus \EdgesMax \uplus \EdgesRob$ is the set
        of \emph{edges} with
	\begin{align*}
		\EdgesMax &=
		\big\{\big((\loc, \val) \moveto{\move}
		(\loc', (\val + \delay)[\reset \coloneqq 0] )\big) ~|~ \loc \in \LocsMax, 
		\trans = (\loc, \guard, \reset, \loc') \in \Trans, 	
		\val + \delay \models \guard
		\big\} \\
		\EdgesMin &= \big\{\big((\loc, \val) \moveto{\move}
		((\loc, \val), \move)\big) ~|~ \loc \in \LocsMin,
		\trans = (\loc, \guard, \reset, \loc') \in \Trans,
		[\val + \delay, \val + \delay + \perturbation] \subseteq I
		\big\} \\
		\EdgesRob &= \big\{\big(((\loc, \val), \move)
		\moveto{\trans, \varepsilon}
		(\loc', (\val + t +\varepsilon)[\reset \coloneqq 0])\big) ~|~
		\trans = (\loc, \guard, \reset, \loc') \in \Trans, 
		\varepsilon \in [0, \perturbation]
		\big\}
	\end{align*}
	\item 
        $\weight \colon \States \cup \Edges \to \Z$ the weight
        function such that for all states $\states \in \States$ with
        $\states = (\loc, \val)$ or $\states = ((\loc, \val), \move)$,
        $\weight(\states) = \weight(\loc)$, and all edges
        $\edge \in \Edges$, $\weight(\edge) = \weight(\trans)$ if
        $\edge = (\states \moveto{\move} \states')$ with
        $\states \in \Conf$, or $\weight(\edge) = 0$, otherwise.
\end{itemize}
\end{definition}

\begin{remark}
  In~\cite{bouyer2013robust,bouyer2015robust, Guha2015}, the set of
  possible perturbations for \MaxPl forms the symmetric interval 
  $[-\perturbation', \perturbation']$.
  For technical reasons, we use here an equivalent formalisation with an
  increase of at most $\perturbation =2 \perturbation'$ of the delays
  proposed by $\MinPl$.
\end{remark}

Note that states in $\StatesMax$ are of two forms: either a configuration of \MaxPl 
in which they choose a delay and an action, or a state in $\ConfMin \times \Trans \times \Rplus$, 
from which \MaxPl only chooses the delay perturbation via an edge in  $\EdgesRob$.
The latter will be referred to as \emph{virtual states}, as they do not 
correspond to configurations of $\game$; all other states are \emph{real states}.

In the special case of $p=0$, the infinite transition system $\sem[0]{\game}$
describes the exact semantics of the game.
Perturbation edges in $\EdgesRob$ are trivial: 
each step of $\MinPl$ with its choice in $\EdgesMin$ of delay and transition $(\delta,t)$,
is followed by a dummy edge $((\loc, \val), \delta, t) \xrightarrow{\delta, 0} (\loc', \val') \in \EdgesRob$
corresponding to $\MaxPl$ having no possibility to perturb the chosen delay $t$.

For $s$ a state of $\sem[\perturbation]{\game}$, we write $\Edges(s)$
for the set of \emph{enabled edges} in $\sem[\perturbation]{\game}$
from $s$, that is, the edges with source state $s$.  A state $s$ is a
\emph{deadlock} when $\Edges(s) = \emptyset$.  We extend the notation
$\Edges$ for the set of \emph{enabled transitions} and terminology of
deadlock to locations in~$\game$.

Notably, the $\perturbation$-robust semantics may turn real states of \MinPl into 
deadlocks even if the exact semantics contains an enabled edge.
Thus, unlike in the literature~\cite{Alur2004}, we allow state \emph{and} location deadlocks.

\smallskip
\noindent\textbf{Reaching as cheap as possible.} In the game $\game$
under $\perturbation$-robust (or exact) semantics, the objective of \MinPl is to
reach a target configuration, while minimising the cumulated weight of the
finite play built with \MaxPl.

A \emph{finite play} of $\game$ w.r.t. the $\perturbation$-robust semantics
is a sequence of edges in the transition system $\sem[\perturbation]{\game}$ 
that starts in a real state of $\game$.
Said differently, finite plays cannot start in a virtual state.
We denote by $|\play|$ the number of edges from $\EdgesMax \cup \EdgesMin$ of $\play$
(i.e. it is not exactly the length of the play), and
by $\last(\play)$ its last state
A \emph{maximal play} is a maximal sequence of consecutive edges: 
it is either a finite play reaching a deadlock state of $\game$,
or an infinite sequence such that all its prefixes are finite plays.
We use $\FPlaysG[\perturbation]$ to denote the set of finite plays on
$\sem[\perturbation]{\game}$, and write $\FPlaysGMin[\perturbation]$
(resp.~$\FPlaysGMax[\perturbation]$) for the subset of finite
plays ending in a state of \MinPl (resp.~\MaxPl).

To deal with the objective of \MinPl, we associate to every finite play
$\play= \states_0 \moveto{\move[0]} \states_1
\moveto{\move[1]} \cdots \states_k$
its \emph{cumulated weight}, taking into account both discrete and continuous weights:
$\weightC(\play)=\sum_{i=0}^{k-1} [t_i\times \weight(\states_i) + \weight(\trans_i)]$.
For a maximal play $\play$, its \emph{weight} $\weight(\play)$ is set to $\weightC(\play)$
if it ends in a target real state  $(\loc,\val)$
with $\loc\in \LocsT$, and  $\weight(\play) = +\infty$ if \play does not reach $\LocsT$
(either \play is infinite or it ends in a deadlock state that is not a target state).

A \emph{path} is a finite or infinite sequence $\ppath$ of transitions 
of~$\game$ and we write $\FPaths_\game$ for the set of finite paths.  
A \emph{target path} is a finite path ending in the target set $\LocsT$.
We denote by $\ppath_1 \xrightarrow{\trans} \ppath_2$ the concatenation
of two paths $\ppath_1 \xrightarrow{\trans}$ and $\ppath_2$, and by
$|\ppath|$ the length of $\ppath$.
Note that each play~$\play$ of $\sem[\perturbation]{\game}$
projects to a unique path $\ppath$ (by keeping only the transitions
of edges from real states of $\game$):
we say that $\play$ \emph{follows} path $\ppath$.

A \emph{strategy} for \MinPl (resp.~\MaxPl) is a mapping from finite plays
ending in a state of \MinPl (resp.~\MaxPl) to a decision in $(\move)$
labelling an edge of $\sem[\perturbation]{\game}$ from the last state
of the~play.
As argued earlier, states of \MinPl may contain deadlock states, so that
we consider strategies of \MinPl to be partial mappings.
For instance, in the \WTG depicted in \figurename{~\ref{fig:rob-valIt_ex-3}} and a
perturbation $\perturbation$, a strategy for \MinPl in all plays
ending in $(\loc_3, \val)$ can be defined only when
$\val(x) \leq 2 - \perturbation$ since, otherwise, there are no
outgoing edges in $\sem[\perturbation]{\game}$ from this state.
Symmetrically, we require  \MaxPl to always propose a move unless the
play is in a deadlock state - since otherwise \MaxPl can create deadlock by
not playing.
More formally, a strategy for \MinPl,
denoted $\robminstrategy$, is a (possibly partial) mapping
$\robminstrategy \colon \FPlaysMin[\perturbation] \to \Edges$ such
that $\robminstrategy(\play) \in \Edges(\last(\play))$.
A strategy for \MaxPl, denoted $\robmaxstrategy$, is a (possibly partial) mapping
$\robmaxstrategy \colon \FPlaysMax[\perturbation] \to \Edges$ such
that for all $\play$, if $\Edges(\last(\play))\neq \emptyset$, then
$\robminstrategy(\play)$ is defined, and in this case $\robminstrategy(\play) 
\in \Edges(\last(\play))$.
The set of strategies of \MinPl (resp.~\MaxPl) with the perturbation $\perturbation$
is denoted by $\rStratMin[\game]$ (resp.~$\rStratMax[\game]$).
If the game is clear from the context, we remove $\game$ from (all) the notations.

A play or finite play $\play= \states_0 \moveto{\move[0]} \states_1
\moveto{\move[1]} \cdots$ \emph{conforms} to a strategy $\robminstrategy$ 
of $\MinPl$ if for all $k$ such that $\states_k$ belongs to $\MinPl$, we have that
$(\move[k]) = \robminstrategy(\states_0 \moveto{\move[0]} \cdots \states_k)$.
Similarly, we define when plays conform to strategies $\robmaxstrategy$ of \MaxPl.
For every pair of strategies $(\robminstrategy, \robmaxstrategy)$ of players \MinPl and \MaxPl, 
and for every real state~$(\loc_0,\val_0)$, we let
$\outcomes((\loc_0,\val_0),\robminstrategy,\robmaxstrategy)$ be
the \emph{outcome} of $\robminstrategy$ and $\robmaxstrategy$ from $(\loc_0,\val_0)$, 
defined as the unique maximal play that starts in $(\loc_0,\val_0)$  
and conforms to $\robminstrategy$ and $\robmaxstrategy$. 

\smallskip
\noindent\textbf{Robust value problem.}
We are interested in the minimal cumulated weight that \MinPl can
guarantee while reaching the target, whatever the choices of actions
and perturbations \MaxPl makes.  To formalize this, we introduce
\emph{robust values}.  For $p > 0$, $(\loc, \val)$ a real state of
$\game$ and $\robmaxstrategy \in \rStratMax$ a fixed strategy for
\MaxPl, we let
$\rValue[\perturbation,\robmaxstrategy]_\game(\loc,\val) =
\inf_{\robminstrategy \in \rStratMin} \weight(\outcomes((\loc, \val),
\robminstrategy, \robmaxstrategy))$ be the minimal cumulated weight
\MinPl can ensure assuming \MaxPl plays $\robmaxstrategy$. Then the
\emph{$\perturbation$-robust \MaxPl value} is
$\rlowerValue_\game(\loc, \val) = \sup_{\robmaxstrategy \in
  \rStratMax}
\rValue[\perturbation,\robmaxstrategy]_\game(\loc,\val)$.  (Note
however that this value is not defined for virtual states of the
robust semantics.)  Since the $\perturbation$-robust conservative
semantics defines a \emph{quantitative reachability
  game}~\cite{BrihayeGeeraertHaddadLefaucheuxMonmege-15},
the general determinacy result of~\cite[Theorem 2.2]{Brihaye2021}
applies and the \emph{$\perturbation$-robust value} can be defined as
follows:
$\rValue[\perturbation]_\game(\loc,\val) = \rlowerValue_\game(\loc,
\val) = ~\inf_{\robminstrategy \in \rStratMin}~ \sup_{\robmaxstrategy
  \in \rStratMax} \weight(\outcomes((\loc, \val), \robminstrategy,
\robmaxstrategy))$.

For $\perturbation = 0$, $\dValue$ coincides with the \emph{(exact) value} at the
core of the well-studied value problem for
\WTG{s}~\cite{bouyer2006almost,bouyer2015value,monmege2025decidability,GO-concur24,GOV-fossacs26}. 
We sometimes write $\Value$ for $\rValue[0]$ for simplicity. 
Moreover, the values enjoy the following monotony property:
\begin{lemma}[\cite{monmege2024synthesis}]
	\label{lem:rVal-monotony}
	Let $\game$ be a \WTG, and  $\perturbation > \perturbation' \geq 0$ be two perturbations.
	Then, for every real state $(\loc, \val)$,
	$\rValue[\perturbation](\loc, \val) \geq \rValue[\perturbation'](\loc, \val)$.
\end{lemma}

In this paper, we are interested in the value that \MinPl can
guarantee if \MaxPl is restricted to arbitrarily small
perturbations.
Rather than a fixed perturbation $p>0$, the
\emph{robust value problem} is defined for infinitesimal
perturbations.
Given a \WTG \game, the \emph{robust value} is defined, for every
real state $(\loc, \val)$ of $\game$, by
$\rValue[0^+](\loc, \val) = \lim_{\perturbation \to 0, \perturbation > 0}
\rValue(\loc, \val)$. 
The robust value function is defined as the
limit of the $\perturbation$-robust value functions, and it is
well-defined since it is the limit of a non-increasing sequence of
functions (see Lemma~\ref{lem:rVal-monotony}).

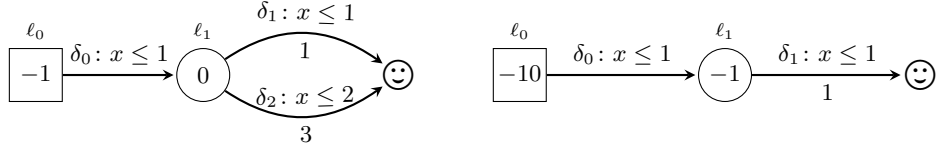
\begin{figure}
	\centering
	\begin{tikzpicture}[xscale=.8]
		\node[max, label={above:{\scriptsize$\loc_{0}$}}] (l0) {\small$-1$};
		\node[min, right=1.5cm of l0, label={above:{\scriptsize$\loc_1$}}] (l1) {\small$0$};
		\node[target, right=2cm of l1] (target) {$\smiley$};

		\draw[trans] (l0) to node[guard] {$\trans_0 \colon x \leq 1$}  (l1);
		\draw[trans] (l1) to[bend left] node[guard] {$\trans_1 \colon x \leq 1$} node[weightlabel] {$1$} (target);
		\draw[trans] (l1) to[bend right] node[guard] {$\trans_2 \colon x \leq 2$} node[weightlabel] {$3$} (target);

		\begin{scope}[xshift=8cm]
			\node[max, label={above:{\scriptsize$\loc_{0}$}}] (l0) {$-10$};
			\node[min, right=2cm of l0, label={above:{\scriptsize$\loc_{1}$}}] (l1) {$-1$};
			\node[target, right=2cm of l1] (target) {\smiley};

			\draw[trans] (l0) to node[guard] {$\trans_0 \colon x \leq 1$} (l1);
			\draw[trans] (l1) to node[guard] {$\trans_1 \colon x \leq 1$} node[weightlabel] {1} (target);

		\end{scope}

	\end{tikzpicture}
	\caption{Left: a WTG where \MaxPl can increase the value by  perturbing the delays of \MinPl.	
	Right: a WTG in which the exact and robust values coincide.
	}
	\label{fig:ex-values}
\end{figure}

\begin{example}
	\label{ex:value}
	To illustrate the notions of values, consider the \WTG depicted in~\cref{fig:ex-values} (left).
	From the initial real state $(\ell_0,0)$, this is a typical example
	where the exact and robust values do not coincide. %
	For the exact value, we observe that \MaxPl has no interest in staying in location $\loc_0$.
	Indeed, whatever the delay chosen by \MaxPl in $\ell_0$, \MinPl
	will always choose $\trans_1$ from $\loc_1$, and waiting in $\loc_0$ decreases 
	the accumulated weight, which is profitable to \MinPl.
	Thus $\Value(\loc_0,0) = 1$. %
	For the robust value however, intuitively \MaxPl can force \MinPl to take
	$\trans_2$ by reaching $\loc_1$ with clock value $1$.
	If we fix $\perturbation>0$, from $(\loc_1,1)$, \MinPl cannot choose $\trans_1$
	even with $0$ delay, as $[1,1+\perturbation]$ is not included in where the guard of $\trans_1$.
	Spending $1$ time unit in $\loc_0$, \MaxPl ensures $\rValue[\perturbation](\loc_0,0) \geq 2$. 
	In fact $\rValue[\perturbation](\loc_0,0) =2 + \perturbation$ and thus $\rValue[0+](\loc_0,0) =2$.

	For the WTG in~\cref{fig:ex-values} (right), the gap between
        the exact and robust values is even infinite!  The exact value
        from $(\loc_0,0)$ corresponds to the behaviour where \MaxPl
        leaves $\loc_0$ immediately, and \MinPl stays 1 time unit in
        $\loc_1$.  Thus $\Value(\loc_0,0) =-1$.  Now, fixing a
        perturbation $\perturbation >0$, it is in \MaxPl's interest to
        reach $\loc_1$ with clock value $1-\perturbation+\varepsilon$,
        for any $\varepsilon \in ]0,p]$.  Doing so, \MaxPl blocks
        \MinPl in $\loc_1$: indeed, \MinPl cannot propose any delay,
        as even for delay $0$, the interval
        $[1 - \perturbation + \varepsilon, 1+ \varepsilon]$ is not
        included in the guard of the unique outgoing transition
        $\trans_1$.
	Therefore, \MaxPl can prevent \MinPl to reach the target and $\rValue[0^+](\ell_0,0) = +\infty$.
\end{example}

The decision problem associated with the robust value is called the
\emph{robust value problem}.  It is naturally defined as follows:
given a \WTG $\game$, an initial location $\loc_0$, and a threshold
$\lambda \in \Q$, does
$\rValue[0^+](\loc_0, 0) \leq \lambda$?

Unsurprisingly, this problem is undecidable when generalized to
\WTG{s} with multiple clocks~\cite[Theorem 4]{bouyer2013robust}.  Some
recent results~\cite{monmege2024synthesis} prove the decidability for
restricted classes of \WTG{s}, for instance \WTG{s} that are
acyclic~\cite{Alur2004} or
divergent~\cite{bouyer2006almost,BusattoGaston2017}.  We focus here on
another class of \WTG{s} for which the exact value problem is known to
be decidable~\cite{monmege2025decidability}, namely one-clock
\WTG{s}. As in~\cite{monmege2024synthesis}, we consider WTGs in which
no state has (exact) value $-\infty$, and such that the clock
valuation is bounded by a maximal constant $\clockbound \in
\nats$. Under these two common assumptions, our main contribution is
the following:
\begin{theorem}
	\label{thm:main}
	The robust value problem is decidable for one-clock \WTG{s}.
      \end{theorem}

The high-level proof idea to establish this decidability result is as follows.
\begin{enumerate} 
\item Since we saw in~\cref{ex:value} that robustness may introduce
  deadlocks, our first step is to discriminate paths\footnote{Recall
    that paths are sequence of transitions in the game description
    $\game$.}  that are feasible or unfeasible
	in the
  $\perturbation$-robust semantics.  Typically, the path
  $\loc_0 \xrightarrow{\trans_0} \loc_1 \xrightarrow{\trans_1}$ on the
  WTG right of~\cref{fig:ex-values} is $\perturbation$-unfeasible, for
  any $\perturbation >0$.
\item Building on the characterization of $r$-feasible paths (paths
  that are $\perturbation$-feasible for small values of
  $\perturbation$), from a game $\game$, we construct another 1-clock
  \WTG $\game'$ with same robust value, and in which every path is
  either $r$-feasible or unfeasible (that is, $0$-unfeasible).  A
  naive construction of $\game'$ that would keep only $r$-feasible
  paths is however not correct.  To illustrate this, consider the \WTG
  left of~\cref{fig:naive-copies}.  This example only differs from the
  one left of~\cref{fig:ex-values} in the weights in $\loc_0$ and on
  $\trans_2$.  However, because of the smaller weight $-3$ in
  $\loc_0$, it is not profitable to \MaxPl to block $\trans_1$, and
  one can show that $\rValue[0^+](\loc_0,0) = \Value(\loc_0,0) = 1$.
  Therefore, simply removing $\trans_1$ would not preserve the robust
  value: it would increase to $\rValue[0^+](\loc_0,0) = 2$.  Instead,
  in our construction, $\game'$ consists in several copies of $\game$
  indexed by integer clock values (up to the maximal constant) and
  \MaxPl can choose to jump to a copy with larger index, say
  $i \in \nats$, to mimick the fact that they want to use robustness
  and block transitions with guards $I$ such that $\rightg{I} \leq i$.
  We call it the \emph{copy game}. %
  \cref{fig:naive-copies} illustrates the construction on the
  above-mentioned \WTG. From $(\loc_0,0)$ \MaxPl may choose to jump to
  copy~$1$ to prevent \MinPl to play $\trans_1$, but they may as well
  stay in copy~$0$.
\item Finally, in the copy game, we prove that the robust value
  coincides with the exact value.  We therefore provided a polynomial
  reduction from the robust value problem in $\game$ to the exact
  value problem in $\game'$, which is decidable.
\end{enumerate}

\begin{figure}
	\centering
	\begin{tikzpicture}[xscale=.8]
		\node[max, label={above:{\scriptsize$\loc_{0}$}}] (l0) {\small$-3$};
		\node[min, right=1.5cm of l0, label={above:{\scriptsize$\loc_1$}}] (l1) {\small$0$};
		\node[target, right=2cm of l1] (target) {$\smiley$};

		\draw[trans] (l0) to node[guard] {$\trans_0 \colon x \leq 1$}  (l1);
		\draw[trans] (l1) to[bend left] node[guard] {$\trans_1 \colon x \leq 1$} node[weightlabel] {$1$} (target);
		\draw[trans] (l1) to[bend right] node[guard] {$\trans_2 \colon x \leq 2$} node[weightlabel] {$2$} (target);

 		\begin{scope}[xshift=9cm]
			\node[max, label={above:{\scriptsize$\loc_{0}$}}] (l0) {\small$-3$};
			\node[min, right=1.5cm of l0, label={above:{\scriptsize$\loc_1$}}] (l1) {\small$0$};
			\node[min, above=1cm of l1, label={above:{\scriptsize$\loc_1^1$}}] (l2) {\small$0$};
			\node[target, right=2cm of l1] (target) {$\smiley$};
			\node[target, above=1.25cm of target] (target2) {$\smiley$};

			\draw[trans] (l0) to node[guard, below] {$\trans_0 \colon x \leq 1$}  (l1);
			\draw[trans] (l0) to node[guard] {$\trans_0^j \colon x = 1$}  (l2);
			\draw[trans] (l1) to[bend left] node[guard] {$\trans_1 \colon x \leq 1$} node[weightlabel] {$1$} (target);
			\draw[trans] (l1) to[bend right] node[guard] {$\trans_2 \colon x \leq 2$} node[weightlabel] {$2$} (target);
			\draw[trans] (l2) to node[guard] {$\trans_2^1 \colon x \leq 2$} node[weightlabel] {$2$} (target2);

                        \begin{pgfonlayer}{background}  
                          \draw[dashed, rounded corners=5pt,gray]
		([shift={(-0.5,-1)}]l0.south west) rectangle
		([shift={(0.5,.9)}]target.north east);

		\draw[dashed, rounded corners=5pt,gray]
                ([shift={(-0.5,-.25)}]l2.south west) rectangle
                ([shift={(0.5,1)}]target2.north east);
	\end{pgfonlayer}
        \node[above=.35cm of target2,xshift=-.2cm] {Copy 1};
	\node[below=.5cm of target,xshift=-.2cm] {Copy 0};
		\end{scope}

	\end{tikzpicture}
	\caption{A WTG (left) where \MaxPl has no incentive to block
          \MinPl to play optimally and the corresponding copy game
          (right).
          \MaxPl has a new jumping transition $\delta_o^j$ from
          $\ell_0$ to $\ell_1^1$ they can use to prevent \MinPl to
          play $\delta_1$ in the next move: $\delta_1$ is no longer
          present in copy~1.}
	\label{fig:naive-copies}
\end{figure}
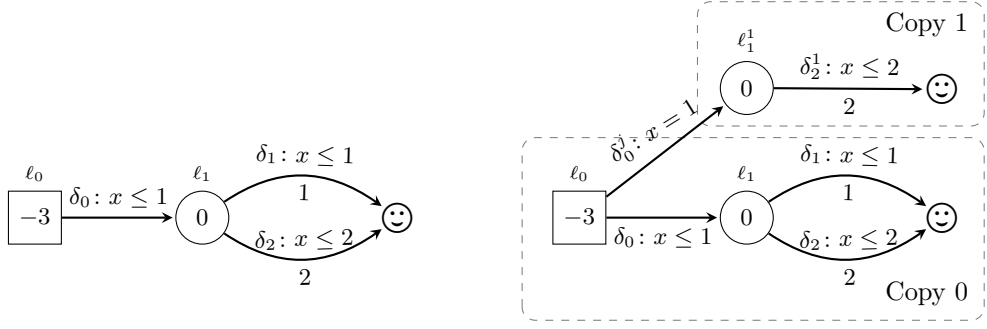

%% file: deadlocks.tex
The first step of the proof of~\cref{thm:main} consists in the
identification of paths in the $\perturbation$-robust semantics where
\MaxPl can exploit delay perturbation to create a deadlock, as
illustrated in~\cref{ex:value}. They do so by choosing a perturbed
delay such that the reached valuation is close enough to the
upper-bound of the guard of \MinPl's next transition. Our objective is
to determine when these deadlocks due to perturbations are profitable
for \MaxPl. We do so directly on the WTG structure, i.e. only by
inspecting its paths.

In the following, we fix a WTG $\game$. We assume that $\game$ does
not contain transitions with a singleton guard\footnote{A guard $I$ is \emph{punctual} when $\leftg{I} = \rightg{I}$.}
from locations of \MinPl. This is without lost of generality, since \MinPl cannot choose
a transition with a singleton guard otherwise any positive delay
perturbation of \MaxPl would enforce the guard to be unsatisfied.

We introduce the notions of \emph{feasible} and \emph{r-feasible}
paths (resp. unfeasible and r-unfeasible paths) representing the fact
that \MaxPl cannot constrain, via any choice of delay, \MinPl to be
unable to take the next transition on the path. Said differently, a
path is feasible when \MaxPl cannot avoid all plays along this
path. To properly define this notion, we define the \emph{sub-game
  induced by a path $\ppath$} by the unfolded WTG, denoted by
$\game_\ppath$, where locations are exactly the ones of $\ppath$
(locations from $\game$ may have many occurrences in $\game_\ppath$),
and its transitions are only ones of $\ppath$.

\begin{definition}
	\label{def:feasible}
	A path $\pi$ of $\game$ is \emph{r-unfeasible}
	if there exists $\perturbation \geq 0$ and
        $\robmaxstrategy \in \rStratMax[\game_\ppath]$ such that for
        every strategy $\robminstrategy \in \rStratMin[\game_\ppath]$
        of \MinPl in $\game_\ppath$, we have\footnote{We recall that,
          given a play $\play$, $|\play|$ denotes its number of edges
          except the ones of $\EdgesRob$, that is, the length of its
          underlying path.}:
        $|\rPlays(\stateI, \robminstrategy, \robmaxstrategy)| \neq
        |\ppath|$ where, $\loc_0$ is the first location of $\ppath$.
        When $\perturbation = 0$, we simply say that $\ppath$ is
        \emph{unfeasible}.
\end{definition}
For instance, the path
$\loc_0 \xrightarrow{\trans_0} \loc_1 \xrightarrow{\trans_1} \smiley$
of the WTG depicted in~\cref{fig:ex-values} (right) is feasible since
\MinPl can always reach the target. However, it is r-unfeasible
(since for every $\perturbation > 0$, the strategy of \MaxPl reaching
$(\loc_1,1)$ in the first step avoids \MinPl to take $\trans_1$ as
$\Edges(\loc_1, 1) = \emptyset$.

Although the r-feasibility of a path quantifies over all strategies of
\MinPl, fortunately, we observe that an r-unfeasible path is blocking
already in the extremal case when \MinPl's strategy is to play minimum
delays.  Thus, checking whether a path is r-unfeasible can be done
on such a strategy.  Note however that the minimum delay may not always
exist (if the guard interval is left-open), thus we define a family of
strategies for \MinPl indexed by smaller and smaller errors
$\varepsilon$.
Formally, given a path
$\ppath=\loc_0 \xrightarrow{\trans_0} \dots \xrightarrow{\trans_{k-1}}
\loc_k$, a perturbation bound $\perturbation$, and $\varepsilon > 0$,
we define an \emph{asap strategy} of \MinPl on $\game_\ppath$, denoted
by $\robminstrategy[\varepsilon]_\ppath$, such that for all plays
$\play \in \rPlaysMin$ where $\play$ follows $\ppath$ and
$|\play| < |\ppath|$:
\[
	\robminstrategy[\varepsilon]_\ppath(\play)=
	\begin{cases}
		(\trans_{|\play|}, \delay_{\min}(\val_{|\play|-1}, \guard_{|\play|})) &
		\text{if } \val_{|\play|-1} +  \delay_{\min} \models \guard_{|\play|}; \\
		(\trans_{|\play|}, \delay_{\min}(\val_{|\play|-1}, \guard_{|\play|}) + \varepsilon) & \text{otherwise};
	\end{cases}
\]
where $\guard_{|\play|}$ is the guard of the transition $\trans_{|\play|}$, and, for all 
valuations $\val$ and guards $\guard$, we define
$\delay_{\min}(\val, \guard) = \inf \{\delay \mid \val + \delay \in \closureg{\guard}\}$.
As the intuition suggests, one can show that asap strategies suffice
to check r-unfeasibility of paths (see~\cref{app:asap} for its proof).

\begin{restatable}{lemma}{asap}
	\label{lem:deadlock_asap}
	Let $(\robminstrategy[\varepsilon]_\ppath)_\varepsilon$ be a
        family of asap strategies for \MinPl for a path $\ppath$
        starting at location $\loc_0$.  Then, the two following
        statements are equivalent:
	\begin{enumerate}
		\item\label{itm:deadlock_asap-1} $\ppath$ is r-unfeasible;
		\item\label{itm:deadlock_asap-2} there exists
                  $\perturbation > 0$, a robust strategy
                  $\robmaxstrategy$ for \MaxPl in $\game_\ppath$ and
                  $\epsilon >0$ such that for all
                  $0 < \widehat{\varepsilon} \leq \varepsilon$,
                  $|\rPlay(\stateI,\robminstrategy[\widehat{\varepsilon}]_\ppath,
                  \robmaxstrategy) | < |\pi|$.
	\end{enumerate}
\end{restatable}

Since unfeasible paths can be detected by a reachability query in the
induced timed game~\cite{Jurdzinski2007}, we now have all ingredients
to characterize paths that are feasible but r-unfeasible. Indeed,
deadlocks induced by robustness are of two types:
\begin{description}
\item[type~1] \MinPl cannot guarantee a sufficient margin with the
  upper-bound of the guard $\guard$ on one of its transitions, because
  since the last reset, a guard $\guard'$ on one of its transition
  imposes a valuation greater or equal than the upper-bound of $\guard$;
\item[type~2] \MaxPl can exploit delays by pushing the clock valuation
  towards the upper-bound of the guard $\guard$ of one of its
  transitions, and there are no reset before a transition of \MinPl
  whose guard $\guard'$ has a smaller upper-bound than $\guard$.
\end{description}
We now formalize these two possible patterns for feasible but
r-unfeasible paths, using $\Paths^{Y=\emptyset}_{\MinPl}$
to denote the set of paths in $\game$ without reset
transitions ending in a state of \MinPl.

\begin{restatable}{proposition}{blocage}
	\label{prop:blocage}
	Let $\ppath$ be a feasible path in $\game$.  The path $\ppath$
        is r-unfeasible if and only if it has the following form:
        $\ppath = \pi_0 \xrightarrow{\trans_0} \pi_1
        \xrightarrow{\trans_1} \pi_2$ with
        $\pi_1 \in \Paths^{Y=\emptyset}_{\MinPl}$, and
	\begin{description}
        \item[\emph{type~1}]  $\pi_0 \in \Paths_{\MinPl}$ and
          $\leftg{\guard_0} = \rightg{\guard_1}$; or
        \item[\emph{type~2}]  $\pi_0 \in \Paths_{\MaxPl}$ and
          $\rightg{\guard_0} = \rightg{\guard_1}$;
	\end{description}
	with $\guard_i$  the guard of $\trans_i$ for $i \in \{0,1\}$.
\end{restatable}

%% file: copy.tex
The second step in the proof of~\cref{thm:main} is the construction
from $\game$ of a new WTG, the \emph{copy game}, whose exact value
coincides with the robust value in $\game$. In this WTG, \MaxPl can
either choose to play an r-unfeasible path or stick to a feasible
path. As explain in~\cref{sec:prelim}, it would be incorrect to remove
r-unfeasible paths (or parts of such paths) without giving extra power
to \MaxPl. Specifically, the copy game allows \MaxPl to emulate a
robust deadlock in the exact semantics (without impacting them when
they prefer not to use robust deadlocks). As an illustrative example,
the WTG depicted in~\cref{fig:copie} is the associated copy game of
the WTG depicted in~\cref{fig:rob-valIt_ex-3}.

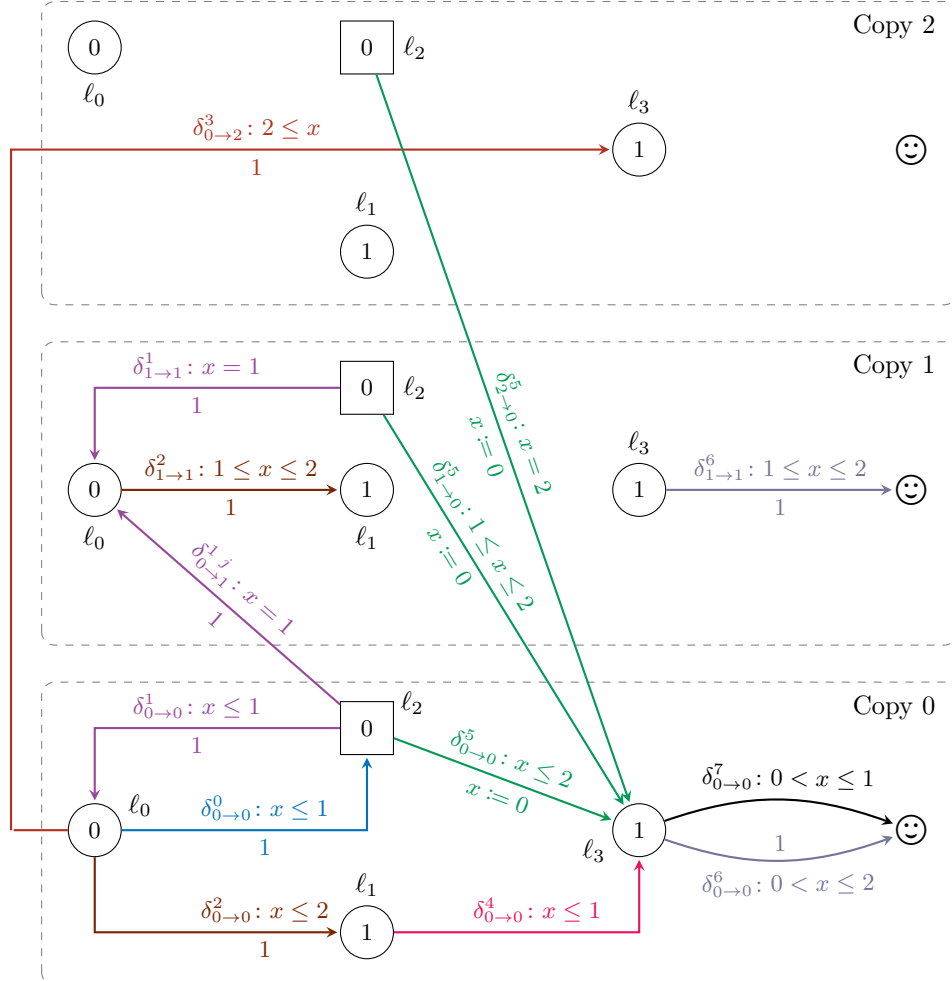
\begin{figure}[h]
  \begin{tikzpicture}[scale = .9]

		\node[min, label={[xshift=.6cm, yshift=-0.3cm]$\loc_0$}] (l00) at (0,0) {$0$};
		\node[max, label={[xshift=.6cm, yshift=-0.3cm]$\loc_2$}] (l20) at (4,1.5) {$0$};
		\node[min, label={above:$\loc_1$}] (l10) at (4,-1.5) {$1$};
		\node[min, label={[xshift=-.6cm, yshift=-0.9cm]$\loc_3$}] (l30) at (8,0) {$1$};
		\node[target] (target1) at (12,0) {\smiley};

		\node[above right] at (11,1.5) {Copy 0};
		\node[above right] at (11,6.5) {Copy 1};
		\node[above right] at (11,11.5) {Copy 2};

		\node[min, label={below:$\loc_0$}] (l01) at (0,5) {$0$};
		\node[max, label={right:$\loc_2$}] (l21) at (4,6.5) {$0$};
		\node[min, label={below:$\loc_1$}] (l11) at (4,5) {$1$};
		\node[min, label={above:$\loc_3$}] (l31) at (8,5) {$1$};
		\node[target] (target2) at (12,5) {\smiley};

		\node[min, label={below:$\loc_0$}] (l02) at (0,11.5) {$0$};
		\node[max, label={right:$\loc_2$}] (l22) at (4,11.5) {$0$} ;
		\node[min, label={above:$\loc_1$}] (l12) at (4,8.5) {$1$};
		\node[min, label={above:$\loc_3$}] (l32) at (8,10) {$1$};
		\node[target] (target3) at (12,10) {\smiley};

		\begin{pgfonlayer}{background}  
			\draw[dashed, rounded corners=5pt,gray]
			([shift={(-0.5,-2)}]l00.south west) rectangle
			([shift={(0.5,2)}]target1.north east);

			\draw[dashed, rounded corners=5pt,gray]
			([shift={(-0.5,-2)}]l01.south west) rectangle
			([shift={(0.5,2)}]target2.north east);

			\draw[dashed, rounded corners=5pt,gray]
			([shift={(-0.5,-3.5)}]l02.south west) rectangle
			([shift={(0.5,2)}]target3.north east);
		\end{pgfonlayer}

		\node[above=1.25cm of l0] (p1) {};
		\node[right=8cm of p1] (p2) {};
		\node[left=.75cm of l00,coordinate] (p3) {};

		\draw[trans, Brown] (l00) |- node[guard,xshift=2.25cm] {$\delta^2_{0\rightarrow 0} \colon \clockx \leq 2$}
		node[weightlabel, below,xshift=2.25cm] {1} (l10);
		\draw[trans, RoyalBlue] (l00) node[guard,xshift=2.25cm] {$\delta^0_{0\rightarrow 0} \colon \clockx \leq 1$}
		node[weightlabel, below,xshift=2.25cm] {1} -| (l20);

		\draw[trans, OrangeRed] (l10) node[guard,xshift=2.25cm] {$\delta^4_{0\rightarrow 0} \colon \clockx \leq 1$} -| (l30);

		\draw[trans, Purple] (l20)  node[guard,xshift=-2.25cm] {$\delta^1_{0\rightarrow 0} \colon \clockx \leq 1$}
		node[weightlabel, below,xshift=-2.25cm] {1} -| (l00);
		\draw[trans, ForestGreen] (l20) to node[guard] {$\delta^5_{0\rightarrow 0} \colon \clockx \leq 2$}
		node[reset]  {$\clockx \coloneqq 0$} (l30);

		\draw[trans, CadetBlue] (l30) to[bend right=20]
		node[guard, below, xshift=.1cm] {$\delta^6_{0 \rightarrow 0} \colon 0< \clockx \leq 2$}
		node[weightlabel, above] {1} (target1);
		\draw[trans] (l30) to[bend left=20] node[guard, xshift=.1cm]
			{$\delta^7_{0 \rightarrow 0} \colon 0 < \clockx \leq 1$} (target1);


		\draw[trans, Purple] (l20) to node[guard] {$\delta^{1 \ j}_{0\rightarrow 1} \colon x = 1$}
		node[weightlabel, below] {1} (l01);

		\draw[trans, Brown] (l01) to node[guard]  {$\delta^2_{1\rightarrow 1}\colon 1 \leq x \leq 2$}
		node[weightlabel, below] {1} (l11);

		\draw[trans, Purple] (l21) node[guard,xshift=-2.25cm] {$\delta^1_{1\rightarrow 1} \colon x = 1$}
		node[weightlabel, below,xshift=-2.25cm] {1} -| (l01);
		\draw[trans, CadetBlue] (l31) to node[guard] {$\delta^6_{1 \rightarrow 1}\colon 1 \leq x \leq 2$}
		node[weightlabel, below] {1} (target2);

		\draw[trans,ForestGreen] (l21) to node[guard, xshift=-1cm]
			{$\delta^5_{1\rightarrow 0} \colon 1 \leq x \leq 2$} node[reset, xshift=-1cm]  {$x \coloneqq 0$} (l30);

		\draw[trans, BrickRed, -] (l00) to (p3);
		\draw[trans, BrickRed] (p3) |- node[guard,xshift=3.25cm] {$\delta^3_{0\rightarrow 2}\colon 2 \leq x $}
		node[weightlabel,xshift=3.25cm] {1} (l32);

		\draw[trans,ForestGreen] (l22) to node[guard] {$\delta^5_{2\rightarrow 0}\colon x = 2$}
		node[reset]  {$x \coloneqq 0$} (l30);
	\end{tikzpicture}

	\caption{The copy game associated with the WTG depicted in~\cref{fig:rob-valIt_ex-3}.}
	\label{fig:copie}
\end{figure}

Before coming to the formal definition, let us explained intuitively how
the construction proceeds. The copy game is composed of different
copies of the WTG, one for each \emph{open region}\footnote{In
single-clock WTGs, the classical notion of regions reduces to simple
intervals on $\R$, see \cite{bouyer2006almost} for details.  Here,
	we call open region, the infinite interval $(a, +\infty)$ where $a$
	is a lower-bound of a region.}. Each original transition of \MinPl
with guard $\guard$ is duplicated in the copy~$a$ (for the interval
$(a,+\infty)$) as a transition with guard $\guard \cap [a, +\infty)$
and with destination the highest copy between $a$ and the lower-bound
of $\guard$ (only when $a$ is smaller than the upper-bound of
$\guard$). Each original transition of \MaxPl is transformed in
copy~$a$ to a transition with guard $\guard \cap [a, +\infty)$, and
also to any \emph{jumping transition}, denoted with the superscript
$j$, with destination in copy~$b$ (with $b \neq a$) with guard
$\{b\}$.  Formally, an original transition
$\trans = (\loc, \guard, \reset, \loc') \in \Trans$ can be transformed
into two different ways:
\begin{displaymath}
	\trans_{a\rightarrow b} =
	(\tuple{\loc, a}, \guard \cap [a + \infty), \reset, \tuple{\loc', b})
	\hfill
	\trans^j_{a\rightarrow b} =
	(\tuple{\loc, a}, \{b\}, \reset, \tuple{\loc', b})
	\qquad
\end{displaymath}

In the copy game, the two players have different abilities regarding
the moves between copies. On the one hand, \MinPl cannot choose the
next copy: they are constrained by the minimal open region they must
use. On the other hand, for non-resetting transitions \MaxPl may
choose which copy to move to: they may use a jumping transition
$\trans^j_{a\rightarrow b}$ for any $b$ in the guard closure.

We now formally define the copy game.
\begin{definition}
	\label{def:copy}
	For $\game$ a WTG with maximal constant $\clockbound$, its
	associated \emph{copy game} is
	$\C= \langle \Locs^\C, \Trans^\C, \wt^\C \rangle$ defined by
	\begin{itemize}
		\item $\Locs^\C = \Locs \times \llbracket 0, \maxBound\rrbracket$ ;
		\item $\Trans^\C = \bigcup_{a \in \llbracket 0, \maxBound\rrbracket} \cup_{\trans \in \Trans}
		\Trans^{\C}_a(\trans)$ where,
		for all $a \in \llbracket 0, \maxBound\rrbracket$, and all $\trans = (\loc,\guard, \reset,\loc')$,
		we fix
		\begin{displaymath}
			\Trans^{\C}_a(\trans) =
			\begin{cases}
				\{\trans_{a\rightarrow \max(a, \leftg{\guard})}\} &
				\text{if $\loc \in \LocsMin$, $\reset = \emptyset$, and $\rightg{\guard} > a$}  \\
				\{\trans_{a\rightarrow 0}\} &
				\text{if $\loc \in \LocsMin$, $\reset = \{x\}$, and $\rightg{\guard} > a$} \\
				\{\trans_{a\rightarrow a}\} \cup \{\trans_{a\rightarrow b}^j \mid b \in \overline{\guard} \cap \llbracket a + 1, \maxBound \rrbracket\} &
				\text{if $\loc \in \LocsMax$, $\reset = \emptyset$, and $\rightg{\guard} \geq a$} \\
				\{\trans_{a\rightarrow 0}\} &
				\text{if $\loc \in \LocsMax$, $\reset = \{x\}$, and $\rightg{\guard} \geq a$ ;}
			\end{cases}
		\end{displaymath}
		\item $\wt^\C(\langle \loc, a \rangle) = \wt(\loc)$, for all $\langle \loc, a \rangle \in \Locs^C$, and
		$\wt^\C(\trans_{a\rightarrow b})=\wt^\C(\trans_{a\rightarrow b}^j)=\wt(\delta)$,
		for all transitions of $\Trans^\C$.
	\end{itemize}
\end{definition}

Note that the constraint $\rightg{\guard}> a$ in the transitions from
\MinPl locations ensures that the guard is neither empty and nor
punctual. This is convenient since such transitions would not be
possible for \MinPl in the robust semantics. Consequently, the copy
game exhibits no robust deadlocks of type 1. Moreover, the number of
copies is exponential in the size of $\game$ (since $\clockbound$ is
encoded in binary). Finally, the copy index represents the infimum
valuation of plays because the lower-bound of transitions guards
between two copies is the index of the highest copy, and resetting
transitions have destination copy~$0$. These interesting properties
are summarized in the following lemma, whose proof is given
in~\cref{app:copy}.
\begin{restatable}{lemma}{copyProp}
	\label{lem:copy}
	Let $\C$ be the copy game associated with $\game$.
	Then,
	\begin{enumerate}
        \item\label{itm:copy_caracterisation} there are no feasible
          and r-unfeasible paths of type 1 in $\C$;
		\item\label{itm:copy_size} $\C$ is exponential in the size of
		$\game$;
		\item\label{itm:copy_end-val} for every finite play starting
		in the $0$-th copy and ending in a location
		$\tuple{\loc, a}$, its last valuation $\val$ satisfies
		$\val \geq a$.
	\end{enumerate}
\end{restatable}

%% file: values.tex
To conclude the proof of~\cref{thm:main}, we show the correctness of
the reduction to the copy game. More precisely, the robust value in
$\game$ is equal to the (exact) value in $\C$. This not only holds for
initial states of the form $(\ell_0,0)$ but also for any
$(\loc,\val)$, provided $\val$ is either $0$ or not integer. Formally:
\begin{theorem}
	\label{thm:proof}
	Let $(\loc, \val)$ be a real state such that $\val \notin \N^*$,
	$\rValue[0+]_\game(\loc, \val) = \Value_\C(\tuple{\loc, 0}, \val)$.
\end{theorem}

Before proceeding to the proof of~\cref{thm:proof}, let us justify how
it allows one to conclude the proof of~\cref{thm:main}. 
Let $\game$ be a WTG and $\C$ be its associated copy game. Thanks
  to~\cref{thm:proof}, solving the robust value problem in $\game$ for
  $(\ell_0,0)$ reduces to solving the exact value problem in $\C$ from
  $(\tuple{\loc_0,0},0)$ in copy~$0$.  Since computing the exact value
  in a WTG with one-clock can be done in exponential
  time~\cite{monmege2025decidability}, we derive that the robust value
  problem in $\game$ can be solved in exponential time in the size of
  $\C$. Now, $\C$ is exponential in the size of $\game$
  (by~\cref{lem:copy}), so that we conclude that the robust value
  problem is solvable in doubly exponential time in the size of
  $\game$. In particular, the robust value problem is decidable for
  1-clock WTG. This concludes the proof of~\cref{thm:main}.

\medskip

The rest of this section presents the key ingredients to
prove~\cref{thm:proof}.  The proof is in two steps.  First, we prove
that the construction of the copy game preserves the \emph{robust}
value. We do so by showing the existence of a perturbation bound such
that the fixed-perturbation robust values in the two games are
arbitrarily close. Second, we prove that in the copy game, the robust
value is equal to the (exact) value, as far as one considers for
initial valuation either $0$, or any valuation in a left-open
interval.

\paragraph*{Preservation of the robust value}
The first step of the correctness proof of the copy game consists in
proving the preservation of the robust value between the two WTGs
$\game$ and $\C$. More precisely, we prove that from a given robust
strategy of \MaxPl for one of these games, one can define an
equivalent one (i.e. with an at least as good value) in the other
game. Both transformations are technical yet rather natural.

\begin{restatable}{proposition}{preservationRval}
	\label{prop:preservation_rval}
	Let $(\loc, \val)$ be a real state.
	Then, $\rValue[0+]_{\game}(\loc, \val) = \rValue[0+]_\C(\langle\loc,0\rangle, \val)$.
\end{restatable}

The first inequality
$\rValue[0+]_{\game}(\loc, \val) \leq
\rValue[0+]_\C(\langle\loc,0\rangle, \val)$ is the most natural, since
$\C$ does not constrain the robust behaviours of \MaxPl.  In
particular, we can naturally embed a robust strategy of \MaxPl in
$\game$ into one of $\C$ by ignoring all jumping transitions.  The
plays conforming to this new strategy in $\C$ will remain in copy~$0$.

The second inequality
$\rValue[0+]_{\game}(\loc, \val) \geq
\rValue[0+]_\C(\langle\loc,0\rangle, \val)$ is more involved, since
\MaxPl may use jumping transitions in $\C$ with clock valuation that
does not satisfy the guard of the original transition in $\game$.
Indeed, by definition of $\C$, if $\guard$ is the guard in $\game$,
then the guard $\{b\}$ of the jumping transition only ensures
$b \in \overline{\guard}$, yet we may have $b \notin \guard$.  In
particular, when \MaxPl chooses such a transition in $\C$, to mimick
it in $\game$, one needs to slightly alter the delay (so that it
satisfies the guard $\guard$).  However, since the robust value
$\rValue$ is, a priori, not continuous, such a deviation must be
compensated in the rest of the play. Ultimately, the plays in $\C$ and
$\game$ should reacht the exact same valuation. To control precisely
the potential weight difference (due to a delay difference), this
delay alteration must depend on the largest location weight in
$\game$, and also on the length of the play.

The formal proofs of these two inequalities are given
in~\cref{app_preservation-rval} and~\cref{app_preservation-rval-2}.
Interestingly, the proof of the first inequality further induces that
\MaxPl can play optimally without jumping transitions in $\C$, i.e.:
\begin{restatable}{corollary}{maxWithoutJumping}
	\label{cor:val_rval-Max-without-jump}
	Let $\perturbation > 0$ and $\varepsilon > 0$.  Then, \MaxPl
        has an $\varepsilon$-optimal robust strategy in
        $\rStratMax[\C]$ such that all its conforming plays contain no
        jumping transitions.
\end{restatable}

\paragraph*{Exact and robust values coincide in the copy game}
  The next step is to prove that the robust and the exact values are
equal in $\C$.
\begin{restatable}{proposition}{preservationVal}
	\label{prop:preservation_Val}
	Let $\tuple{\loc, 0}$ be a location of the $0$-th copy of
        $\C$.  Then, for every $\val \notin \N^*$,
        $\rValue[0+]_{\C}(\tuple{\loc, 0}, \val) =
        \Value_\C(\tuple{\loc, 0}, \val)$.
      \end{restatable}

      This result constitutes the most technical part of this paper
      and forms the core of our contribution. The high-level intuition
      is that the robust and exact values coincide because the copy
      game $\C$ does not contain any path that is feasible and
      r-unfeasible. Compared to similar decidability results for the
      robust value problems, an additional layer of difficulty arises
      in our case with the copy game, and establishing the equality on
      the copy game is non trivial. One needs to massage strategies,
      to prove that optimal strategies for each of the players can be
      taken in a restricted class of \emph{well-formed}
      strategies. Roughly speaking, \MaxPl can encode deadlocks that
      are incurred by robustness into classical deadlocks, by using
      jumping transitions instead of playing very close to the upper
      bound of a guard. For \MinPl, restricting to well-formed
      strategies is less immediate and relies on the continuity of the
      classical value. Thanks to continuity, \MinPl can avoid taking
      the risk of playing close to the bound of guards without
      significantly altering the weight; and since the plays are
      finite, \MinPl can control the difference in weight between
      robust and exact plays. The rest of this section presents in
      more details the successive steps of the proof
      of~\cref{prop:preservation_Val}.

      To prove the proposition, we show the following result: there
      exists an integer $N \in \mathbb{N}$ such that for every
      $\varepsilon>0$ and every location $\loc \in \Locs$, there
      exists a perturbation bound $p>0$ such that
      $\forall i \in \mathbb{N^\star}$ if
      $\nu \in ]i,i+1-N\cdot 2p] \cup \{0\}$, then
  \begin{equation}
  \label{eq:ineq-values}
  \Value_{\C}(\langle \loc, 0 \rangle,\nu) \leq \rValue_\C(\langle
\loc, 0 \rangle,\nu) \leq \Value_{\C}(\langle \loc, 0 \rangle,\nu) +
\varepsilon\enspace.
\end{equation}

Taking the limit when $\perturbation$ tends to $0$, we obtain, for all
real states $(\tuple{\loc, 0}, \val)$ such that
$\val \notin \N^\star$, that
$\Value_\C(\tuple{\loc, 0}, \val) \leq \rValue[0+]_\C(\tuple{\loc, 0},
\val) \leq \Value_\C(\tuple{\loc, 0}, \val) + \varepsilon$.  Since
this inequality hold for all $\varepsilon$, we conclude that
$\rValue[0+]_\C(\tuple{\loc, 0}, \val) \leq \Value_\C(\tuple{\loc, 0},
\val)$.

Now, we successively prove the two inequalities
from~\cref{eq:ineq-values}. First, by
$\cite{monmege2025decidability}$, we already have that for all
$\perturbation > 0$, and $(\tuple{\loc, 0}, \val)$ real states of
$\C$,
$\rValue_\C(\tuple{\loc, 0}, \val) \geq \Value_\C(\tuple{\loc, 0},
\val)$.  Thus,
$\rValue[0+]_\C(\tuple{\loc, 0}, \val) \geq \Value_\C(\tuple{\loc, 0},
\val)$.

Conversely, we prove the following statement. For $\varepsilon >0$,
there exist an integer $N \in \N$ and a perturbation bound
$\perturbationBound > 0$ such that for all
$\perturbation \leq \perturbationBound$, and real states
$(\tuple{\loc, 0}, \val)$, such that
$\val \in (i, i + 1 - N\cdot \perturbation] \cup \{0\}$ with
$i \in \N$, we have
$\rValue_\C(\tuple{\loc, 0},\nu) \leq \Value_{\C}(\tuple{\loc,
  0},\val) + \varepsilon$.  To do so, we first fix $N$ and
$\perturbationBound$ by considering the maximal length of plays to
guarantee that \MinPl can achieve the value.  Then, we prove
in~\cref{lem:rval-eq-val_restrictstrat}
(resp.~\cref{lem:copy_Min_rstrict}) that \MaxPl (resp. \MinPl) can
play almost optimally far enough of positive integer valuations in
$\C$ (resp. $\game$) by enforcing jumping transitions where a
valuation is too close to a guard upper-bound (resp. by the continuity
of the exact value~\cite{Brihaye2021}).  Finally, in
Lemmas~\ref{lem:copy_wf-to-r} and~\ref{lem:rval-eq-val_to-exact_conf},
we embed an exact strategy of \MaxPl into a robust one ensuring a
value at least as large.

\smallskip
\noindent
\textbf{Determining the constants $N$ and $\perturbationBound$.}
Intuitively, $N$ is related to the valuations margin required by
\MinPl to guarantee that the target remains reachable under a robust
semantics, independent of the perturbation chosen by \MaxPl.  A result
from robust reachability in timed games~\cite[Proposition
3.2]{bouyer2015robust}, implies that there exists $N \in \N$ such that
for every reachable pair of location and open region, the region
contains an interval $(i, i+1 - N \cdot \perturbation[$ when
$\perturbation \leq \frac{1}{3N}$. In more
details,~\cite{bouyer2015robust} exploits the data structure of
\emph{shrunk DBM}~\cite{Sankur2014} that builds upon Difference Bound
Matrices (DBM), and which consists in a DBM constrained with a
shrinking matrix indicating how far \MinPl has to remains from border
of regions. $N$ is set as by the maximal element of this matrix and is
computable in
\EXPTIME~\cite{bouyer2015robust}\footnote{Reference~\cite{bouyer2015robust}
  considers a robustness notion, called \emph{excessive semantics},
  that slightly differs from ours, in that the guard is checked before
  the perturbation.  However, a polynomial reduction in the two
  directions between both semantics was later
  provided~\cite{monmege2024synthesis}.  Proposition 3.2
  of~\cite{bouyer2015robust} thus applies to our conservative
  semantics as well.}.

Now, the constant $\perturbationBound$ must account for the weight
alteration due to the perturbations of \MaxPl. To define
$\perturbationBound$, we use the maximal length of plays, denoted
$\kappa$.  Following~\cite{monmege2025decidability}, under the exact
semantics, there exists $\kappa$ such that for every $\epsilon>0$,
there exists an $\epsilon$-optimal strategy $\minstrategy$ for \MinPl,
for which the length of all conforming plays is bounded by $\kappa$.
Formally, it is such that for every $\maxstrategy \in \StratMax[\C]$,
for every real state $(\tuple{\loc, 0}, \val)$,
$|\Play((\tuple{\loc, 0}, \val), \minstrategy, \maxstrategy)|\leq
\kappa$. We observe that $\kappa$ also bounds the length of plays such
that \MinPl reaches the target with an optimal weight both under the
exact and robust semantics, because the (exact) value is always at
most any robust value (see~\cref{lem:rVal-monotony}).  We let
  $\perturbationBound \leq \min\{\dfrac{\varepsilon}{\maxWeightLoc
    \cdot \kappa^2},\dfrac{\varepsilon}{8 N \cdot \maxWeightLoc \cdot
    \kappa}\}$ where $\maxWeightLoc$ is the maximal weight in
  locations of $\C$, $\kappa$ the maximal length of plays, and $N$ the
  bound derived from the above-mentioned (unweighted) robust
  reachability analysis.

Finally, we fix $p \leq \perturbationBound$ and
$(\tuple{\loc, 0}, \val)$ a real state of $\C$ such that
$\val \in (i, i + 1 - N\cdot \perturbation] \cup \{0\}$.

\smallskip
\noindent
\textbf{Restricting to relevant strategies of \MaxPl.}  The copy game
$\C$ introduces \emph{jumping transitions} that give \MaxPl the
opportunity to choose a punctual guard to reach a copy with higher
index, in which the options of \MinPl can be limited compared to the
original game.  We will show that in $\C$, it is in the interest of
\MaxPl --or rather they lose close to nothing-- to choose a jumping
transition instead of staying in the same copy but close to an integer
valuation. In other words, \MaxPl can play almost optimally with a
so-called \emph{well-formed} strategy, guaranteeing that if they stay
in the same copy, then the valuation remains bounded away from
integers.  

Formally, a robust strategy of \MaxPl, $\robmaxstrategy$ is
\emph{$p$-well-formed} if for every robust play $\play$ ending in a
real state of \MaxPl, denoted $(\tuple{\loc, a}, \val)$, such that
$\robmaxstrategy(\play) = (\trans_{a \to a}, \delay)$, then, for every
$i \in \N^*$, $\val + \delay \notin F_i$ where we define
$F_i = (i- N \cdot \perturbation, i)$ as the \emph{forbidden interval}
for $i \in \N$.  \MaxPl can play optimally by picking strategies
within the class of well-formed strategies, by forcing a jumping
transition when the original valuation would fall in a forbidden
interval.  From any almost-optimal robust strategy of \MaxPl that does
not take any jumping transition
(by~\cref{cor:val_rval-Max-without-jump}), we define a new strategy
that will take a jumping transition when a forbidden interval is
reached.  The choice of the size of the interval will control the
weight alteration incurred by jumping to a higher copy, so that the
value this new strategy induces is close enough to the one of the
original strategy.  Thus, we obtain the following lemma
(see~\cref{sec:app_copy-well-formed} for its proof).

\begin{restatable}{lemma}{copyWellFormed}
	\label{lem:rval-eq-val_restrictstrat}
	\MaxPl has a $\perturbation$-well-formed $\varepsilon$-optimal strategy in $\C$.
\end{restatable}
In the sequel, we therefore consider only $\perturbation$-well-formed
robust strategy for \MaxPl in $\C$.

\smallskip
\noindent
\textbf{Restricting to relevant strategies of \MinPl.}
Against a $\perturbation$-well-formed strategy of \MaxPl, \MinPl
always has a robust decision along a feasible path in $\C$ (which is
r-feasible by~\cref{lem:copy}-\cref{itm:copy_caracterisation}).  Thus,
a crux of the proof is to show that we can restrict the choices of
\MinPl to robust ones in $\C$ without impacting too much the value.
To control the value alteration, we use the fact that the (exact)
value is continuous over regions~\cite{Brihaye2021}.

Formally, an (exact) play is \emph{$\perturbation$-well-formed} if for
every edge
$(\loc, \val) \xrightarrow{(\trans, \delay)} (\loc', \val')$ along the
play, we have $\delay = 0$ or $\val + \delay \notin F_i$ for all
$i \in \N^*$ ($F_i$ is again the forbidden set for interval
$(i,i+1)$). Similar to \MaxPl in the previous paragraph, one can prove
that \MinPl has an almost-optimal $\perturbation$-well-formed strategy
in $\C$.  The proof is given in~\cref{app:copy_e-to-wf}.

\begin{restatable}{lemma}{copyMin}
	\label{lem:copy_Min_rstrict}
	\MinPl has a $\perturbation$-well-formed $\varepsilon$-optimal
        strategy in $\C$.
\end{restatable}

\smallskip
\noindent
\textbf{Constructing a classical strategy from a robust strategy
  preserving values.}  We now have all intermediate results to
conclude the proof of~\cref{prop:preservation_Val}. Let
$\robmaxstrategy$ be a $p$-well-formed $\varepsilon$-optimal strategy for \MaxPl
as defined in~\cref{lem:rval-eq-val_restrictstrat}. The proof consists in the
definition of an exact strategy, $\maxstrategy$, from
$\robmaxstrategy$, such that for every play $\play$ conforming to
$\maxstrategy$ and starting from $(\loc, \val)$ with
$\val \notin \bigcup_{i \in \N^\star} F_i$, there exist a play
$\play^r$ conforming to $\robmaxstrategy$ such that
$\weight(\play^r) \leq \weight(\play) + \varepsilon$.  Indeed, we
obtain the following inequality:
\begin{align*}
	\rValue_\C(\loc, \nu) &= \inf_{\sigma^r \in \rStratMin[\C]} \weight(\rPlay(\loc, \val, \sigma^r, \tau^r))  \\
	&\leq \inf_{\sigma \in \StratMin[\C]} \weight(\Play(\loc, \val,\sigma, \tau)) + \varepsilon  =  \Value_\C(\loc, \nu) + \varepsilon
\end{align*}

The first step to come up with $\maxstrategy$ is the definition of a
"projection" mapping that transforms an exact play into a robust play.
The main difficulty here is the need to adapt the decisions of \MinPl
when they are not robust (i.e. when \MinPl plays "too close" to the
upper-bound of a guard). We therefore use~\cref{lem:copy_Min_rstrict} to
restrict the set of plays that we consider to plays conforming to
a $\perturbation$-well-formed strategy.  Formally, the projection
$\convert_{\wf \to \rmath}^{\robmaxstrategy} : \FPlays_{wf} \to \rFPlays$ is
such that
$\convert_{\wf \to \rmath}^{\robmaxstrategy}(\tuple{\loc, a}, \val) =
(\tuple{\loc, a}, \val)$, and, for every
$\play = \play_1 \xrightarrow{\trans, \delay} (\tuple{\loc, a},
\val)$, by letting $\val_1^c$ be the last valuation
of $\convert_{\wf \to \rmath}^{\robmaxstrategy}(\play_1)$,
according to the last configuration, we distinguish two cases:
\begin{itemize}
	\item if $\last(\play_1) \in \ConfMin$, then we fix $\convert_{\wf \to \rmath}^{\robmaxstrategy}(\play)$ to be equal to
	\begin{itemize}
		\item $\convert_{\wf \to \rmath}^{\robmaxstrategy}(\play_1) \xrightarrow{\trans, \val - \val_1^c} s
		\xrightarrow{\robmaxstrategy(\convert_{\wf \to \rmath}^{\robmaxstrategy}(\play_1)
			\xrightarrow{\trans, \val - \val_1^c} s)} s'$
		\hfill if  $\val - \val_1^c >  0$
		\item $\convert_{\wf \to \rmath}^{\robmaxstrategy}(\play_1) \xrightarrow{\trans, \delay} s
		\xrightarrow{\robmaxstrategy(\convert_{\wf \to \rmath}^{\robmaxstrategy}(\play_1)
			\xrightarrow{\trans, \delay} s)} s'$
		\hfill otherwise
	\end{itemize}
	\item otherwise $\last(\play_1) \in \ConfMax$, and we fix $\convert_{\wf \to \rmath}^{\robmaxstrategy}(\play)$ to be equal to
	\item \begin{itemize}
		  \item $\convert_{\wf \to \rmath}^{\robmaxstrategy}(\play_1)
		  \xrightarrow{\robmaxstrategy(\convert_{\wf \to \rmath}^{\robmaxstrategy}(\play_1))} s$
			\hfill  if $\robmaxstrategy(\convert_{\wf \to \rmath}^{\robmaxstrategy}(\play_1)) =
		  (\trans, \delay - \val^c_1 + \val_1)$
		  \item $\convert_{\wf \to \rmath}^{\robmaxstrategy}(\play_1) \xrightarrow{\trans, \delay} s$
		  \hfill otherwise
	\end{itemize}
\end{itemize}

We note that, given $\play$ a $\perturbation$-well-formed play,
$\convert_{\wf \to \rmath}^{\robmaxstrategy}(\play)$ is indeed a robust
play, since every decision for \MinPl (resp. \MaxPl) on $\play$ is
robust by~\cref{lem:copy_Min_rstrict}
(resp.~\cref{lem:rval-eq-val_restrictstrat}).
Moreover, some interesting properties on the projection mapping
(see~\cref{app:copy_wf-to-r} for their proof):
\begin{restatable}{lemma}{copyWFtoR}
	\label{lem:copy_wf-to-r}
	Let $\play$ be a finite $p$-well-formed play.
	Then,
	\begin{enumerate}
		\item\label{itm:copy_wf-to-r_last-loc}
		$\convert_{\wf \to \rmath}^{\robmaxstrategy}$ preserves the last
		location of $\play$. Formally, if $\tuple{\loc, a}$ is the
		last location of $\play$, then $\tuple{\loc, a}$
		is the last location of
		$\convert_{\wf \to \rmath}(\play)$;

		\item\label{itm:copy_wf-to-r_last-val} the valuation deviation of $\convert_{\wf \to \rmath}^{\robmaxstrategy}$ is controlled. Formally, if $\val$
		is the last valuation of $\play$, then the last valuation $\val^c$ of
		$\convert_{\wf \to \rmath}(\play)$ is such that
		$0 \leq \val^c - \val \leq \perturbation \cdot |\play|$;

		\item\label{itm:copy_wf-to-r_weight} the weight
		difference is controlled, i.e.
		$\weightC(\convert_{\wf \to \rmath}(\play)) \leq
		\weightC(\play) + \maxWeightLoc \cdot \perturbation \cdot |\play|^2$.
	\end{enumerate}
\end{restatable}

Using the projection mapping, we define the new strategy,
$\maxstrategy$, for \MaxPl to play in the exact semantics.
Intuitively, $\maxstrategy$ mimicks as much as possible the choices of
the robust strategy (by adapting delays to reach good valuations)
against plays conforming to $\minstrategy$.  Indeed, since
$\minstrategy$ is $\varepsilon$-optimal for \MinPl, it suffices for
\MaxPl to play almost-optimally against it.  Formally, for every
(exact) play $\play$ conforming to $\minstrategy$, letting $\val$
(resp. $\val^c$) be the last valuation of $\play$ (resp.
$\convert_{\wf \to \rmath}^{\robmaxstrategy}(\play)$), we define
$\maxstrategy$ such that
\begin{displaymath}
	\maxstrategy(\play) = (\trans_{a \to b}^{(j)}, \val^c + \delay^c - \val) \qquad \text{if }
	(\trans_{a \to b}^{(j)}, \delay^c) = \robmaxstrategy(\convert_{\wf \to \rmath}^{\robmaxstrategy}(\play))
\end{displaymath}
On other play prefixes, $\maxstrategy$ may take any enable transition.
So defined, $\maxstrategy$ is a legal strategy for \MaxPl, since the
chosen delay satisfies the guard by reaching the same valuation of
$ \robmaxstrategy$ (that is a robust strategy),
and it is non-negative, by~\cref{itm:copy_wf-to-r_last-val}
of~\cref{lem:copy_wf-to-r}.

Finally, by construction, we have the following property
(see~\cref{app:rval-eq-val_to-exact_conf} for its proof):
\begin{restatable}{lemma}{eqValConf}
	\label{lem:rval-eq-val_to-exact_conf}
	If a play $\play$ conforms to $\maxstrategy$, then
        $\convert_{\wf \to \rmath}^{\robmaxstrategy}(\play)$ conforms to
        $\robmaxstrategy$.

\end{restatable}

Given $\play$ a play conforming to $\maxstrategy$ and $\minstrategy$,
we define
$\play' = \convert_{\wf \to \rmath}^{\robmaxstrategy}(\play)$, and
this play conforms to $\robmaxstrategy$ and satisfies
$\weight(\play') \leq \weight(\play) + \varepsilon$.  Indeed,
  by~\cref{itm:copy_wf-to-r_weight} of~\cref{lem:copy_wf-to-r},
  $\weightC(\convert_{\wf \to \rmath}^{\robmaxstrategy}(\play)) \leq
  \weightC(\play) + \maxWeightLoc \cdot \perturbation \cdot |\rho|^2$,
  and from $\perturbation \leq \perturbationBound$ and
  $|\play| \leq \kappa$ we get:
\begin{align*}
  \weightC(\convert_{\wf \to \rmath}^{\robmaxstrategy}(\play)) &\leq 	\weightC(\play)  +
	\maxWeightLoc \cdot \perturbationBound \cdot \kappa^2 \\
	&\leq  	\weightC(\play) +
	\maxWeightLoc  \cdot \dfrac{\varepsilon}{\maxWeightLoc \cdot \kappa^2}  \cdot \kappa^2 \\
	&\leq \weightC(\play) + \varepsilon
\end{align*}
A consequence of~\cref{itm:copy_wf-to-r_last-loc} of~\cref{lem:copy_wf-to-r}, is that either $\play$ and
$\convert_{\wf \to \rmath}^{\robmaxstrategy}(\play)$ both reach the
target or none of them do. Thus, 
$\weight(\convert_{\wf \to
  \rmath}^{\robmaxstrategy}(\play))  \leq \weight(\play)  + \varepsilon$.  That concludes the
proof of~\cref{prop:preservation_Val}.

%% file: conclusion.tex
This paper shows the decidability of the robust value problem for
1-clock weighted timed games under mild assumptions (no $-\infty$
value and a bounded clock $\nu \in [0,\clockbound]$). We do so by
reducing the robust value problem to the exact value problem in a WTG
of exponential size (in the representation of the maximal constant
$\clockbound$). Note that our proof technique more broadly establishes
the computability of $\rValue[0^+](\loc_0, \nu)$, for any valuation $\nu
\in [0,\clockbound] \setminus \mathbb{N}^*$.

Our current decision procedure runs in doubly-exponential
time. However, one can expect to lower this complexity: (1) instead of
building all copies indexed from $0$ to $\clockbound$, one could
restrict to polynomially many copies indexed by the useful regions
over 1-clock timed models~\cite{LMS-concur04}, thus yielding a
polynomial reduction to the exact value problem; (2) the exact value
problem is suspected to be solvable in \PSPACE (improving the known
\EXPTIME upper bound~\cite{monmege2025decidability}) by an encoding
into the existential theory of reals. If these two facts are
confirmed, solving the value problem for 1-clock WTG with
robustness would be in \PSPACE, that is no harder (in terms of
complexity classes) than the exact value problem for the same class:
robustness induces no complexity increase.

As future work, beyond clarifying these complexity issues, we plan to
explore whether our copy game construction extends to other classes of
WTGs with decidable robust value problem, such as acyclic WTGs and
divergent WTGs. If so, it would offer a tradeoff between the
complexity (suspected higher with the copy game) and the freedom to
choose the initial valuation (the copy game construction currently only
allows one to compute the value for a fixed initial valuation).

%% file: app_asap.tex
\section[Proof of Lemma 8]{Proof of~\cref{lem:deadlock_asap}}
\label{app:asap}

Before  proving this new characterisation of r-unfeasible paths, we recall
the definition of asap strategies.
Given a path $\ppath=\loc_0 \xrightarrow{\trans_0} \dots \xrightarrow{\trans_{k-1}} \loc_k$,
a perturbation bound $\perturbation$, and $\varepsilon > 0$, we define
an \emph{asap strategy} of \MinPl on $\game_\ppath$,
denoted by $\robminstrategy[\varepsilon]_\ppath$, such that
for all plays $\play \in \rPlaysMin$ where $\play$ follows $\ppath$ and $|\play| < |\ppath|$, we fix
\[
	\robminstrategy[\varepsilon]_\ppath(\play)=
	\begin{cases}
		(\trans_{|\play|}, \delay_{\min}(\val_{|\play|-1}, \guard_{|\play|})) &
		\text{if } \val_{|\play|-1} +  \delay_{\min} \models \guard_{|\play|}; \\
		(\trans_{|\play|}, \delay_{\min}(\val_{|\play|-1}, \guard_{|\play|}) + \varepsilon) & \text{otherwise};
	\end{cases}
\]
where $\guard_{|\play|}$ is the guard of the transition $\trans_{|\play|}$, and, for all
valuations $\val$ and guards $\guard$, we define
$\delay_{\min}(\val, \guard) = \inf \{\delay \mid \val + \delay \in \closureg{\guard}\}$.

We start by proving that this family of strategies allow \MinPl to play as long as it is possible against
a given robust strategy of \MaxPl.
In particular, we prove that, given a robust strategy of \MaxPl, it can mimic its behaviour
against an asap strategy of \MinPl, whatever \MinPl does.

\begin{lemma}
	\label{lem:deadlock-asap-prop}
	Let $\robminstrategy[\varepsilon]_\ppath$ be an asap
        strategies for \MinPl for the path $\ppath$, and
        $\robmaxstrategy \in \rStratMax[\game_\ppath]$.  Then, there
        exists
        $\robmaxstrategy[\varepsilon] \in
        \mathsf{Strat}^{\perturbation + \varepsilon}_{\game_\ppath,
          \MaxPl}$, such that for every robust strategy
        $\robminstrategy \in \rStratMin[\game_\ppath]$ of \MinPl we
        have
        $|\rPlay(\stateI, \robminstrategy,
        \robmaxstrategy[\varepsilon])| \leq |\rPlay(\stateI,
        \robminstrategy[\varepsilon]_{\ppath}, \robmaxstrategy)|$.
\end{lemma}
\begin{proof}
  We define a new robust strategy for \MaxPl,
  $\robmaxstrategy[\varepsilon]$, that will mimic, against all
  possible behaviours of \MinPl, the strategy $\robmaxstrategy$
  against $\robminstrategy[\varepsilon]_\ppath$.  In particular,
  $\robmaxstrategy[\varepsilon]$ always aims at reaching the next
  valuation (when it is possible) obtained by $\robmaxstrategy$
  against $\robminstrategy[\varepsilon]_\ppath$.  Thus, when a move of
  \MinPl does not reach the same next valuation, \MaxPl chooses a
  delay to compensate the deviation of \MinPl.  Otherwise, \MaxPl
  chooses $0$ as delay.

  Formally, given a robust play $\play$ ending in a state of \MaxPl,
  we let $\val$ (resp. $\tilde{\val}$) be the last
  valuation\footnote{We use here the convention that if $\play$ ends
    in a virtual state, say $(\loc, \val', \trans, \delay)$, the last
    valuation of $\play$ is $\val' + \delay$.} of $\play$ (resp. the
  $(|\play| + 1)$-th valuation of the play conforming to
  $\robminstrategy[\varepsilon]_\ppath$ and $\robmaxstrategy$ starting
  in the same real state than $\play$ when its exists), and we define
  $\robmaxstrategy[\varepsilon]$ such that
	\begin{displaymath}
		\robmaxstrategy[\varepsilon](\play) =
		\begin{cases}
			(\trans_{|\play|}, 0) & \text{if } \play \xrightarrow{\robmaxstrategy(\play)}
			(\loc_{|\play| + 1}, \val')
			\text{ with } \val' \geq \tilde{\val}\\
			(\trans_{|\play|}, \tilde{\val}_ - \val) & \text{otherwise}
		\end{cases}
	\end{displaymath}
	First, this strategy is well-defined.  Indeed, since $\pi$ is
        feasible, the move $(\trans_{|\play|}, 0)$ is always enabled.
        Moreover, a difference appears between valuations of both
        plays when \MinPl plays delays that differ from its asap
        strategy.  More precisely, this difference happens when the
        strategy followed by \MinPl approximates the delay with an
        error less than $\varepsilon$ when the guard of the transition
        is open.  Since, $\robmaxstrategy$ does not use a perturbation
        greater that $\perturbation$, this shift can be compensated
        into the perturbation played by \MaxPl.

	\medskip

	Now, we prove by induction on the length of
        $\rPlay(\stateI, \robminstrategy,
        \robmaxstrategy[\varepsilon])$ that, the valuation of real
        states of
        $\rPlay(\stateI, \robminstrategy,
        \robmaxstrategy[\varepsilon])$ is greater than the
        corresponding one in
        $\rPlay(\stateI, \robminstrategy[\varepsilon]_\ppath,
        \robmaxstrategy)$.  More precisely, by letting $\val_n$
        (resp. $\tilde{\val}_n$) the $n$-th valuation of
        $\rPlay(\stateI, \robminstrategy[\varepsilon]_{\ppath},
        \robmaxstrategy)$ (resp. of
        $\rPlay(\stateI, \robminstrategy,
        \robmaxstrategy[\varepsilon])$), we prove by induction on
        $n \in \N$ that $\val_n \leq \tilde{\val}_n$.

        \noindent If $n = 0$, then $\val_0 = 0 = \tilde{\val}_0$, so
        that the base case holds.

        \noindent Otherwise, let $n \in \N$ such that
        $\val_n \leq \tilde{\val}_n$.  We denote by $\loc_n$ the
        $n$-th location of both plays, that follow the same path.  The
        interesting case is when transition $\trans_n$, with source
        $\loc_n$, does not contain a reset.  Otherwise,
        $\val_{n+1} = 0 = \tilde{\val}_{n+1}$ and the property
        trivially holds.  If $\trans_n$ does not reset $x$, we
        distinguish two cases:
	\begin{itemize}
        \item If $\loc_n \in \LocsMax$, then as $\ppath$ is feasible,
          in both configurations, there exists an enabled move.
          Moreover, by definition of $\robmaxstrategy[\varepsilon]$,
          we know that
          $\tilde{\val}_{n+1} = \max(\val_{n+1}, \tilde{\val}_n)$ and
          $\val_{n+1} \leq \tilde{\val}_{n+1}$.

        \item Otherwise, $\loc_n \in \LocsMin$, and we again
          distinguish two subcases.  If the delay played by
          $\robminstrategy$ is smaller than the one played by
          $\robminstrategy[\varepsilon]_\ppath$, then
          $\robmaxstrategy[\varepsilon]$ compensates the difference
          that is at most $\varepsilon$ (by definition of
          $\robminstrategy[\varepsilon]_\ppath$) ensuring
          $\val_{n+1} = \tilde{\val}_{n+1}$. Else, the delay played by
          $\robminstrategy$ is greater than the one played by
          $\robminstrategy[\varepsilon]_\ppath$, thus
          $\val_{n+1} \leq \tilde{\val}_{n+1}$.
	\end{itemize}

	\medskip

	Finally, to prove that
        $|\rPlay(\stateI, \robminstrategy,
        \robmaxstrategy[\varepsilon])| \leq |\rPlay(\stateI,
        \robminstrategy[\varepsilon]_{\ppath}, \robmaxstrategy)|$, we
        reason by contradiction.  Suppose that
        $|\rPlay(\stateI, \robminstrategy[\varepsilon]_{\ppath},
        \robmaxstrategy)| < |\rPlay(\stateI, \robminstrategy,
        \robmaxstrategy[\varepsilon])|$, and the last
        location\footnote{Formally, this last location can be the same
          as the one of $\ppath$, but is not its last occurrence.} of
        $\rPlay(\stateI, \robminstrategy[\varepsilon]_{\ppath},
        \robmaxstrategy)$ is $\loc_k$ with $k < |\ppath|$.  Again, we
        distinguish two cases:
	\begin{itemize}
        \item If $\loc_k \in \LocsMin$, then we let
          $\last(\rPlay(\stateI,
          \robminstrategy[\varepsilon]_{\ppath}, \robmaxstrategy)) =
          (\loc_k, \val)$ and $(\loc_k, \tilde{\val})$ be the $k$-th
          valuation of
          $\rPlay(\stateI, \robminstrategy,
          \robmaxstrategy[\varepsilon]_{\ppath})$.  Moreover, we know
          that $\robminstrategy[\varepsilon]_{\ppath}$ is not defined
          in
          $\rPlay(\stateI, \robminstrategy[\varepsilon]_{\ppath},
          \robmaxstrategy)$, but $\robminstrategy$ is well-defined
          (with $(\trans_{k+1}, \delay)$) in the prefix of length $k$
          of
          $\rPlay(\stateI, \robminstrategy,
          \robmaxstrategy[\varepsilon]_{\ppath})$.  Thus, we deduce
          that $(\trans, \delay + \tilde{\val} - \val)$ is an enabled
          edge (since $\val \leq \tilde{\val}$ by the previous
          remark). This contradicts the definition of asap strategy.

        \item Otherwise, $\loc_k \in \LocsMax$, and we let $\states$
          (resp. $\tilde{\states}$) be the last state of
          $\rPlay(\stateI, \robminstrategy[\varepsilon]_{\ppath},
          \robmaxstrategy)$ (resp.
          $\rPlay(\stateI, \robminstrategy,
          \robmaxstrategy[\varepsilon]_{\ppath})$).  In particular,
          since $\robmaxstrategy$ is not defined on
          $\rPlay(\stateI, \robminstrategy[\varepsilon]_{\ppath},
          \robmaxstrategy)$, we know that
          $E^\perturbation(\states) = \emptyset$.  However, since
          $\rPlay(\stateI, \robminstrategy,
          \robmaxstrategy[\varepsilon]_{\ppath})$ is defined after the
          $k$-th real state,
          $E^\perturbation(\tilde{\states}) \neq \emptyset$.  By
          letting $\val$ (resp. $\tilde{\val}$) the
          valuation\footnote{If $\states$ is a virtual state,
            e.g. $((\loc, \val), \trans, \delay)$, its last valuation
            is $\val + \delay$.} of $\states$
          (resp. $\tilde{\states}$), the previous remark ensures that
          $\val \leq \tilde{\val}$ and
          $E^\perturbation(\states) \subseteq
          E^\perturbation(\tilde{\states})$.  Here also we reached a
          contradiction.  \qedhere
	\end{itemize}
\end{proof}

Now, we can prove~\cref{lem:deadlock_asap}.

\asap*
\begin{proof}
  By definition of an r-unfeasible path, it suffices to prove the
  bottom to top implication.  We suppose the existence of
  $\perturbation > 0$, $\robmaxstrategy \in \rStratMax$  a robust
  strategy for $\MaxPl$, and $\varepsilon >0 $ such that for every
  $\widehat{\varepsilon}<\varepsilon$,
  $|\rPlay(\stateI,\robminstrategy[\widehat{\varepsilon}]_\ppath,
  \robmaxstrategy) | < |\ppath|$, and we prove that $\ppath$ is
  r-unfeasible.

  We consider $\perturbation' = \perturbation + \varepsilon$, and
  $\robmaxstrategy[\varepsilon]$ the robust strategy of \MaxPl given
  by~\cref{lem:deadlock-asap-prop} applied to
  $(\robminstrategy[\varepsilon]_{\ppath})_\varepsilon$ and
  $\robmaxstrategy$.  Let $\robminstrategy$ be a robust strategy of
  \MinPl, and $\widehat{\varepsilon}<\varepsilon$. W have:
	\begin{align*}
		|\rPlay(\stateI, \robminstrategy, \robmaxstrategy[\varepsilon])| &\leq
		|\rPlay(\stateI, \robminstrategy[\widehat{\varepsilon}]_{\ppath}, \robmaxstrategy)|
		& \text{(by~\cref{lem:deadlock-asap-prop})}
		\\ &< |\ppath| &\text{(by hypothesis)}
	\end{align*}
	Therefore, $\ppath$ is an r-unfeasible path.
\end{proof}

%% file: app_blocage.tex
\section[Proof of Proposition 9]{Proof of~\cref{prop:blocage}}
\label{app:blocage}

\blocage*

We prove the two implications, starting with the easiest one.

\medskip
\noindent
\textbf{($\Leftarrow$) A path following one of the patterns is
  r-unfeasible.}  We suppose that $\ppath$ can be decomposed as
$\ppath_0 \xrightarrow{\trans_0} \ppath_1 \xrightarrow{\trans_1}
\ppath_2$ such that $\ppath_1 \in \Paths^{Y=\emptyset}_{\MinPl}$, and
$\guard_i$ denote the guard of $\trans_i$ for $i \in \{0, 1\}$.  We
distinguish two cases according the pattern followed by $\ppath$.
\begin{description}
\item[type 1] We suppose that $\ppath_0 \in \Paths_{\MinPl}$ and
  $\leftg{\guard_0} = \rightg{\guard_1}$.  Let $\perturbation > 0$ and
  $\robmaxstrategy \in \rStratMax[\game_\ppath]$ be a robust strategy
  of \MaxPl that perturbs the choice of $\MinPl$ by $\perturbation$
  after $\trans_0$. Formally, for every play $\play$ such that
  $\last(\play) = (\states, \trans_0, \delay)$, we have
  $\robmaxstrategy(\play) = \perturbation$.  In particular, by
  definition of $\trans_0$, after the perturbation applied by
  $\robmaxstrategy$, we have
  $\play \xrightarrow{\robmaxstrategy(\play)} (\loc, \val)$ with
  $\val \geq \leftg{\guard_0} + \perturbation$.

  Now, we consider a robust strategy of \MinPl,
  $\robminstrategy \in \rStratMin[\game_\ppath]$, and we prove that
  $\ppath$ is r-unfeasible by showing that
  $|\rPlay(\stateI, \robminstrategy, \robmaxstrategy)| < |\ppath|$.
  By contradiction, we suppose that
  $|\rPlay(\stateI, \robminstrategy, \robmaxstrategy)| \geq |\ppath|$,
  i.e. $\robminstrategy$ can find a robust decision for $\trans_1$.
  Let $\play$ be the prefix of
  $\rPlay(\stateI, \robminstrategy, \robmaxstrategy)$ ending in the
  last location of $\ppath_1$ and $\val'$ be its last valuation.
  Since $\ppath_1$ contains no reset (by hypothesis), we have
  $\val' \geq \leftg{\guard_0} + \perturbation = \rightg{\guard_1} +
  \perturbation > \rightg{\guard_1}$.  Thus,
  $E^\perturbation(\last(\play)) = \emptyset$ and \MinPl has no robust
  decision for $\trans_1$. We reached a contradiction.

\item[type 2] We suppose that $\ppath_0 \in \Paths_{\MaxPl}$ and
  $\rightg{\guard_0} = \rightg{\guard_1}$.  Let $\perturbation > 0$,
  and $\robmaxstrategy \in \rStratMax[\game_\ppath]$ be a strategy
  that chooses to wait the valuation
  $\rightg{\guard_0} - \perturbation$ along $\trans_0$. Formally, for
  every play $\play$ following $\ppath_0$, if $\val$ denotes its last
  valuation,
  $\robmaxstrategy(\play) = (\trans_0, \max(\rightg{\guard_0} -
  \perturbation - \val, 0))$.  In particular, after applying
  $\robmaxstrategy$ on a play $\play$ following
  $\ppath_0 \xrightarrow{\trans_0}$, we have
  $\play \xrightarrow{\robmaxstrategy(\play)} (\loc, \val)$ with
  $\val \geq \rightg{\guard_0} - \perturbation$.

  Now, we consider a robust strategy of \MinPl,
  $\robminstrategy \in \rStratMin[\game_\ppath]$, and we prove that
  $\ppath$ is r-unfeasible by showing that
  $|\rPlay(\stateI, \robminstrategy, \robmaxstrategy)| < |\ppath|$.
  By contradiction, we suppose that
  $|\rPlay(\stateI, \robminstrategy, \robmaxstrategy)| \geq |\ppath|$,
  i.e. $\robminstrategy$ can find a robust decision for $\trans_1$.
  Let $\play$ be the prefix of
  $\rPlay(\stateI, \robminstrategy, \robmaxstrategy)$ ending in the
  last location of $\ppath_1$ and $\val'$ be its last valuation.
  Since $\ppath_1$ contains no reset (by hypothesis), we have
  $\val' \geq \rightg{\guard_0} - \perturbation = \rightg{\guard_1} -
  \perturbation$.  Thus, $E^\perturbation(\last(\play)) = \emptyset$
  and \MinPl has no robust decision for $\trans_1$. We reached a
  contradiction.
\end{description}

\medskip
\noindent
\textbf{($\Rightarrow$) r-unfeasible paths always follow one of the
  patterns.}  The converse implication relies
on~\cref{lem:deadlock_asap}.  We suppose that $\ppath$ is an
r-unfeasible yet feasible path.  Then, by letting
$(\robminstrategy[\varepsilon]_\ppath)_\varepsilon$ be a familly of
asap strategies, there exists $\perturbation > 0$, $\varepsilon > 0$,
and $\robmaxstrategy \in \rStratMax[\game_\ppath]$ such that for all
$0 < \varepsilon' \leq \varepsilon$, and
$|\rPlay(\stateI,\robminstrategy[\varepsilon']_{\ppath},
\robmaxstrategy)| < |\ppath|$, by~\cref{lem:deadlock_asap}.

If $\states$ denotes the last state of
$\rPlay(\stateI,\robminstrategy[\varepsilon']_{\ppath},
\robmaxstrategy)$, and $\guard$ the guard of the last transition of
$\rPlay(\stateI,\robminstrategy[\varepsilon']_{\ppath},
\robmaxstrategy)$, we prove that $\ppath$ follows one of the two
patterns.  Since $\ppath$ is feasible, \MaxPl can always choose a next
move whatever the current configuration. We deduce that
$\states \in \StatesMin$ and we write $\states = (\loc, \val)$.
Moreover, by definition of asap strategies, we know that
$E(\loc, \val) = \emptyset$ and
$\val > \rightg{\guard} - \perturbation$ (with equality when
$\rightg{\guard} \notin \guard$) since $\ppath$ is feasible.  The
valuation $\val$ is reached by one of the players in a previous move
either by choice (for \MaxPl), or by force (for \MinPl).  In
particular, we distinguish two cases.

\noindent First suppose that \MinPl chooses to reach this valuation
since there exists, after the last reset along $\ppath$, a transition
from a location of \MinPl such that its guard $\tilde{\guard}$ is such
that $\rightg{\guard} \leq \leftg{\tilde{\guard}}$.  In particular,
$\ppath$ can be decomposed as
$\ppath_0 \xrightarrow{\tilde{\trans}} \ppath_1 \xrightarrow{\trans}
\ppath_2$ where $\ppath_0 \in \PPathsMin$ and
$\ppath_1 \in \PPathsMin^{Y=\emptyset}$.  To conclude, we prove that
$\leftg{\tilde{\guard}} = \rightg{\guard}$ since $\ppath$ is
r-unfeasible.  By contradiction, we suppose that
$\rightg{\guard} < \leftg{\tilde{\guard}}$.  In particular, for every
(non-robust) strategies of \MinPl and \MaxPl, denoted respectively by
$\minstrategy$ and $\maxstrategy$, we have
$\last(\Play((\loc_0,0), \minstrategy_{\ppath_0},
\maxstrategy_{\ppath_0})) = (\loc,\val)$ with
$\val \geq \leftg{\tilde{\guard}} > \rightg{\guard}$.  Morevoer, since
$\ppath_1$ contains no reset, all plays conforming to these two
strategies end $\ppath_1$ with a valuation greater than
$\rightg{\guard}$.  Thus, $\ppath$ is unfeasible, or
$\rightg{\guard} = \leftg{\tilde{\guard}}$. This corresponds to a type
1 pattern.

Now suppose that \MaxPl chooses to reach this valuation, i.e. there
exists, after the last reset along $\ppath$, a transition
$\tilde{\trans}$ from a location $\tilde{\loc}$ of \MaxPl such that
$\rightg{\guard} \leq \rightg{\tilde{\guard}}$ where $\tilde{\guard}$
is the guard of $\trans$.  In particular, $\ppath$ can be decomposed as
$\ppath_0 \xrightarrow{\tilde{\trans}} \ppath_1 \xrightarrow{\trans}
\ppath_2$ where $\ppath_0 \in \PPathsMax$ and
$\ppath_1 \in \PPathsMin^{Y=\emptyset}$.  To conclude, we prove that
$\rightg{\tilde{\guard}} = \rightg{\guard}$ since $\ppath$ is
r-unfeasible.  By contradiction, we suppose that
$\rightg{\guard} < \rightg{\tilde{\guard}}$.  In particular, there
exists $\maxstrategy$ a (non-robust) strategy for \MaxPl such that for
all plays $\play$ with $\last(\play) = (\tilde{\loc}, \val)$, we
have
\begin{displaymath}
	\maxstrategy(\play)= (\tilde{\trans}, \delay)
	\qquad\qquad \text{where $\delay$ is such that
		$\delay + \val \in (\rightg{\guard}, \rightg{\tilde{\guard}}] \cap \tilde{\guard}$}
\end{displaymath}
By hypothesis on bounds of guards, this delay and $\maxstrategy$
exist.  Moreover, since $\ppath_1$ contains no reset, all plays
conforming to $\maxstrategy$ end $\ppath_1$ with a valuation greater
than $\rightg{\guard}$.  Thus, $\ppath$ is unfeasible, or
$\rightg{\tilde{\guard}} = \rightg{\guard}$. This corresponds to a type 2 pattern. \qed

%% file: app_copy.tex
\section[Proof of Lemma 11]{Proof of~\cref{lem:copy}}
\label{app:copy}

\copyProp*

\medskip
\noindent
\textbf{Proof of~\cref{itm:copy_caracterisation}} We reason by
contradiction, and we suppose that there exists $\ppath$ a feasible
and r-unfeasible path of type 1 in $\C$. We let
$\ppath = \pi_0 \xrightarrow{d^0} \pi_1 \xrightarrow{d^1}
\pi_2$ be such a path, with $\pi_1 \in \Paths^{Y=\emptyset}_{\MinPl}$,
$\pi_0 \in \Paths_{\MinPl}$, and
$\leftg{\guard_0} = \rightg{\guard_1}$ where $\guard_i$ is the guard
of $d^i$ for $i \in \{0, 1\}$.

Since in $\game$ --and hence\footnote{In $\C$, the guard of a
  transition from a location of \MinPl may only be reduced when a new
  copy is reached, whose index is smaller than the supremum of the
  original guard in $\game$. As a consequence, the reduced guard
  contains an interval of length at least $1$, and is thus not
  punctual.} in $\C$-- there are no punctual guards from locations of
\MinPl, and $\pi_1$ contains no reset, we deduce that
$a :=\leftg{\guard_0} = \rightg{\guard_1} > 0$. By~\cref{def:copy},
$d^0$ in $\C$ stems from a transition $\trans^0$ in $\game$ and
reaches the $a$-th copy: $d^0 = \trans^0_{a'\to a}$, for some
$a' \in \N$.  Moreover, since $\ppath_1$ does not contain any reset,
$d^1$ is a transition in the $b$-th copy, for some $b \geq a$, i.e.
$d^1 = \trans^1_{b \to b'}$ for some transition $\trans^1$ in $\game$
and some $b' \in \N$.  Again by~\cref{def:copy}, it must be that
$\rightg{\guard_1} > b \geq a = \leftg{\guard_0}$. We reached a
contradiction with the assumption that
$\leftg{\guard_0} = \rightg{\guard_1}$

\medskip
\noindent
\textbf{Proof of~\cref{itm:copy_size}}
We evaluate the number of  locations, edges and the constants of $\C$.

By definition of the copy game,
$\Locs^\C = \Locs \times \llbracket 0, M \rrbracket$, thus
$|\Locs^\C| = |\Locs| \cdot (|\clockbound| +1)$.  Since constants are
encoded in binary, $|\Locs^\C|$ is exponential in the size of $\game$.

Again by~\cref{def:copy},
$\Trans^\C = \bigcup_{a \in \llbracket 0, \maxBound\rrbracket}
\cup_{\trans \in \Trans} \Trans^{\C}_a(\trans)$, thus
$|\Trans^\C| = \sum_{a \in \llbracket 0, \maxBound\rrbracket}
\sum_{\trans \in \Trans} |\Trans^{\C}_a(\trans)|$.  Since,
$\Trans^{\C}_a(\trans)$ contains at most $\clockbound + 1$ transitions
(this happens when $\trans$ is a transition from a location of \MaxPl
that contains no reset), we deduce that
$|\Trans^\C| \leq \sum_{a \in \llbracket 0, \maxBound\rrbracket}
\sum_{\trans \in \Trans} (\clockbound + 1) = \clockbound \cdot
|\Trans| \cdot (\clockbound + 1)$.  Thus, $|\Trans^\C|$ is exponential
in the size of $\game$.

Finally, the maximal constant in $\C$ coincides with the maximal
constant of $\game$, as weights and upper bounds on
  guards are not modified in $\C$.

All in all, the size of $\C$ is exponential on the size of $\game$.

\medskip
\noindent
\textbf{Proof of~\cref{itm:copy_end-val}} The proof is by induction on
the length of a play~$\play$.  The base case is immediate as the
initial copy is copy~$0$.  Suppose now that
$\play = \play_1 \xrightarrow{d, \delay} \states$, and, letting
$\val_1$ be the last valuation\footnote{If $\play_1$ ends in a
  virtual state $((\loc, \val), d, \delay)$,
  $\val_1 = \val + \delay$.} of $\play_1$, we distinguish cases
according to the transition $d$ of $\C$.
\begin{itemize}
\item We suppose that $a = 0$ and $d = \trans_{c \to 0}$ for a given
  transition $\trans$ of $\game$ and for some $c \in \N$, i.e. $c = 0$
  or $\trans$ --and $d$-- contain a reset.  Then,
  $\states = (\tuple{\loc, 0}, \val)$ for $\loc$ a location of $\game$
  and $\val \in \Rplus$ a valuation, and $\val \geq 0 = a$.

\item We suppose that $a \neq 0$ and $d = \trans_{c \to a}$ for a
  given $\trans$ and a given $c \in \N$ such that
  $c = \max(a, \leftg{\guard})$, i.e. $\trans$ and $d$ have no reset
  and if $c \neq a$, then the source both of $\trans$ and $d$ is a location of 
  \MinPl.  Then $\states = (\tuple{\loc, a}, \val)$ for
  $\loc$ a location of $\game$ and $\val$ a valuation of $\Rplus$ such
  that $\val = \val_1 + \delay \geq \val_1$ (since $d$ has no reset,
  and $\delay \geq 0$).  To conclude, we apply the induction
  hypothesis to $\play_1$, and obtain that $\val_1 \geq c \geq a$.
	
\item We suppose that $a \neq 0$ and $d = \trans^j_{c \to a}$ for some
  $\trans$ and a given $c \in \N$, i.e. $d$ is a jumping transition
  without reset.  Then, by~\cref{def:copy},
  $\states = (\tuple{\loc, a}, a)$, and $\val = a \geq a$.
\end{itemize}

This concludes the proof of~\cref{lem:copy}.~\qed

%% file: app_copy-rval.tex
\section[Proof of Proposition 13]{Proof of~\cref{prop:preservation_rval}}

\preservationRval*

As explained in the core of the paper, we prove the two inequalities
successively.

\subsection[Proof of the first inequality]{Proof of the first inequality
	$\rValue[0+]_\game(\loc, \val) \leq \rValue[0+]_\C(\tuple{\loc, 0}, \val)$}
\label{app_preservation-rval}

To prove this inequality it suffices to show that
\begin{equation}\label{eq:value_greater}
  \forall \perturbation > 0,\ 	\rValue_\game(\loc, \val)\leq \rValue_\C(\langle \loc,0 \rangle, \val)
\end{equation}
Indeed, taking the limit of each term of the inequality when
$\perturbation$ tends to $0$, yields the desired
inequality.  Thus, the proof consists in defining
$\robmaxstrategy_\C \in \rStratMax[\C]$ from a given strategy
$\robmaxstrategy_\game \in \rStratMax[\game]$ such that its robust
values is greater than the one of $\robmaxstrategy_\game$.
Intuitively, in $\C$, \MaxPl
can follow the strategy $\robmaxstrategy_\game$ without modifying any
delay, and more importantly without taking any jumping transition.  It
therefore suffices to consider plays in $\C$ without jumping
transitions, and one can naturally define a mapping from such
non-jumping plays in $\C$ to plays in $\game$. Building on this
mapping, $\robmaxstrategy_\C \in \rStratMax[\C]$ can be formally
defined so that it mimicks
$\robmaxstrategy_\game \in \rStratMax[\game]$ and the weights of plays
are preserved. In the rest of this subsection, we formalize the
mapping of plays from $\C$ to $\game$ and the definition of strategy
$\robmaxstrategy_\C \in \rStratMax[\C]$.

\medskip
\noindent
\textbf{Definition of $\convert_{\C \to \game}$.}  We start by the
natural definition of a "projection" function that maps robust plays
without jumping transitions from $\C$ to robust plays in $\game$.  We
let $\FPlays[\C]_{\neg j}$ be the set of plays in $\C$ without jumping
transitions, and we define
$\convert_{\C \to \game} \colon \rFPlays[\C]_{\neg j} \to
\rFPlays[\game]$ by induction on the length of plays in $\C$ by
\begin{displaymath}
	\convert_{\C \to \game}(\play) = 
	\begin{cases}
		(\loc, \val) & \text{if $\play = (\tuple{\loc, a}, \val)$ } \\
		\convert_{\C \to \game}(\play_1) \xrightarrow{\trans, \delay} (\loc, \val)
		& \text{if $\play = \play_1 \xrightarrow{\trans_{a \to b}, \delay} (\tuple{\loc, b}, \val)$} \\
		\convert_{\C \to \game}(\play_1) \xrightarrow{\trans, \delay} ((\loc, \val), \trans, \delay)
		& \text{if $\play = \play_1 \xrightarrow{\trans_{a \to b}, \delay}
		((\tuple{\loc, b}, \val), \trans, \delay)$}
	\end{cases}
\end{displaymath}
Since the upper bounds of guards in $\C$ coincide with those in
$\game$ for non-jumping transitions, all decisions along
$\convert_{\C \to \game}(\play)$ are robust, and
$\convert_{\C \to \game}(\play)$ is well-defined.  Moreover,
$\convert_{\C \to \game}$ enjoys the following interesting properties
on plays:
\begin{lemma}
	\label{lem:rVal_convert-C-G}
	Let $\play \in \rFPlays[\C]_{\neg j}$.
	Then,
	\begin{enumerate}
		\item\label{itm:rVal_convert-C-G_last}
		$\convert_{\C \to \game}$ preserves the last real state, i.e.
		if $\last(\play) = (\tuple{\loc, a}, \val)$, then
		$\last(\convert_{\C \to \game}(\play)) = (\loc, \val)$;
	    \item\label{itm:rVal_convert-C-G_weight}
		$\convert_{\C \to \game}$ preserves the weight, i.e.
		$\weight(\play) = \weight(\convert_{\C \to \game}(\play))$.
	\end{enumerate}
\end{lemma}
\begin{proof}
	\begin{enumerate}
	    \item is immediate by definition of $\convert_{\C \to \game}$.
		\item We reason by induction on the length of $\play$, and we prove that
		$\weightC(\play) = \weightC(\convert_{\C \to \game}(\play))$.
		Indeed, by~\cref{itm:rVal_convert-C-G_last}, if $\play$ ends in a real state, then
		$\convert_{\C \to \game}(\play)$ ends in a real state belonging to the same player.
		Thus,
		\begin{displaymath}
			\weight(\convert_{\C \to \game}(\play)) =
			\begin{cases}
				+\infty & \text{if $\last(\play)$ is not a target state} \\
				\weightC(\play) & \text{otherwise}
			\end{cases}
		\end{displaymath}
		and this coincides with the definition of $\weight(\play)$.

		To conclude, we prove that
                $\weightC(\play) = \weightC(\convert_{\C \to
                  \game}(\play))$.  First, if
                $\play = (\tuple{\loc, a}, \val)$, then
                $\convert_{\C \to \game}(\play) = (\loc, \val)$ and
                $\weightC(\play) = 0 = \weightC(\convert_{\C \to
                  \game}(\play)))$.

		Then, we suppose that $\play = \play_1 \xrightarrow{\trans_{a \to b}, \delay}
		(\tuple{\loc, b}, \val)$ and, by definition, $\convert_{\C \to \game}(\play) =
		\convert_{\C \to \game}(\play_1) \xrightarrow{\trans, \delay} (\loc, \val)$.
		Thus, we have
		\begin{align*}
			\weightC(\convert_{\C \to \game}(\play)) &=
			\weightC(\convert_{\C \to \game}(\play_1)) + \weight(\trans) +
			\delay\cdot\weight(\last(\convert_{\C \to \game}(\play_1))) \\
			&= \weightC(\play_1) + \weight(\trans) +
			\delay\cdot\weight(\last(\convert_{\C \to \game}(\play_1))) \qquad
			\text{(by (IH))}\\
			&= \weightC(\play_1) + \weight(\trans) +
			\delay\cdot\weight(\last(\play_1)) \qquad
			\text{(by~\cref{itm:rVal_convert-C-G_last})}\\
			&= \weightC(\play_1) + \weight(\trans_{a \to b}) +
			\delay\cdot\weight(\last(\play_1)) \qquad
			\text{(by definition of $\C$)}\\
			&= \weightC(\play)
		\end{align*}
		as claimed.
		\qedhere
	\end{enumerate}
\end{proof}

\medskip
\noindent
\textbf{Definition of $\robmaxstrategy_{\C}$.} We are now in a
position to define formally $\robmaxstrategy_{\C}$. Let
$\robmaxstrategy_\game$ be a robust strategy of \MaxPl in $\game$.
Intuitively, $\robmaxstrategy_{\C}$ in $\C$ mimicks the choices of
$\robmaxstrategy_\game$ in $\game$ and does not use jumping
transitions.  Formally, for a play $\play \in \rFPlays[\C]_{\neg j}$
ending in a state of \MaxPl, letting
$(\trans, \delay) = \robmaxstrategy_{\game}(\convert_{\C \to
  \game}(\play))$, we define
\begin{displaymath}
	\robmaxstrategy_{\C}(\play) =
	\begin{cases}
		(\trans_{a \to a}, \delay) & \text{if $\last(\play)$ is in the $a$-th copy, and $\trans$ does not reset $x$}  \\
		(\trans_{a \to 0}, \delay) & \text{if $\last(\play)$ is in the $a$-th copy, and $\trans$ resets $x$}  \\
                \text{any enabled move} & \text{otherwise.}
	\end{cases}
\end{displaymath}
Observe that, by definition of $\C$, the guard is not modified from
$\trans$ to $\trans_{a \to b}$ for any $b \in \{a, 0\}$.  Thus, robust
choices of $\maxstrategy_\game$ are also robust in $\C$, and this new
strategy is well-defined.  Finally, we prove that $\robmaxstrategy_\C$
induces non-jumping plays only, and that it is coherent with the
above-defined mapping on plays:
\begin{lemma}
	\label{lem:rVal_strat-G-C}
	Let $\play \in \rFPlays[\C]$ be a play that conforms to $\robmaxstrategy_\C$.
	Then,
	\begin{enumerate}
	    \item\label{itm:rVal_strat-G-C_1} $\play \in \rFPlays[\C]_{\neg j}$ ;
            \item\label{itm:rVal_strat-G-C_2}
              $\convert_{\C \to \game}(\play)$ conforms to
              $\robmaxstrategy_\game$.
	\end{enumerate}
\end{lemma}
\begin{proof}
	\begin{enumerate}
        \item We reason by contradiction, and we suppose that $\play$
          contains at least one jumping transition.  Let
          $\trans^j_{a \to b}$ be the first jumping transition of
          $\play$, and let $\play_1$ be the prefix of $\play$ until
          $\trans^j_{a \to b}$, i.e.
          $\play_1 \in \rFPlays[\C]_{\neg j}$.  By definition of $\C$,
          $\trans^j_{a \to b}$ as for source a configuration of \MaxPl
          (jumping transitions only outgoing locations of \MaxPl).
          Thus, by assumption on $\play$, $\trans^j_{a \to b}$ was
          chosen by $\robmaxstrategy_\C$.  This contradicts the fact that 
          $\play_1 \in \rFPlays[\C]_{\neg j}$,
          $\robmaxstrategy_\C(\play_1)$ does not use jumping
          transitions.

        \item We reason by induction on the length of $\play$.  If
          $\play = (\tuple{\loc, a}, \val)$, then
          $\convert_{\C \to \game}(\play) = (\loc, \val)$ and it
          conforms to $\robmaxstrategy_\game$.  Otherwise, we suppose
          that
          $\play = \play_1 \xrightarrow{\trans_{a \to b}, \delay}
          \states$, and we distinguish two cases according to the last
          state of $\play_1$.
		\begin{itemize}
                \item If $\last(\play_1)$ is a configuration of
                  \MinPl, then the property holds by induction
                  hypothesis applied to $\play_1$ and
                  $\convert_{\C \to \game}(\play_1)$.

                \item Otherwise, $\last(\play_1)$ is a state of \MaxPl
                  and $\states = (\tuple{\loc, b}, \val)$.  Then , by
                  definition of $\convert_{\C \to \game}$, we deduce
                  that
                  $\convert_{\C \to \game}(\play) = \convert_{\C \to
                    \game}(\play_1) \xrightarrow{\trans, \delay}
                  (\loc, \val)$.  Moreover, since $\play$ conforms to
                  $\robmaxstrategy_\C$, $\last(\play_1)$ must be in
                  the $a$-th copy of $\C$ and
                  $\robmaxstrategy_\game(\convert_{\C \to
                    \game}(\play_1)) = (\trans, \delay)$.  Applying
                  the induction hypothesis to $\play_1$, we conclude
                  that $\convert_{\C \to \game}(\play)$ conforms to
                  $\robmaxstrategy_\game$. \qedhere
		\end{itemize}
	\end{enumerate}
\end{proof}

To conclude, we have shown that for every
$\robmaxstrategy \in \rStratMax[\game]$, we have
\begin{displaymath}
	\inf_{\robminstrategy \in \rStratMin[\game]} \weight(\outcomes((\loc, \val), \robminstrategy, \robmaxstrategy_\game))
	\leq  \inf_{\robminstrategy \in \rStratMin[\C]} \weight(\outcomes((\tuple{\loc, 0}, \val), \robminstrategy, \robmaxstrategy_\C))
	\leq \rValue[\perturbation]_\C(\tuple{\loc, 0}, \val)
\end{displaymath}
Since
this inequality holds for every
$\robmaxstrategy \in \rStratMax[\game]$, we deduce the claimed
inequality
$\rValue[0+]_\game(\loc, \val) \leq \rValue[0+]_\C(\tuple{\loc, 0}, \val)$.
\qed

\subsection[Proof of the second inequality]{Proof of the second inequality $\rValue[0+]_\game(\loc, \val) \geq \rValue[0+]_\C(\tuple{\loc, 0}, \val)$}
\label{app_preservation-rval-2}
Let us prove this second inequality by showing:
\begin{displaymath}
\forall	  \perturbation > 0,\ \forall   \varepsilon > 0,\ \rValue_\C(\tuple{\loc, 0}, \val)
	\leq \rValue_\game(\loc, \val) + \varepsilon.
\end{displaymath}
For a fixed perturbation $\perturbation$ and a fixed margin
$\varepsilon$, from a given strategy
$\robmaxstrategy_\C \in \rStratMax[\C]$ we define
$\robmaxstrategy_\game \in \rStratMax[\game]$ such that their $\perturbation$-robust
values differ by at most $\varepsilon$.

The difficulty in defining $\robmaxstrategy_\game$ is when
$\robmaxstrategy_\C$ chooses a jumping transition for which the guard
is not included in the one of the original transition of $\game$.  In
all other cases, guards of $\C$ are included in $\game$, and \MaxPl
can mimick the decisions from $\C$ in $\game$. In order to deal with
such jumping transitions, we slightly alter the delay chosen by
$\robmaxstrategy_\C$.  Before formally defining
$\robmaxstrategy_\game$, we first define a mapping from plays in
$\game$ to plays in $\C$ that conform to $\robmaxstrategy_\C$.

\medskip
\noindent
\textbf{Definition of $\convert^{\robmaxstrategy_\C}_{\game \to \C}$}
We define a function that maps a robust play in $\game$ to a play in
the copy game $\C$ while preserving decisions of $\robmaxstrategy_\C$.
Formally,
$\convert^{\robmaxstrategy_\C}_{\game \to \C} \colon \rFPlays[\game]
\to \rFPlays[\C]$ is defined by induction on the length of plays in
$\game$. Initially,
$\convert^{\robmaxstrategy_\C}_{\game \to \C}((\loc, \val)) =
(\tuple{\loc, 0}, \val)$.  Moreover, letting $\maxWeightLoc$ for the
maximal absolute value of weights in $\C$, for every
$\play \in \rFPlays[\game]$ of the form
$\play = \play_1 \xrightarrow{\trans, \delay} s$, we define
$\convert^{\robmaxstrategy_\C}_{\game \to \C}$ according to the last
state of $\play_1$: 
\begin{itemize}
\item if $\last(\play_1)$ is a state of \MaxPl, then
$s= (\loc, \val)$ is a configuration of \MinPl and we
  distinguish three cases according to the conformance of
  $\convert^{\robmaxstrategy_\C}_{\game \to \C}(\play_1)$ to
  $\robmaxstrategy_\C$.
	\begin{itemize}
        \item If
          $\convert^{\robmaxstrategy_\C}_{\game \to \C}(\play_1)$
          conforms to $\robmaxstrategy_\C$ and
          $\robmaxstrategy_\C(\convert^{\robmaxstrategy_\C}_{\game \to
            \C}(\play_1)) = (\trans_{a \to b}^{(j)}, \delay)$, then we
          let
		\begin{displaymath}
                  \convert^{\robmaxstrategy_\C}_{\game \to \C}(\play) =
                  \convert^{\robmaxstrategy_\C}_{\game \to \C}(\play_1)
                  \xrightarrow{\trans_{a \to b}^{(j)}, \delay} (\tuple{\loc, b}, \val)  \,;
		\end{displaymath}
              \item If
                $\convert^{\robmaxstrategy_\C}_{\game \to
                  \C}(\play_1)$ conforms to $\robmaxstrategy_\C$ and
                $\robmaxstrategy_\C(\convert^{\robmaxstrategy_\C}_{\game
                  \to \C}(\play_1)) = (\trans_{a \to b}^{j}, \delay +
                \frac{\varepsilon}{W_{\loc} \cdot 2^{|\play|+1}})$,
                then we let
		\begin{displaymath}
                  \convert^{\robmaxstrategy_\C}_{\game \to \C}(\play) =
                  \convert^{\robmaxstrategy_\C}_{\game \to \C}(\play_1)
                  \xrightarrow{\trans_{a \to b}^j, \delay + \frac{\varepsilon}{W_{\loc} \cdot 2^{|\play|+1}}} (\tuple{\loc, b}, b)  \,;
		\end{displaymath}
              \item Otherwise, we let
                $ \convert^{\robmaxstrategy_\C}_{\game \to \C}(\play)
                = \convert^{\robmaxstrategy_\C}_{\game \to
                  \C}(\play_1) \xrightarrow{\trans_{a \to a}, \delay}
                (\tuple{\loc, a}, \val) $;
	\end{itemize}

      \item if $\last(\play_1) \in \ConfMin[\game]$, then, by letting
        $a$ the copy reached by
        $\convert^{\robmaxstrategy_\C}_{\game \to \C}(\play_1)$,
        $\val$ (resp. $\val_1^c$) the last valuation of $\play$
        (resp.
        $\convert^{\robmaxstrategy_\C}_{\game \to \C}(\play_1)$),
        $\guard$ the guard of $\trans$, and
        $b = \max(a, \leftg{\guard})$, we define
	\begin{displaymath}
		\convert_{\game \to \C}(\play) = \convert_{\game \to \C}(\play_1) \xrightarrow{\trans_{a \to b},
			\max(0, \val - \val_1^c)}
		s' 
	\end{displaymath}
	where
		\begin{displaymath}
			s' = \begin{cases}
					 (\tuple{\loc, b}, (\val + \max(0, \val - \val_1^c))[Y]) & \text{if $s = (\loc, \val)$} \\
					((\tuple{\loc, b}, \val), \trans_{a \to b}, \max(0, \val - \val_1^c)) & \text{if
						$s = ((\loc, \val), \trans, \delay)$} \\
			\end{cases}
		\end{displaymath}
	 where $Y$ is the set of reset of $\trans$.
            \end{itemize}

            By definition of $\C$, since $\tuple{\loc, a}$ cannot be
            reached with a valuation $\val < a$, and since decisions
            along $\convert_{\game \to \C}(\rho)$ are driven by
            $\robmaxstrategy_\C$, they are robust. Following the same
            proof structure as for the first inequality, we now state
            the analogous of \cref{lem:rVal_convert-C-G}.  However,
            this function preserves neither the last state nor the
            weight.  Thus, we need to adapt
            \cref{lem:rVal_convert-C-G}.

\begin{lemma}
	\label{lem:rVal_convert-G-C}
	Let $\play \in \rFPlays[\game]$. 
	Then,
	\begin{enumerate}
        \item\label{itm:rVal_convert-G-C_last-loc}
          $\convert_{\game \to \C}^{\robmaxstrategy_\C}$ preserves the
          location in the last state: if $\loc$ denotes the last
          location of $\play$, then there exists $a \in \N$ such that
          the last location of
          $\convert_{\game \to \C}^{\robmaxstrategy_\C}(\play)$ is
          $\tuple{\loc, a}$.
		
		\item\label{itm:rVal_convert-G-C_last-val}
		$\convert_{\game \to \C}^{\robmaxstrategy_\C}$ controls the valuation difference:
		by letling $\val$ (resp. $\val^\C$) be
		the last valuation of $\play$ (resp. of
		$\convert_{\game \to \C}^{\robmaxstrategy_\C}(\play)$).  Then,
		$0 \leq \val^\C - \val \leq \frac{\varepsilon}{W_{\loc}}(1 -
		2^{-|\play|})$.
		
		\item\label{itm:rVal_convert-G-C_weight} 
		$\convert_{\game \to \C}^{\robmaxstrategy_\C}$ controls the weight difference: 
		$\weightC(\play) \leq \weightC(\convert_{\game \to
		\C}^{\robmaxstrategy_\C}(\play)) + \varepsilon(1- 2^{-|\play|})|$.
	\end{enumerate}
\end{lemma}
\begin{proof}
	We prove each item successively.\\
	\textbf{\cref{itm:rVal_convert-G-C_last-loc}}
	is immediate by definition of $\convert_{\game \to \C}^{\robmaxstrategy_\C}$.
%
%

	\smallskip

	\noindent
	\textbf{Proof
	of~\cref{itm:rVal_convert-G-C_last-val}}:
	We reason by induction on the length of $\play$.  First, if
	$\play = (\tuple{\loc, 0}, \val)$, then
	$\convert_{\game \to \C}^{\robmaxstrategy_\C}(\play) = (\tuple{\loc,
		0}, \val) = \play$.  Thus, the property trivially holds.

	Suppose now that
	$\play = \play_1 \xrightarrow{\trans_{a \to b}, \delay} s$ with
	$\last(\play_1)$ a configuration of \MinPl.  In this case, writing
	$\val$ (resp. $\val^\C$, $\val_1$, and $\val_1^\C$) for the last
	valuation of $\play$ (resp.
	$\convert_{\game \to \C}^{\robmaxstrategy_\C}(\play)$, $\play_1$,
	and $\convert_{\game \to \C}^{\robmaxstrategy_\C}(\play_1)$) and $c$ for
	the copy reached by
	$\convert_{\game \to \C}^{\robmaxstrategy_\C}(\play_1)$, we have
	\begin{displaymath}
		\convert_{\game \to \C}^{\robmaxstrategy_\C}(\play) =
		\convert_{\game \to \C}^{\robmaxstrategy_\C}(\play_1) \xrightarrow{\trans_{a \to b}, \max(0, \val - \val_1^\C)} s'
	\end{displaymath}
	Moreover, if $\trans$ contains a reset, then $0 = \val -
	\val^\C$ since resets are preserved from $\game$ to $\C$.
	Otherwise, $\val = \val_1 + \delay$ and
	$\val^\C = \val_1^\C + \max(0, \val - \val_1^c)$, and
	$\val^\C - \val = \val_1^\C + \max(0, \val - \val_1^\C) - \val_1
	- \delay$.  In particular, since $\delay = \val - \val_{1}$
	and by definition of a maximum, we have
	\begin{displaymath}
		\val^\C - \val = \val_1^\C - \val  + \max(0, \val - \val_1^\C)
		\geq  \val_1^\C - \val  + \val - \val_1^\C = 0
	\end{displaymath}
	Conversely, by induction hypothesis applied to $\play_1$, we have
	\begin{displaymath}
		\val^\C - \val = \val_1^\C - \val_1 + \max(0, \val - \val_1^\C)   - \delay
		\leq \frac{\varepsilon}{W_{\loc}}(1 - 2^{-|\play_1|}) + \max(0, \val - \val_1^\C)  - \delay
	\end{displaymath}
	To conclude, we distinguish two cases according to the maximum:
	\begin{itemize}
		\item If $\max(0, \val - \val_1^\C) = 0$, then
		$\val^\C - \val \leq \frac{\varepsilon}{W_{\loc}}(1 - 2^{-|\play_1|}) - \delay \leq
		\frac{\varepsilon}{W_{\loc}}(1 - 2^{-|\play_1|})$,
		since $\delay \geq 0$;
		\item Otherwise, $\max(0, \val - \val_1^\C) = \val - \val_1^\C$, and
		since $\delay = \val - \val_1$, we have
		\begin{displaymath}
			\val^\C - \val \leq  \frac{\varepsilon}{W_{\loc}}(1 - 2^{-|\play_1|}) + \val - \val_1^\C - (\val - \val_1)
			= \frac{\varepsilon}{W_{\loc}}(1 - 2^{-|\play_1|}) + \val_1 - \val_1^\C \,.
		\end{displaymath}
		Applying the induction hypothesis to $\play_1$, we get
		that $\val_1 - \val_1^\C \leq 0$. Thus
		$\val^\C - \val \leq \frac{\varepsilon}{W_{\loc}}(1 -
		2^{-|\play_1|}) \leq \frac{\varepsilon}{W_{\loc}}(1 -
		2^{-|\play|})$.
	\end{itemize}

	Finally, suppose that
	$\play = \play_1 \xrightarrow{\trans, \delay} s$ with
	$\last(\play_1)$ a state of \MaxPl.  Let $\val$
	(resp. $\val_1$, $\val^\C$, and $\val_1^\C$) be the last
	valuation of $\play$ (resp. $\play_1$,
	$\convert_{\game \to \C}^{\robmaxstrategy_\C}(\play)$, and
	$\convert_{\game \to \C}^{\robmaxstrategy_\C}(\play_1)$).  We
	suppose that $\trans$ does not have a reset, otherwise
	$\val = \val^\C = 0$ and the property holds, 
	and so we distinguish three
	subcases according to the definition of
		$\convert_{\game \to \C}^{\robmaxstrategy_\C}$.
	\begin{itemize}
        \item First, we suppose that
          $\convert^{\robmaxstrategy_\C}_{\game \to \C}(\play_1)$
          conforms to $\robmaxstrategy_\C$, and
          $\robmaxstrategy_\C(\convert^{\robmaxstrategy_\C}_{\game \to
            \C}(\play_1)) = (\trans_{a \to b}^{(j)}, \delay)$. In
          particular,
		\begin{displaymath}
			\convert^{\robmaxstrategy_\C}_{\game \to \C}(\play) =
			\convert^{\robmaxstrategy_\C}_{\game \to \C}(\play_1)
			\xrightarrow{\trans_{a \to b}^{(j)}, \delay} (\tuple{\loc, b}, \val)  \,.
		\end{displaymath}
		Thus, $\val^\C - \val = \val_1^\C + \delay - \val_1 - \delay = \val_1^\C - \val_1$, and
		the property holds by applying the induction hypothesis to $\play_1$.

		\item Otherwise, we suppose that $\convert^{\robmaxstrategy_\C}_{\game \to \C}(\play_1)$
			conforms to $\robmaxstrategy_\C$, and
		$\robmaxstrategy_\C(\convert^{\robmaxstrategy_\C}_{\game \to \C}(\play_1)) =
		(\trans_{a \to b}^{j}, \delay + \frac{\varepsilon}{W_{\loc} \cdot 2^{|\play|+1}})$. In particular,
		\begin{displaymath}
			\convert^{\robmaxstrategy_\C}_{\game \to \C}(\play) =
			\convert^{\robmaxstrategy_\C}_{\game \to \C}(\play_1)
			\xrightarrow{\trans_{a \to b}^j, \delay + \frac{\varepsilon}{W_{\loc} \cdot 2^{|\play|+1}}} (\tuple{\loc, b}, b)  \,.
		\end{displaymath}
		In this case, applying the induction hypothesis to $\play_1$,
		we obtain
		\begin{displaymath}
			\val^\C - \val = \val_1^\C + \delay + \frac{\varepsilon}{W_{\loc} \cdot 2^{|\play|+1}} - \val_1 - \delay
			= \val_1^\C - \val + \frac{\varepsilon}{W_{\loc} \cdot 2^{|\play|+1}} \geq \frac{\varepsilon}{W_{\loc} \cdot 2^{|\play|+1}} \geq 0 \,.
		\end{displaymath}
		Conversely, by letting $\delay^\C = \delay + \frac{\varepsilon}{W_{\loc} \cdot 2^{|\play|+1}}$, we have
		\begin{displaymath}
			\val^\C - \val = \val_1^\C + \delay^\C  - \val_1 - \delay
			= \val_1^\C + \delay^\C  - \val - \delay^\C  + \frac{\varepsilon}{W_{\loc} \cdot 2^{|\play|+1}} =
			\val_1^\C - \val + \frac{\varepsilon}{W_{\loc} \cdot 2^{|\play|+1}}\,.
		\end{displaymath}
		Applying again the induction hypothesis to $\play_1$ and
		observing that
		$(1- 2 ^{-k}) = \sum_{i=0}^k \frac{1}{2^{k+1}}$, we obtain
		\begin{displaymath}
			\val^\C - \val
			\leq \frac{\varepsilon}{W_{\loc}}(1- 2 ^{-|\play_1|}) + \frac{\varepsilon}{W_{\loc} \cdot2^{|\play|+1}} =
			\frac{\varepsilon}{W_{\loc}}(1- 2 ^{-|\play|})\,.
		\end{displaymath}

              \item Finally, we suppose that
                $\convert^{\robmaxstrategy_\C}_{\game \to
                  \C}(\play_1)$ does not conform to
                $\robmaxstrategy_\C$. Then,
			\begin{displaymath}
				\convert^{\robmaxstrategy_\C}_{\game \to \C}(\play) =
				\convert^{\robmaxstrategy_\C}_{\game \to \C}(\play_1)
				\xrightarrow{\trans_{a \to a}, \delay} (\tuple{\loc, b}, \val)  \,.
			\end{displaymath}
			Thus, $\val^\C - \val = \val_1^\C + \delay - \val_1 - \delay = \val_1^\C - \val$, and
			the property holds by applying the induction hypothesis to $\play_1$.
	\end{itemize}

	\smallskip

	\noindent
	\textbf{Proof of~\cref{itm:rVal_convert-G-C_weight}}: We
        reason by induction on the length of $\play$.  First, we
        suppose that $\play = (\tuple{\loc, 0}, \val)$.  In this case,
        $\convert_{\game \to \C}^{\robmaxstrategy_\C}(\play) =
        (\tuple{\loc, 0}, \val) = \play$ and
        $\weightC(\play) = 0 = \weightC(\convert_{\game \to
          \C}^{\robmaxstrategy_\C}(\play))$.

	Assume now that
	$\play = \play_1 \xrightarrow{\trans_{a \to b}^{(j)}, \delay}
	s$ with $\tuple{\loc, b}$ the last location of $\play$.  Then
	\begin{displaymath}
		\convert_{\game \to \C}^{\robmaxstrategy_\C}(\play) =
		\convert_{\game \to \C}^{\robmaxstrategy_\C}(\play_1) \xrightarrow{\trans_{c \to d}, \delay^\C} s'
	\end{displaymath}
	where $\trans_{c \to d}$ is a copy of $\trans$ possibly in a
	new copy.  In this case, we have
	$\weightC(\play) = \weightC(\play_1) + \weight(\trans) +
	\delay \weight(\loc)$.  Since $\C$ preserves the weights of
	transitions between copies, by induction hypothesis
	applied to $\play_1$, we deduce that
	\begin{displaymath}
		\weightC(\play)
		\leq  \weightC(\convert_{\game \to \C}^{\robmaxstrategy_\C}(\play_1)) + \varepsilon(1- 2^{-|\play_1|}) +
		\weight(\trans_{a \to d}) + \delay \weight(\loc)
	\end{displaymath}
	Moreover, we remark, by definition of
	$\convert_{\game \to \C}^{\robmaxstrategy_\C}$
	and~\cref{itm:rVal_convert-G-C_last-val} applied to $\play$
	and $\play_1$, that
	\begin{displaymath}
		\delay^\C = \delay + \frac{\varepsilon}{W_\loc \cdot 2^{|\play|} + 1} \,.
	\end{displaymath}
	Thus,
	\begin{displaymath}
		\weightC(\play) \leq
		\weightC(\convert_{\game \to \C}^{\robmaxstrategy_\C}(\play_1)) + \varepsilon(1- 2 ^{-|\play_1|}) +
		\weight(\trans_{c \to d}) +  (\delay^\C + \frac{\varepsilon}{ W_\loc \cdot 2^{|\play|} + 1})~ \weight(\loc)
	\end{displaymath}
	Since, $\weight(\loc) \leq W_{loc}$ and,
        by~\cref{itm:rVal_convert-G-C_last-loc},
        $\convert_{\game \to \C}^{\robmaxstrategy_\C}$ preserves the
        last location (up to the copy index), we conclude that
	\begin{align*}
		\weightC(\play) &\leq  \weightC(\convert_{\game \to \C}^{\robmaxstrategy_\C}(\play)) +
		\varepsilon(1- 2^{-|\play_1|}) + \frac{\varepsilon}{2^{|\play|} + 1}\\
		&\leq \weightC(\convert_{\game \to \C}^{\robmaxstrategy_\C}(\play)) + \varepsilon(1- 2^{-|\play|})| \,. \qedhere
	\end{align*}
\end{proof}

\medskip
\noindent
\textbf{Definition of $\robmaxstrategy_\game$} Using the above-defined
mapping, we can define the robust strategy $\robmaxstrategy_{\game}$
of \MaxPl in $\game$.  It mimicks the choices of $\robmaxstrategy_\C$
in $\C$ ignoring the copy numbers.  Formally, for every
$\play \in \rFPlays[\game]$ ending in a state of \MaxPl, we define
\begin{displaymath}
	\robmaxstrategy_{\game}(\play) =
	\begin{cases}
		(\trans, \delay)
	&\text{if $\robmaxstrategy_{\C}(\convert_{\game \to \C}(\play)) = (\trans_{a \to b}^{(j)}, \delay)$ and
		$(\trans, \delay)$ is available in $\game$}  \\
	(\trans, \delay - \frac{\varepsilon}{W_{\loc} \cdot 2^{|\play| + 1}})
	&\text{if $\robmaxstrategy_{\C}(\convert_{\game \to \C}(\play)) = (\trans_{a \to b}^{j}, \delay)$ and
		$(\trans, \delay)$ is not available in $\game$}
	\end{cases}
\end{displaymath}
Observe that either the guard of $\trans_{a \to b}^{(j)}$ is included
into the one of $\trans$, or we apply a slight adaptation (when
$\guard \subset \overline{\guard}$).  Thus, robust choices of
$\maxstrategy_\C$ are also robust in $\game$, and this new strategy is
well-defined.  Interestingly, it fulfills the following properties:

\begin{lemma}
	\label{lem:rVal_convert-G-C_conf}
	Let $\play \in \rFPlays[\game]$ be a play that conforms to
        $\robmaxstrategy_{\game}$.  Then,
        $\convert_{\game \to \C}^{\robmaxstrategy_\C}(\play)$ conforms
        to $\robmaxstrategy_C$.
\end{lemma}
\begin{proof}
	The proof is by induction on the length of $\play$.
	If $\play = (\loc, \val)$, then $\convert_{\game \to \C}(\play) = (\tuple{\loc, 0}, \val)$
	and conforms to $\robmaxstrategy_\C$.
	Otherwise, we suppose that $\play = \play_1 \xrightarrow{\trans, \delay} s$, and
	we distinguish two cases according to the last state of $\play_1$.
	\begin{itemize}
        \item If $\last(\play_1)$ is a configuration of \MinPl, then
          the property holds by induction hypothesis applied to
          $\play_1$ and $\convert_{\game \to \C}(\play_1)$.

        \item If $\last(\play_1)$ is a state of \MaxPl, then by
          induction hypothesis applied to $\play_1$,
          $\convert_{\game \to \C}^{\robmaxstrategy_\C}(\play_1)$
          conforms to $\robmaxstrategy_C$.  We consider two subcases
          according to the definition of
          $\convert_{\game \to \C}^{\robmaxstrategy_\C}$.  First, in
          the case where
          $\robmaxstrategy_\C(\convert^{\robmaxstrategy_\C}_{\game \to
            \C}(\play_1)) = (\trans_{a \to b}^{(j)}, \delay)$ and
		\begin{displaymath}
			\convert^{\robmaxstrategy_\C}_{\game \to \C}(\play) =
			\convert^{\robmaxstrategy_\C}_{\game \to \C}(\play_1)
			\xrightarrow{\trans_{a \to b}^{(j)}, \delay} (\tuple{\loc, b}, \val) 
		\end{displaymath}
		the property holds by definition.

		Otherwise,
                $\robmaxstrategy_\C(\convert^{\robmaxstrategy_\C}_{\game
                  \to \C}(\play_1)) = (\trans_{a \to b}^{j}, \delay +
                \frac{\varepsilon}{W_{\loc} \cdot 2^{|\play|+1}})$,
                and
		\begin{displaymath}
			\convert^{\robmaxstrategy_\C}_{\game \to \C}(\play) =
			\convert^{\robmaxstrategy_\C}_{\game \to \C}(\play_1)
			\xrightarrow{\trans_{a \to b}^j, \delay + \frac{\varepsilon}{W_{\loc} \cdot 2^{|\play|+1}}} (\tuple{\loc, b}, b)  \,.
		\end{displaymath}
		Since $\play$ conforms to $\robmaxstrategy_{\game}$,
                and the guard of $\trans_{a \to b}^j$ is punctual
                (thus $(\delta, \delay)$ is not possible in $\game$),
                we deduce that
                $\delay + \frac{\varepsilon}{2^{|\play|+1}} = \delay^c
                - \frac{\varepsilon}{2^{|\play|+1}} +
                \frac{\varepsilon}{2^{|\play|+1}} $ where
                $\robmaxstrategy_{\C}(\convert_{\game \to \C}(\play))
                = (\trans_{a \to b}^{j}, \delay)$.  Thus, the property
                holds.  \qedhere
	\end{itemize}
\end{proof}

To conclude the proof of this second inequality,
\cref{itm:rVal_convert-G-C_weight} of~\cref{lem:rVal_convert-G-C} entails that for every $\play_\C$
from $(\tuple{\loc, 0}, \val)$ that conforms to $\robmaxstrategy_\C$,
one can build a play $\play_\game$ from $(\loc, \val)$ that conforms
to $\robmaxstrategy_\C$ and satisfies
$\weightC(\play) \leq \weightC(\convert_{\game \to \C}) +
\varepsilon(1 - 2^{-|\play|})$.  Indeed,
by~\cref{itm:rVal_convert-G-C_last-loc} of~\cref{lem:rVal_convert-G-C} the last location of
$\convert_{\C \to \game}(\play)$ is the same as the one of $\play$ and
thus belongs to the same player.  Thus,
\begin{displaymath}
	\weight(\convert_{\game \to \C}(\play)) =
	\begin{cases}
		+\infty & \text{if $\last(\play)$ is not a target} \\
		\weightC(\play) & \text{otherwise}
	\end{cases}
\end{displaymath}
which coincides with the definition of $\weight(\play)$.  Thus,
$\weight(\play) \leq \weight(\convert_{\game \to \C}) + \varepsilon(1
- 2^{-|\play|}) \leq \weight(\convert_{\game \to \C}) + \varepsilon$.
Finally,
\begin{displaymath}
	\inf_{\robminstrategy \in \rStratMin[\game]} \weight(\outcomes((\loc, \val), \robminstrategy, \robmaxstrategy_\game))
	\leq  \inf_{\robminstrategy \in \rStratMin[\C]} \weight(\outcomes((\loc, \val), \robminstrategy, \robmaxstrategy_\C)) + \varepsilon
	\leq \rValue[\C](\loc, \val) + \varepsilon
\end{displaymath}
Since this inequality holds for every
$\robmaxstrategy \in \rStratMax[\game]$ and every $\varepsilon > 0$,
we obtain  the claimed inequality
$\rValue[0+]_\game(\loc, \val) \geq \rValue[0+]_\C(\tuple{\loc, 0}, \val)$. \qed

\subsection[Proof of Corollary 14]{Proof of~\cref{cor:val_rval-Max-without-jump}}
\label{subsec:value_cor-without-jump}

We conclude this Appendix with the proof of the next corollary.

\maxWithoutJumping*

\begin{proof}
  In~\cref{app_preservation-rval}, given an $\varepsilon$-optimal
  robust strategy for \MaxPl in $\game$, we define
  $\robmaxstrategy_\C$ robust strategy for \MaxPl in $\C$, that only
  chooses non-jumping transitions, and such that the values of real
  states are at least as large as the original ones. More precisely
  \cref{eq:value_greater} holds: the value of the game is at least as
  large as the fixed-perturbation robust value of $\game$ up to an
  error of $\varepsilon$.  Moreover, by~\cref{prop:preservation_rval},
  robust values of $\game$ and $\C$ are equal.  Thus,
  $\robmaxstrategy_\C$ is $\varepsilon$-optimal in $\C$.
\end{proof}

%% file: app_copy-well-formed.tex
\section[Proof of Lemma 16]{Proof of~\cref{lem:rval-eq-val_restrictstrat}}
\label{sec:app_copy-well-formed}

\copyWellFormed*

Throughout the proof, we fix $\varepsilon >0$. Applying
\cref{cor:val_rval-Max-without-jump}, let
$\robmaxstrategy \in \rStratMax[\C]$ be an $\varepsilon/2$-optimal
robust strategy for \MaxPl such that all induced plays contain no
jumping transitions: that is every play $\play$ that conforms to
$\robmaxstrategy \in \rStratMax[\C]$ is such that
$\play \in \rFPlays_{\neg j}$.  To prove
\cref{lem:rval-eq-val_restrictstrat}, we define
$\widehat{\robmaxstrategy} \in \rStratMax[\C]$ a $p$-well-formed
strategy such that, for every real state $(\tuple{\loc, 0}, \val)$ of
$\C$, we have
\begin{equation}
	\label{eq:rval-eq-val_restrictstrat-ineq}
	\rValue[\perturbation, \robmaxstrategy](\tuple{\loc, 0}, \val) \leq
	\rValue[\perturbation, \widehat{\robmaxstrategy}](\tuple{\loc, 0}, \val) + \varepsilon/2 \,.
\end{equation}
From the fact that $\robmaxstrategy$ is an $\varepsilon/2$-optimal
robust strategy for \MaxPl, we will deduce that for every real state
$(\tuple{\loc, 0}, \val)$
\begin{displaymath}
	\rValue(\tuple{\loc, 0}, \val)-  \varepsilon/2 \leq \rValue[\perturbation, \robmaxstrategy](\tuple{\loc, 0}, \val) \leq
	\rValue[\perturbation, \widehat{\robmaxstrategy}](\tuple{\loc, 0}, \val) + \varepsilon/2 \,.
\end{displaymath}
As a consequence,
$\rValue(\tuple{\loc, 0}, \val)- \varepsilon \leq
\rValue[\perturbation, \widehat{\robmaxstrategy}](\tuple{\loc, 0},
\val)$ and $\widehat{\robmaxstrategy}$ will be $\varepsilon$-optimal.

\medskip
\noindent
\textbf{Definition of $\convert_{\wf \to \rmath}$.}  Before defining
$\widehat{\robmaxstrategy}$, we start with the definition of a mapping
that maps the jumping transitions chosen by \MaxPl in an arbitrary
strategy to decisions conforming to $\robmaxstrategy$ that are ``too
close'' to the right border of a guard.  This mapping will be driven
by strategy $\robmaxstrategy$.  Formally, we define a function
$\convert_{\wf \to \rmath}^{\robmaxstrategy} : \rFPlays \to
\rFPlays_{\neg j}$ by induction on the length of plays as follows.
For every real state $s$, we let
$\convert_{\wf \to \rmath}^{\robmaxstrategy}(s) = s$. For every play
$\play \in \rFPlays$ such that
$\play = \play_1 \xrightarrow{\trans_{a \to b}^{(j)}, \delay} s$, we
distinguish two cases according to the last state of $\play_1$:
\begin{displaymath}
	\convert_{\wf \to \rmath}^{\robmaxstrategy}(\play) =
	\begin{cases}
		\convert_{\wf \to \rmath}^{\robmaxstrategy}(\play_1) \xrightarrow{\trans_{c \to b}, \max(0, \val - \val_1^\C)} s'
		& \text{if $\last(\play_1) \in \ConfMin$} \\
		\convert_{\wf \to \rmath}^{\robmaxstrategy}(\play_1) \xrightarrow{\robmaxstrategy(\convert_{\wf \to \rmath}^{\robmaxstrategy}(\play_1))} s'
		& \text{otherwise}
	\end{cases}
\end{displaymath}
where $\val_1^\C$ (resp. $\val$) is the last valuation of
$\convert_{\wf \to \rmath}^{\robmaxstrategy}(\play_1)$ (resp. $\play$) and
$c$ is the copy reached by
$\convert_{\wf \to \rmath}^{\robmaxstrategy}(\play_1)$ (that is, some
integer in $\llbracket a, b \rrbracket \cup \{0\}$).  The intuition is
to compensate the potential difference in weight by altering the
delays.

Observe that $\convert_{\wf \to \rmath}^{\robmaxstrategy}$ only picks robust
decisions for \MinPl and \MaxPl.  Indeed, decisions of \MaxPl are
robust by definition of $\robmaxstrategy$.  For decisions of \MinPl,
we observe that the final valuation reached by
$\convert_{\wf \to \rmath}^{\robmaxstrategy}$ is always least or equal to
the one reached by $\play$.  Thus, since $\play$ is robust, the
decision of \MinPl by $\convert_{\wf \to \rmath}^{\robmaxstrategy}$ is also
robust.

\medskip
\noindent
\textbf{Definition of $\widehat{\robmaxstrategy}$.}  We are now in a
position to define a well-formed strategy for \MaxPl.  Formally, for
every robust play $\play$, writing $\val$ (resp. $\val^\C$) for the
last valuation of $\play$ (resp. of
$\convert_{\wf \to \rmath}^{\robmaxstrategy}(\play)$), and
$(\trans_{a \to b}, \delay^\C) = \robmaxstrategy(\convert_{\wf \to
  r}^{\robmaxstrategy}(\play))$ for the move given by
$\robmaxstrategy$ (note that it cannot be a jumping transition,
therefore $b \in \{0, a\}$), we define
\begin{displaymath}
  \widehat{\robmaxstrategy}(\play) = 
	\begin{cases}
		(\trans_{c \to i}^j, i - \val) & \text{if $b = a$ and $\val + \delay \in [i - N \cdot 2p, i)$ for $i \in \N^*$}; \\
	(\trans_{c \to b}, \delay^\C) & \text{else and if ($c = a$ or $\trans$  resets \clockx);} \\
		(\trans_{c \to c}, \max(\delay^\C - (\val - \val^\C)), 0) & \text{otherwise}.
	\end{cases}
\end{displaymath}
where $c$ denotes the copy reached by $\play$ (that may differ from
the copy reached by $\convert_{\wf \to
  r}^{\robmaxstrategy}(\play)$). Note that, the last two cases of this
definition correspond to a decision of $\robmaxstrategy$ that is
$p$-well-formed, i.e.  if $\last(\play) \notin \ConfMax$,
$\robmaxstrategy(\play) = (\trans^j_{a \to b}, \delay)$,
$\robmaxstrategy(\play) = (\trans_{a \to 0}, \delay)$, or
$\robmaxstrategy(\play) = (\trans_{a \to a}, \delay)$ such that for
all $i \in \N^*$, $\val + \delay \notin [i- N \cdot 2p, i)$.

Clearly enough $\widehat{\robmaxstrategy}$ is a robust strategy for
\MaxPl. Indeed the delay chosen for jumping transitions satisfies the
equality guard by definition, and in other cases, the delay is at most
the delay chosen by $\robmaxstrategy$, which is robust. Moreover,
$\widehat{\robmaxstrategy}$ is by construction a $p$-well-formed
strategy.

\medskip

To complete the proof, it remains to
prove~\cref{eq:rval-eq-val_restrictstrat-ineq}. We will show that for
every play $\play$ that conforms to $\widehat{\robmaxstrategy}$, there
exists a play $\play'$ that conforms to $\robmaxstrategy$ and such
that $\weight(\play') \leq \weight(\play) + \varepsilon/2$. To do so,
it suffices to show that the function
$\convert_{\wf \to \rmath}^{\robmaxstrategy}$ has good properties when it
is applied to a play that conforms to
$\widehat{\robmaxstrategy}$. Precisely:
\begin{lemma}
	\label{lem:rval-eq-val_restrictstrat-convert}
	Let $\play$ be a finite play that conforms to
        $\widehat{\robmaxstrategy}$.  Then,
	\begin{enumerate}
        \item\label{itm:rval-eq-val_restrictstrat_conf}
          $\convert_{\wf \to \rmath}^{\robmaxstrategy}(\play)$ conforms to
          $\robmaxstrategy$;
              \item\label{itm:rval-eq-val_restrictstrat_last-loc}
                $\convert_{\wf \to \rmath}^{\robmaxstrategy}$ preserves the
                last location, i.e. if
                $\tuple{\loc, a}$ is the last location of $\play$ in
                $\C$, then there exists $b \geq a$ such that
                $\tuple{\loc, b}$ is the last location of
                $\convert_{\wf \to \rmath}^{\robmaxstrategy}(\play)$;
              \item\label{itm:rval-eq-val_restrictstrat_last-val} the
                valuation deviation of
                $\convert_{\wf \to \rmath}^{\robmaxstrategy}$ is
                controlled, i.e. if $\val$ is the last valuation of
                $\play$, the last valuation $\val^\C$ of
                $\convert_{\wf \to \rmath}^{\robmaxstrategy}(\play)$ is
                such that $0 \leq \val^\C - \val \leq N \cdot 2 p$;
              \item\label{itm:rval-eq-val_restrictstrat_weight} the
                weight difference is controlled, i.e.
                $\weightC(\convert_{\wf \to
                  r}^{\robmaxstrategy}(\play)) \leq \weightC(\play) + \maxWeightLoc \cdot
                4N \cdot \perturbation \cdot |\play|$.
	\end{enumerate}
\end{lemma}
\begin{proof}
  We prove each item successively.\\
	\textbf{Proof of~\cref{itm:rval-eq-val_restrictstrat_conf}}:
	By definition of $\convert_{\wf \to \rmath}$, every decision of \MaxPl follows $\robmaxstrategy$.

	\smallskip

	\noindent
	\textbf{Proof
          of~\cref{itm:rval-eq-val_restrictstrat_last-loc}}: The proof
        is by induction on the length of $\play$.  For the base case,
        suppose that $\play = (\tuple{\loc, 0}, \val)$.
        Then,
        $\convert_{\wf \to \rmath}^{\robmaxstrategy}(\play) = (\tuple{\loc,
          0}, \val) = \play$ and the property holds.

	Suppose now that
        $\play = \play_1 \xrightarrow{\trans_{a \to b}, \delay} s$
        with $\last(\play_1)$ a state of \MinPl.  In this case,
        writing $\val_1^\C$ (resp. $\val$) for the last valuation of
        $\convert_{\wf \to \rmath}^{\robmaxstrategy}(\play_1)$
        (resp. $\play$) and $c$ for the copy reached by the play
        $\convert_{\wf \to \rmath}^{\robmaxstrategy}(\play_1)$, we have
	\begin{displaymath}
		\convert_{\wf \to \rmath}^{\robmaxstrategy}(\play) =
		\convert_{\wf \to \rmath}^{\robmaxstrategy}(\play_1) \xrightarrow{\trans_{c \to b}, \max(0, \val - \val_1^\C)} s'
	\end{displaymath}
	By definition, $\play$ and
        $\convert_{\wf \to \rmath}^{\robmaxstrategy}(\play)$ use a copy of
        the same underlying transition $\delta$ (possibly in an other
        copy).  Applying the induction hypothesis to $\play_1$ allows
        us to conclude.

	Finally, suppose that
        $\play = \play_1 \xrightarrow{\trans_{a \to b}^{(j)}, \delay}
        s$ with $\last(\play_1)$ a state of \MaxPl.  Then
	\begin{displaymath}
		\convert_{\wf \to \rmath}^{\robmaxstrategy}(\play) =
		\convert_{\wf \to \rmath}^{\robmaxstrategy}(\play_1)
		\xrightarrow{\robmaxstrategy(\convert_{\wf \to \rmath}^{\robmaxstrategy}(\play_1))} s'
	\end{displaymath}
	Since, $\play$ conforms to $\widehat{\robmaxstrategy}$, by
        definition of $\widehat{\robmaxstrategy}$,
        $\widehat{\robmaxstrategy}(\play_1) = (\trans_{a \to b}^{(j)},
        \delay)$ with
        $\robmaxstrategy(\convert_{\wf \to \rmath}(\play_1)) = (\trans_{c,
          d}, t')$.  In particular, $\play$ and
        $\convert_{\wf \to \rmath}^{\robmaxstrategy}(\play)$ use the same
        undelying transition (possibly in different copies).  Here
        again applying the induction hypothesis to $\play_1$ allows
        one to conclude.

	\smallskip

	\noindent
	\textbf{Proof of~\cref{itm:rval-eq-val_restrictstrat_last-val}}:
	Here also, the proof is by induction on the length of $\play$.
        
	If $\play = (\tuple{\loc, 0}, \val)$, 
	then $\convert_{\wf \to \rmath}^{\robmaxstrategy}(\play) = (\tuple{\loc, 0}, \val) = \play$
	and the base case holds.

	Assume now that
        $\play = \play_1 \xrightarrow{\trans_{a \to b}, \delay} s$
        with $\last(\play_1)$ a state of \MinPl.  In this case,
        writing $\val$ (resp. $\val^\C$, $\val_1$, and $\val_1^\C$) for
        the last valuation of $\play$ (resp.
        $\convert_{\wf \to \rmath}^{\robmaxstrategy}(\play)$, $\play_1$, and
        $\convert_{\wf \to \rmath}^{\robmaxstrategy}(\play_1)$) and $c$ for
        the copy reached by
        $\convert_{\wf \to \rmath}^{\robmaxstrategy}(\play_1)$, we have
	\begin{displaymath}
		\convert_{\wf \to \rmath}^{\robmaxstrategy}(\play) =
		\convert_{\wf \to \rmath}^{\robmaxstrategy}(\play_1) \xrightarrow{\trans_{c \to b}, \max(0, \val - \val_1^\C)} s' \enspace.
	\end{displaymath}
	Moreover, if transition $\trans$ resets $\clockx$, then
        $0 = \val - \val^\C$ since resets are preserved in $\C$.
        Otherwise, $\val = \val_1 + \delay$ and
        $\val^\C = \val_1^\C + \max(0, \val - \val_1^\C)$, and
        $\val^\C - \val = \val_1^\C + \max(0, \val - \val_1^\C) -
        \val_1 - \delay$.  In particular, since
        $\delay = \val - \val_{1}$ and by definition of a maximum, we
        have
	\begin{displaymath}
		\val^\C - \val\quad =\quad \val_1^\C - \val  + \max(0, \val - \val_1^\C)
		\quad \geq\quad   \val_1^\C - \val  + \val - \val_1^\C = 0
	\end{displaymath}
	Conversely, the induction hypothesis applied to $\play_1$, implies
	\begin{displaymath}
		\val^\C - \val \quad =\quad  \val_1^\C - \val_1 + \max(0, \val - \val_1^\C)   - \delay
		\quad \leq\quad  2 \cdot N \cdot \perturbation + \max(0, \val - \val_1^\C)  - \delay
	\end{displaymath}
	To conclude, we distinguish two cases depending the maximum:
	\begin{itemize}
		\item If $\max(0, \val - \val_1^\C) = 0$, then
		$\val^\C - \val \leq 2 \cdot N \cdot \perturbation  - \delay \leq 2 \cdot N \cdot \perturbation$,
		since $\delay \geq 0$;
		\item Otherwise, $\max(0, \val - \val_1^\C) = \val - \val_1^\C$, and
		since $\delay = \val - \val_1$, we have
		\begin{displaymath}
			\val^\C - \val \leq  2 \cdot N \cdot \perturbation + \val - \val_1^\C - (\val - \val_1)
			= 2 \cdot N \cdot \perturbation + \val_1 - \val_1^\C \,.
		\end{displaymath}
		The induction hypothesis applied to $\play_1$ entails
                that $\val_1 - \val_1^\C \leq 0$, thus
                $\val^\C - \val \leq 2 \cdot N \cdot \perturbation$.
	\end{itemize}

	Finally, assume that
        $\play = \play_1 \xrightarrow{\trans_{a \to b}^{(j)}, \delay}
        s$ with $\last(\play_1)$ a state of \MaxPl.  In this case,
        letting
        $(\trans_{a \to b}, \delay^\C) = \robmaxstrategy(\convert_{\wf
          \to r}^{\robmaxstrategy}(\play_1))$, we have
	\begin{displaymath}
		\convert_{\wf \to \rmath}^{\robmaxstrategy}(\play) =
		\convert_{\wf \to \rmath}^{\robmaxstrategy}(\play_1)
		\xrightarrow{(\trans_{a \to b}, \delay^\C)} s'  \,.
	\end{displaymath}
	As previouly, the non-trivial case is when $\trans$ does not
        reset $\clockx$.  Since $\play$ conforms to
        $\widehat{\robmaxstrategy}$, we distinguish three cases
        according to the definition of $\widehat{\robmaxstrategy}$.
	\begin{itemize}
		\item If $\widehat{\robmaxstrategy}(\play_1) = (\trans_{a \to i}^j, i - \val_1)$, then
		by definition, $\val_1 + \delay^\C \in [i - N \cdot 2 \perturbation, i)$ for $i \in \N^*$.
		In particular, since $\delay  = i - \val_1$ and
		$\delay^\C \geq i - N \cdot 2\perturbation - \val_1$,, we have
		\begin{displaymath}
			\val^\C - \val \ =\ \val_1^\C + \delay^\C - \val_1 - \delay
			\ =\ \val_1^\C + \delay^\C - i \ \geq\ \val_1^\C +  - N \cdot 2\perturbation - \val_1 \,.
		\end{displaymath}
		Because $N$ and $\perturbation$ are positive, by
                applying the induction hypothesis to $\play_1$, we
                obtain $\val_1^\C - \val_1 \geq 0$, hence
                $\val^\C - \val \geq 0$.

		However, also applying the induction hypothesis to
                $\play_1$, we obtain
		\begin{displaymath}
			\val^\C - \val = \val_1^\C + \delay^\C - \val_1 - \delay
			\leq N \cdot 2\perturbation + \delay^\C - \delay \,.
		\end{displaymath}
		Moreover, since $\delay^\C < i - \val_1$ and $\delay = \val_1 - i$, we deduce that
		\begin{displaymath}
			\val^\C - \val < N \cdot 2\perturbation + i - \val_1 - \val_1 - i =  N \cdot 2\perturbation \,
                      \end{displaymath}
                      which is the desired result.

                    \item If
                      $\widehat{\robmaxstrategy}(\play_1) = (\trans_{c
                        \to b}, \delay^\C)$ then $\delay = \delay^\C$
                      and
                      $\val^\C - \val = \val_1^\C + \delay^\C - \val_1 -
                      \delay = \val_1^\C - \val_1$.  Thus, we can
                      conclude by applying the induction hypothesis to
                      $\play_1$.

                    \item Last, if
                      $\widehat{\robmaxstrategy}(\play_1) = (\trans_{c
                        \to c}, \max(\delay^\C - (\val_1 - \val_1^\C),
                      0))$, then we distinguish two cases according
                      $\delay$.  First, if
                      $\delay = \delay^\C - (\val_1 - \val_1^\C)$, then
                      the property trivially holds since
		\begin{displaymath}
			\val^\C - \val = \val_1^\C + \delay^\C - \val_1 - \delay
			= \val_1^\C + \delay^\C - \val_1 - \delay^\C + (\val_1 - \val_1^\C)
			= 0 \,.
		\end{displaymath}
		Otherwise, $\delay = 0$ and, applying the induction
                hypothesis to $\play_1$, we have
                $\val^\C - \val = \val_1^\C + \delay^\C - \val_1 \geq
                \delay^\C \geq 0$.  Conversely, by applying again the
                induction hypothesis to $\play_1$, we have
                $\val^\C - \val = \val_1^\C + \delay^\C - \val_1 \leq N
                \cdot 2 \perturbation + \delay^\C$.  Thus, by
                definition of the maximum,
                $\delay^\C - (\val_1 - \val_1^\C) < 0$, and we deduce
                that
                $\val^\C - \val \leq N \cdot 2 \perturbation + (\val_1
                - \val_1^\C)$.  We can now conclude since
                $\val_1 - \val_1^\C \leq 0$.
	\end{itemize}

	\smallskip

	\noindent
	\textbf{Proof of~\cref{itm:rval-eq-val_restrictstrat_weight}}:
	Here again, the proof is by induction on the length of $\play$.
        
	Assume first that $\play = (\tuple{\loc, 0}, \val)$. Then 
	 $\convert_{\wf \to \rmath}^{\robmaxstrategy}(\play) = (\tuple{\loc, 0}, \val) = \play$
	and $\weightC(\play) = 0 = \weightC(\convert_{\wf \to \rmath}^{\robmaxstrategy}(\play))$.

	Suppose now that
        $\play = \play_1 \xrightarrow{\trans_{a \to b}^{(j)}, \delay}
        s$ with $\tuple{\loc, b}$ the last location of $\play$.  In
        this case, we have
	\begin{displaymath}
		\convert_{\wf \to \rmath}^{\robmaxstrategy}(\play) =
		\convert_{\wf \to \rmath}^{\robmaxstrategy}(\play_1) \xrightarrow{\trans_{c \to d}, \delay^\C} s'
	\end{displaymath}
	Thus,
        $\weightC(\play) = \weightC(\play_1) + \weight(\trans_{a \to
          b}) + \delay \weight(\loc)$. Since the construction of $\C$
        preserves the weight of transition among different copies, we
        deduce, by induction hypothesis applied to $\play_1$, that
	\begin{displaymath}
		\weightC(\play)
		\geq  \weightC(\convert_{\wf \to \rmath}^{\robmaxstrategy}(\play_1)) - \maxWeightLoc \cdot 4N \cdot \perturbation \cdot |\play_1| +
		\weight(\trans_{c \to d}) + \delay \cdot \weight(\loc)
	\end{displaymath}
	Moreover, we remark, by definition of
        $\convert_{\wf \to \rmath}^{\robmaxstrategy}$
        and~\cref{itm:rval-eq-val_restrictstrat_last-val} applied to
        $\play$ and $\play_1$, that
	\begin{displaymath}
		\delay^\C = \val^\C - \val_1^\C = (\val^\C -\val) + (\val - \val_1) + (\val_1 - \val_1^\C)
		\leq N \cdot 2 \perturbation + (\val - \val_1) = \delay + N \cdot 2 \perturbation \,.
	\end{displaymath}
	Thus,
	\begin{displaymath}
		\weightC(\play) \geq
		\weightC(\convert_{\wf \to \rmath}^{\robmaxstrategy}(\play_1)) - \maxWeightLoc \cdot 4N \cdot \perturbation \cdot |\play_1| +
		\weight(\trans_{c \to d}) +  (\delay^\C - 2 N \cdot \perturbation)~ \weight(\loc)
	\end{displaymath}
	Since, $\weight(\loc) \leq \maxWeightLoc$ and,
        by~\cref{itm:rval-eq-val_restrictstrat_last-loc},
        $\convert_{\wf \to \rmath}^{\robmaxstrategy}$ preserves the last
        location of $\game$ (possibly with a different copy index), we
        conclude that
	\begin{align*}
		\weightC(\play) &\geq  \weightC(\convert_{\wf \to \rmath}^{\robmaxstrategy}(\play)) -
		\maxWeightLoc \cdot 4N \cdot \perturbation \cdot |\play_1| - (2 N \cdot \perturbation) \cdot \weight(\loc) \\
		&\geq \weightC(\convert_{\wf \to \rmath}^{\robmaxstrategy}(\play)) - \maxWeightLoc
		\cdot 4N \cdot \perturbation \cdot |\play| \,.
	\end{align*}
	This finishes the proof of~\cref{lem:rval-eq-val_restrictstrat-convert}.
\end{proof}

We are now in a position to, given $\play$ a play that conforms to
$\widehat{\robmaxstrategy}$, define a play $\play'$ that conforms to
$\robmaxstrategy$ and such that
$\weight(\play') \leq \weight(\play) + \varepsilon/2$, thus
establishing~\cref{eq:rval-eq-val_restrictstrat-ineq}.

For $\play$ be a play that conforms to $\widehat{\robmaxstrategy}$, we
let $\play' = \convert_{\wf \to \rmath}^{\robmaxstrategy}(\play)$.
By~\cref{itm:rval-eq-val_restrictstrat_conf}
of~\cref{lem:rval-eq-val_restrictstrat-convert}, $\play'$ does conform
to $\robmaxstrategy$.  Moreover,
applying~\cref{itm:rval-eq-val_restrictstrat_weight}
of~\cref{lem:rval-eq-val_restrictstrat-convert} we derive
\begin{align*}
\weightC(\convert_{\wf \to \rmath}^{\robmaxstrategy}(\play))   &\leq \weightC(\play) + 
	\maxWeightLoc \cdot 4N \cdot \perturbation \cdot |\play| \\
  &\leq \weightC(\play) +
	\maxWeightLoc \cdot 4N \cdot \perturbationBound \cdot \kappa \quad\quad\quad \text{ since $\perturbation \leq \perturbationBound$ and $|\play| \leq \kappa$} \\
	&\leq \weightC(\play) + 
   \dfrac{\maxWeightLoc \cdot 4N  \cdot \varepsilon \cdot \kappa}{8 N \cdot \maxWeightLoc \cdot \kappa} \quad\quad\quad \text{ since $\perturbationBound < \frac{\varepsilon}{8 N \cdot \maxWeightLoc \cdot \kappa}$}  \\
	&\leq \weightC(\play) +  \varepsilon/2
\end{align*}
Finally, \cref{itm:rval-eq-val_restrictstrat_last-loc}
of~\cref{lem:rval-eq-val_restrictstrat-convert}, yields
$\weight(\convert_{\wf \to r}^{\robmaxstrategy}(\play)) \leq
\weight(\play) + \varepsilon/2$ since either $\play$ and
$\convert_{\wf \to \rmath}^{\robmaxstrategy}(\play)$ both reach a
target state, or none of them reach a target.  That concludes the
proof of~\cref{eq:rval-eq-val_restrictstrat-ineq} and thus of
\cref{lem:rval-eq-val_restrictstrat}. \qed

%% file: app_copy_e-to-wf.tex
\section[Proof of Lemma 17]{Proof of~\cref{lem:copy_Min_rstrict}}
\label{app:copy_e-to-wf}

\copyMin*

The proof of this result uses the same techniques as
for the proof of~\cref{lem:rval-eq-val_restrictstrat}.

\medskip

Let $\minstrategy$ be an $\varepsilon/2$-optimal (exact) strategy for
\MinPl in $\C$.  Our aim is to define a $\perturbation$-well-formed
$\varepsilon$-optimal (exact) strategy $\minstrategy_{\wf}$ for
\MinPl.  Intuitively, this strategy will mimick $\minstrategy$ and
only modify delays that would fall in a forbidden interval and thus
are not $\perturbation$-well-formed.

Towards a formal definition, we use a mapping that transforms
$\perturbation$-well-formed plays into exact plays. Through this
mapping, \MinPl compensates the difference induced by the choice of
$\minstrategy_{\wf}$.  Formally, we inductively define
$\convert_{\wf \to \emath}^\minstrategy \colon \FPlays_{\wf} \to
\FPlays$ such that, for every configuration $(\tuple{\loc, a}, \val)$,
we fix
$\convert_{\wf \to \emath}^\minstrategy(\tuple{\loc, a}, \val) =
(\tuple{\loc, a}, \val)$, and for every play
$\play = \play_1 \xrightarrow{\trans, \delay} (\tuple{\loc, a},
\val)$, by letting $(\tuple{\loc_1, b}, \val_1^\C) = \last(\play_1)$
(resp.
$(\tuple{\loc_1, b}, \val_1) = \last(\convert_{\wf \to
  \emath}^\minstrategy(\play_1))$), we fix
$\convert_{\wf \to \emath}^\minstrategy(\play)$ to be equal to
\begin{displaymath}
	\begin{cases}
		\convert_{\wf \to \emath}^\minstrategy(\play_1) \xrightarrow{\trans, \delay} (\tuple{\loc, a}, \val) & \text{if $\trans$  resets } \clockx \\
		\convert_{\wf \to \emath}^\minstrategy(\play_1) \xrightarrow{\trans, \delay^\C} (\tuple{\loc, a}, \val^\C) & \text{if $\loc_1 \in \LocsMin$ and
			$\minstrategy(\convert_{\wf \to \emath}^\minstrategy(\play_1)) = (\trans, \delay^\C)$} \\
		\convert_{\wf \to \emath}^\minstrategy(\play_1) \xrightarrow{\trans, \val - \val_1^\C} (\tuple{\loc, a}, \val) & \text{otherwise}
	\end{cases}
      \end{displaymath}
      Observe that all decisions of $\convert_{\wf \to \emath}$ are
      indeed enabled, since they are induced by transitions in
      $\play$, defined by decisions of $\minstrategy$, or yield to the
      same valuation as $\play$.

Using this mapping, we can define the $\perturbation$-well-formed
strategy $\minstrategy_{\wf}$ for \MinPl.  Precisely, for every
$\perturbation$-well-formed play $\play$ ending in $(\loc, \val)$  a
configuration of \MinPl, and such that
$\minstrategy(\convert_{\wf \to \emath}^\minstrategy(\play)) = (\trans,
\delay)$, we fix
\begin{displaymath}
	\minstrategy_{\wf}(\play) =
	\begin{cases}
		(\trans, \delay) & \text{if $\val + \delay \notin \cup_{i \in \N} F_i$} \\
		(\trans, \max(0, i - N\cdot \perturbation)) & \text{if there exists $i \in \N$ such that $\val + \delay \in F_i$}
	\end{cases}
\end{displaymath}
In other cases, i.e.\ if $\play$ is not a $\perturbation$-well-formed
play, the strategy is undefined.  Here again, all decisions taken by
$\minstrategy_{\wf}$ are enabled: either the decision agrees with one
of $\minstrategy$, or its delay is non-negative and at most equal to
the delay chosen by $\minstrategy$.  Moreover, by definition, all
plays that conform to $\minstrategy_{\wf}$ are
$\perturbation$-well-formed.

To finalise the proof, we show that for evert maximal play $\play$,
starting in configuration $(\loc, \val)$ and that conforms to
$\minstrategy_{\wf}$, we have
$\weight(\play) \leq \Value(\loc, \val) + \varepsilon$. To do so, we
observe that when $\play$ conforms to $\minstrategy_{\wf}$, the mapping
$\convert_{\wf \to \emath}^\minstrategy$ has the following good
properties:
\begin{lemma}
	\label{lem:copy_e-to-wf_convert}
	Let $\play$ be a play that  conforms to $\minstrategy_{\wf}$.
	Letting $(\tuple{\loc, a}, \val) =  \last(\play)$, we have
	\begin{enumerate}
	    \item\label{itm:copy_e-to-wf_convert-conf}
		$\convert_{\wf \to \emath}^\minstrategy(\play)$ conforms to $\minstrategy$;
		
              \item\label{itm:copy_e-to-wf_convert-last-loc}
                $\convert_{\wf \to \emath}^\minstrategy$ preserves the
                last location, i.e. the last location of
                $\convert_{\wf \to \emath}^\minstrategy(\play)$ is
                $\tuple{\loc, a}$;

              \item\label{itm:copy_e-to-wf_convert-last-val} the
                valuation deviation of
                $\convert_{\wf \to \emath}^\minstrategy$ is
                controlled. Formally, if $\val^{c}$ is last valuation
                of $\convert_{\wf \to \emath}^\minstrategy(\play)$,
                then
                $0 \leq \val^\C - \val \leq N \cdot \perturbation$;

              \item\label{itm:copy_e-to-wf_convert-weight} the weight
                difference is controlled: i.e.\
                $\weightC(\play) \leq \weightC(\convert_{wf \to
                  e}^\minstrategy(\play)) + N \cdot \perturbation
                \cdot \maxWeightLoc \cdot |\play|$.
	\end{enumerate}
      \end{lemma}
\begin{proof}
  We prove each item successively.\\
  
	\noindent \textbf{Proof of~\cref{itm:copy_e-to-wf_convert-conf}}:
	The proof is by induction on the length of $\play$.
        
	If $\play = (\tuple{\loc, a}, \val)$, then
        $\convert_{\wf \to \emath}^\minstrategy(\play) = (\tuple{\loc,
          a},\val)$ and the property trivially holds.

	Otherwise, assume that
        $\play = \play_1 \xrightarrow{\trans, \delay} (\tuple{\loc,
          a}, \val)$.  We distinguish two cases according to which
        player owns the last location of $\play_1$.
	\begin{itemize}
        \item If $\play_1$ ends in a configuration of \MaxPl, then the
          property holds by induction hypothesis applied to $\play_1$.
        \item Otherwise, $\play_1$ ends in a configuration of \MinPl.
          Now, by induction hypothesis applied to $\play_1$, we know
          that $\convert_{\wf \to \emath}^\minstrategy(\play_1)$
          conforms to $\minstrategy$.  Moreover, since $\play$
          conforms to $\minstrategy_{\wf}$, we have that
          $\minstrategy(\convert_{\wf \to
            \emath}^\minstrategy(\play_1)) = (\trans, \delay^\C)$ and
          $\convert_{\wf \to \emath}^\minstrategy(\play) =
          \convert_{wf \to e}^\minstrategy(\play_1)
          \xrightarrow{\trans, \delay^\C} (\tuple{\loc, a}, \val^\C)$.
          Thus, $\convert_{\wf \to \emath}^\minstrategy(\play)$
          conforms to $\minstrategy$.
	\end{itemize}

	\noindent\textbf{Proof of~\cref{itm:copy_e-to-wf_convert-last-loc}}:
	Since $\convert_{\wf \to \emath}^\minstrategy$ preserves the
        sequence of original transitions (i.e. transitions from the
        underlying path), $\play$ and
        $\convert_{\wf \to \emath}^\minstrategy(\play)$ end in the
        same location.

	\medskip
	\noindent\textbf{Proof of~\cref{itm:copy_e-to-wf_convert-last-val}}:
	The proof is also here by induction on the length of $\play$.
        
	If $\play = (\tuple{\loc, a}, \val)$, then
        $\convert_{\wf \to \emath}^\minstrategy(\play) = (\tuple{\loc,
          a},\val)$ and the property trivially holds.

	Otherwise, assume that
        $\play = \play_1 \xrightarrow{\trans, \delay} (\tuple{\loc,
          a}, \val)$.  We write $\val_1$ (resp. $\val_1^\C$) for the
        last valuation of $\play_1$ (resp. of
        $\convert_{\wf \to \emath}^\minstrategy(\play_1)$), and we
        distinguish three cases according to the definition of the
        mapping $\convert_{\wf \to \emath}^\minstrategy$.
	\begin{itemize}
        \item If $\trans$ resets $\clockx$, then
          $\convert_{\wf \to \emath}^\minstrategy(\play) = \convert_{wf \to
            e}^\minstrategy(\play_1) \xrightarrow{\trans, \delay}
          (\tuple{\loc, a}, \val)$ and $\val = \val^\C = 0$.  Thus, the
          property holds.

        \item If $\loc_1 \in \LocsMin$ and
          $\minstrategy(\convert_{\wf \to
            \emath}^\minstrategy(\play_1)) = (\trans, \delay^\C)$,
          then
          $\convert_{\wf \to \emath}^\minstrategy(\play) =
          \convert_{\wf \to \emath}^\minstrategy(\play_1)
          \xrightarrow{\trans, \delay^\C} (\tuple{\loc, a}, \val^\C)$.
          Thus, in this case,
          $\val^\C - \val = \val^\C - \val_1 + \delay$ where $\delay$
          is chosen by $\minstrategy_{\wf}(\play_1)$.  In particular,
          by definition of $\minstrategy_{\wf}(\play_1)$, we have
          $\delay \in \{\delay^\C, \max(0, i - N \cdot \perturbation)\}$.\\
          First, if we suppose that $\delay = \delay^\C$, then
          $\val^\C - \val = \val_1^\C + \delay^\C - \val_{1} -
          \delay^\C = \val_1^\C - \val_1$
          and we conclude by applying the induction hypothesis to $\play_1$.\\
          Otherwise, $\delay = \max(0, i - N \cdot \perturbation)$ and
          $\val^\C - \val = \val^\C - \val_1 - \max(0, i - N \cdot
          \perturbation)$.  In particular, we have
		\begin{displaymath}
			\val^\C - \val = \val_1^\C + \delay^\C - \val_{1} - \max(0, i - N \cdot \perturbation)
			= \underbrace{\val_1^\C - \val_{1}}_{\geq 0} + \underbrace{\delay^\C - \max(0, i - N \cdot \perturbation)}_{\geq 0}
			\geq 0
		\end{displaymath}
		by induction hypothesis applied to $\play_1$ and
                since, in this case, $\delay^\C \geq \delay$. Indeed,
                by definition of $\minstrategy_{\wf}$,
                $\delay = \max(0, i - N \cdot \perturbation)$ when
                $\val_1 + \delay^\C > i - N \cdot \perturbation$.

		Moreover, since $- \max(0, i - N \cdot \perturbation) \leq -i + N \cdot \perturbation$, $\val_1^\C \geq 0$, and
		by definition of $\minstrategy_{\wf}$,  $\val^\C \leq i$, we deduce that
		\begin{displaymath}
			\val^\C - \val = \val^\C - \val_{1} - \max(0, i - N \cdot \perturbation)
			\leq i - i + N \cdot \perturbation = N \cdot \perturbation \,.
		\end{displaymath}

              \item In the last case,
                $\convert_{\wf \to \emath}^\minstrategy(\play) =
                \convert_{\wf \to \emath}^\minstrategy(\play_1)
                \xrightarrow{\trans, \delay} (\tuple{\loc, a}, \val)$,
                i.e. $\val = \val^\C$.  Thus, the property holds.
	\end{itemize}

	\noindent
	\noindent\textbf{Proof of~\cref{itm:copy_e-to-wf_convert-weight}}:
	The proof is again by induction on the length of $\play$.
        
	If $\play = (\tuple{\loc, a}, \val)$, then
        $\convert_{\wf \to \emath}^\minstrategy(\play) = (\tuple{\loc,
          a},\val)$ and the property trivially holds.

	Otherwise, assume that
        $\play = \play_1 \xrightarrow{\trans, \delay} (\tuple{\loc,
          a}, \val)$.  We write $\tuple{\loc_1, b}$
        (resp. $\tuple{\loc_1^\C, c}$) for the last location of
        $\play_1$ (resp. of
        $\convert_{\wf \to \emath}^\minstrategy(\play_1)$), and we
        distinguish three cases corresponding to the case distinction
        in the definition of the mapping
        $\convert_{\wf \to \emath}^\minstrategy$.
	\begin{itemize}
        \item If $\trans$ resets $\clockx$, then
          $\convert_{\wf \to \emath}^\minstrategy(\play) = \convert_{wf \to
            e}^\minstrategy(\play_1) \xrightarrow{\trans, \delay}
          (\tuple{\loc, a}, \val)$ and
		\begin{align*}
			\weightC(\play) &= \weightC(\play_1) + \weight(\trans) + \delay \cdot \weight(\tuple{\loc, a}) \\
			&\leq  \weightC(\convert_{\wf \to \emath}^\minstrategy(\play_1)) +
                   N \cdot \perturbation \cdot \maxWeightLoc \cdot |\play_1| + \weight(\trans) + \delay \cdot \weight(\tuple{\loc, a})\\
                  			&\hspace{5cm}\text{(by induction hypothesis)}\\
				&= \weightC(\convert_{\wf \to \emath}^\minstrategy(\play_1))  +
			N \cdot \perturbation \cdot \maxWeightLoc \cdot |\play_1| + \weight(\trans) + \delay \cdot \weight(\tuple{\loc_1^\C, c})
			\\
			&\hspace{5cm}\text{(by~\cref{def:copy} and~\cref{itm:rval-eq-val_restrictstrat_last-val})}\\
			&= \weightC(\convert_{\wf \to \emath}^\minstrategy(\play)) + N \cdot \perturbation \cdot \maxWeightLoc \cdot |\play|
		\end{align*}

		\item If $\loc_1 \in \LocsMin$ and
		$\minstrategy(\convert_{\wf \to \emath}^\minstrategy(\play_1)) = (\trans, \delay^\C)$,
		then $\convert_{\wf \to \emath}^\minstrategy(\play) =
		\convert_{\wf \to \emath}^\minstrategy(\play_1) \xrightarrow{\trans, \delay^\C} (\tuple{\loc, a}, \val^\C)$.
		where $\delay \in \{\delay^\C,  \max(0, i - N \cdot \perturbation)\}$ (as in~\cref{itm:rval-eq-val_restrictstrat_last-val}).
		Thus, if $\delay = \delay^\C$, then
		\begin{align*}
			\weightC(\play) &= \weightC(\play_1) + \weight(\trans) + \delay \cdot \weight(\tuple{\loc, a}) \\
			&\leq  \weightC(\convert_{\wf \to \emath}^\minstrategy(\play_1)) + N \cdot \perturbation \cdot \maxWeightLoc \cdot |\play_1|
                   + \weight(\trans) + \delay \cdot \weight(\tuple{\loc, a})\\
                  	&\hspace{5cm}\text{(by induction hypothesis)}\\
			&= \weightC(\convert_{\wf \to \emath}^\minstrategy(\play_1))  + N \cdot \perturbation \cdot \maxWeightLoc \cdot
			|\play_1| + \weight(\trans) + \delay \cdot \weight(\tuple{\loc_1^\C, c})
			\\
			&\hspace{5cm}\text{(by~\cref{def:copy} and $\delay = \delay^\C$)}\\
			&= \weightC(\convert_{\wf \to \emath}^\minstrategy(\play)) + N \cdot \perturbation \cdot \maxWeightLoc \cdot |\play|
		\end{align*}
		Otherwise, $\delay = \max(0, i - N \cdot \perturbation)$, then
		we distinguish two cases.
		If $\delay = 0$, then
		\begin{align*}
			\weightC(\play) &= \weightC(\play_1) + \weight(\trans) \\
			&\leq  \weightC(\convert_{\wf \to \emath}^\minstrategy(\play_1)) + N \cdot \perturbation \cdot \maxWeightLoc \cdot
			|\play_1| + \weight(\trans) + \delay^\C \cdot \weight(\tuple{\loc, a})
                   -\delay^\C \cdot \weight(\tuple{\loc, a})\\
                  	&\hspace{5cm}\text{(by induction hypothesis)}\\
			&= \weightC(\convert_{\wf \to \emath}^\minstrategy(\play_1))  + N \cdot \perturbation \cdot \maxWeightLoc \cdot
			|\play_1| + \weight(\trans) + \delay^\C \cdot \weight(\tuple{\loc_1^\C, c})
			-\delay^\C \cdot \weight(\tuple{\loc, a})\\
			&\hspace{5cm}\text{(by~\cref{def:copy})}\\
			&= \weightC(\convert_{\wf \to \emath}^\minstrategy(\play)) + N \cdot \perturbation \cdot \maxWeightLoc \cdot
			|\play_1| -\delay^\C \cdot \weight(\tuple{\loc, a}) \\
			&\leq \weightC(\convert_{\wf \to \emath}^\minstrategy(\play)) + N \cdot \perturbation \cdot \maxWeightLoc \cdot
			|\play_1| + \delay^\C \cdot \maxWeightLoc
			\qquad \text{(since $\weight(\tuple{\loc, a}) \leq  \maxWeightLoc$)} \\
			& \leq \weightC(\convert_{\wf \to \emath}^\minstrategy(\play)) + N \cdot \perturbation
			\cdot \maxWeightLoc \cdot |\play_1| + N \cdot \perturbation \cdot \maxWeightLoc \\
			&\hspace{3cm} \text{(by definition of $\minstrategy_{\wf}$ - it is the unique case where $\trans = 0$)} \\
			& \leq \weightC(\convert_{\wf \to \emath}^\minstrategy(\play)) + N \cdot \perturbation \cdot \maxWeightLoc \cdot |\play|
			\qquad\text{(by definition of $\perturbation$)}
		\end{align*}
		Otherwise $\delay = i - N \cdot \perturbation$, then
		\begin{align*}
			\weightC(\play) &= \weightC(\play_1) + \weight(\trans)+  (i - N \cdot \perturbation) \cdot \weight(\tuple{\loc, a})\\
			&\leq  \weightC(\convert_{\wf \to \emath}^\minstrategy(\play_1)) +  (i - N \cdot \perturbation) \cdot \weight(\tuple{\loc, a})
			+ N \cdot \perturbation \cdot \maxWeightLoc \cdot
                   |\play_1| + \weight(\trans) )\\
                  	&\hspace{5cm}\text{(by induction hypothesis)}\\
			&= \weightC(\convert_{\wf \to \emath}^\minstrategy(\play_1)) +  (i - N \cdot \perturbation) \cdot \weight(\tuple{\loc, a})
			-\delay^\C \cdot \weight(\tuple{\loc, a})\\
			&\hspace{5cm}\text{(by the same reasoning as above)}\\
			&= \weightC(\convert_{\wf \to \emath}^\minstrategy(\play_1)) +
			(i - N \cdot \perturbation + \delay^\C) \cdot \maxWeightLoc
			\qquad\text{(since $|\weight(\tuple{\loc, a})| \leq \maxWeightLoc$)}\\
			&\leq \weightC(\convert_{\wf \to \emath}^\minstrategy(\play_1)) +
			(i - N \cdot \perturbation + N \cdot \perturbation) \cdot \maxWeightLoc
			\\
			&\hspace{3cm}\text{(since $\delay^\C \leq i - \val_1 + N \cdot \perturbation$ by~\cref{itm:rval-eq-val_restrictstrat_last-val}
				and $\val_1 \leq i - N \cdot  \perturbation$ by $\minstrategy_{\wf}$) }\\
			& \leq \weightC(\convert_{\wf \to \emath}^\minstrategy(\play)) + N \cdot \perturbation \cdot \maxWeightLoc \cdot |\play|
			\qquad\text{(by definition of $\perturbation$)}
		\end{align*}

		\item Otherwise, $\convert_{\wf \to \emath}^\minstrategy(\play) =
		\convert_{\wf \to \emath}^\minstrategy(\play_1) \xrightarrow{\trans, \delay} (\tuple{\loc, a}, \val)$,
		i.e. $\val = \val^\C$.
		Thus,
		\begin{align*}
			\weightC(\play) &= \weightC(\play_1) + \weight(\trans) + \delay \cdot \weight(\tuple{\loc, a}) \\
			&\leq  \weightC(\convert_{\wf \to \emath}^\minstrategy(\play_1)) +
                   N \cdot \perturbation \cdot \maxWeightLoc \cdot |\play_1| + \weight(\trans) + \delay \cdot \weight(\tuple{\loc, a})\\
                  	&\hspace{5cm}\text{(by induction hypothesis)}\\
			&= \weightC(\convert_{\wf \to \emath}^\minstrategy(\play_1))  +
			N \cdot \perturbation \cdot \maxWeightLoc \cdot |\play_1| + \weight(\trans) + \delay \cdot \weight(\tuple{\loc_1^\C, c})
			\\
			&\hspace{5cm}\text{(by~\cref{def:copy} and~\cref{itm:rval-eq-val_restrictstrat_last-val})}\\
			&= \weightC(\convert_{\wf \to \emath}^\minstrategy(\play)) + N \cdot \perturbation \cdot \maxWeightLoc \cdot |\play|
		\end{align*}
              \end{itemize}
              This concludes the proof
              of~\cref{lem:copy_e-to-wf_convert}.  \qedhere
\end{proof}

To wrap up the proof of~\cref{lem:copy_Min_rstrict}, given a play
$\play$ that conforms to $\minstrategy_{\wf}$, we have,
by~\cref{itm:copy_e-to-wf_convert-weight},
$\weightC(\play) \leq \weightC(\convert_{wf \to
  e}^\minstrategy(\play)) + N \cdot \perturbation \cdot \maxWeightLoc
\cdot |\play|$.  Moreover, by definition of $\perturbation$ and since
$|\play| \leq \kappa$, we deduce that
$\weightC(\play) \leq \weightC(\convert_{wf \to
  e}^\minstrategy(\play)) + \varepsilon/2$.  Now, since $\play$ and
$\weightC(\convert_{\wf \to \emath}^\minstrategy(\play))$ end in the same
location (by~\cref{itm:rval-eq-val_restrictstrat_last-loc}), we deduce
that
$\weight(\play) \leq \weight(\convert_{\wf \to \emath}^\minstrategy(\play))
+ \varepsilon/2$.  Finally, by hypothesis on $\minstrategy$ and since
$\convert_{\wf \to \emath}^\minstrategy(\play)$ conforms to $\minstrategy$
(by~\cref{itm:rval-eq-val_restrictstrat_conf}), we deduce that
$\weight(\play) \leq \Value(\loc, \val) + \varepsilon$ where
$(\loc, \val)$ is the initial configuration of $\play$.  This
concludes the proof of~\cref{lem:copy_Min_rstrict}. 
\qed

%% file: app_copy-rob-eq-val.tex
\section[Proof of Lemma 18]{Proof of~\cref{lem:copy_wf-to-r}}
\label{app:copy_wf-to-r}

\copyWFtoR*

\bigskip

\noindent
\textbf{\cref{itm:copy_wf-to-r_last-loc}} is immediate by definition
of $\convert_{\wf \to \rmath}^{\robmaxstrategy}$.

\medskip

\noindent
\textbf{Proof of~\cref{itm:copy_wf-to-r_last-val}}:
We reason by induction on the length of $\play$. Assume first that
$\play = (\tuple{\loc, 0}, \val)$.  In this case,
$\convert_{\wf \to \rmath}^{\robmaxstrategy}(\play) = (\tuple{\loc,
  0}, \val) = \play$.  Thus, the property holds.

\smallskip

Then, suppose that $\play = \play_1 \xrightarrow{\trans, \delay} s$
with $\last(\play_1)$ a real state of \MinPl,
and we let $\val$ (resp. $\val^{c}$, $\val_1$, and
$\val_1^{c}$) be the last valuation of $\play$
(resp. of $\convert_{\wf \to \rmath}^{\robmaxstrategy}(\play) $, $\play_1$, and
$\convert_{\wf \to \rmath}^{\robmaxstrategy}(\play_1) $). In this case,
we distinguish two cases according the definition of $\convert_{\wf \to \rmath}^{\robmaxstrategy}(\play_1)$.
First, consider that $\val - \val^c_1 > 0$, i.e. $\trans$ does not contain a reset since otherwise
$\val = 0$ and $\val - \val^c_1 \leq 0$, and we have
$\convert_{\wf \to \rmath}^{\robmaxstrategy}(\play) =
	\convert_{\wf \to \rmath}^{\robmaxstrategy}(\play_1) \xrightarrow{\trans, \val - \val_1^c} s'
	\xrightarrow{\trans, \delay^p} s''$.
In particular, $\val^{c} - \val = \val^{c}_1 + (\val - \val^c_1) + \delay^p - \val = \delay^p$.
Thus, since $0 \leq \delay^p \leq \perturbation$, we deduce that
$0 \leq \val^{c} - \val \leq \perturbation \leq \perturbation \cdot |\play|$.

Otherwise, we have
$\convert_{\wf \to \rmath}^{\robmaxstrategy}(\play) = \convert_{\wf
  \to \rmath}^{\robmaxstrategy}(\play_1) \xrightarrow{\trans, \delay}
s' \xrightarrow{\trans, \delay^p} s''$.  In particular,
$\val^{c} - \val = \val^{c}_1 + \delay + \delay^p - \val_{1} - \delay
= \val^c_1 - \val_1 + \delay^p$.  Thus, since $0 \leq \delay^p$ and
$0 \leq \val^c_1 - \val_1$ (by induction hypothesis), we deduce that
$\val^{c} - \val \geq 0$.  Conversely, since
$\delay^p \leq \perturbation$ and
$\val^c_1 - \val_1 \leq \perturbation \cdot |\play_1|$ (by induction
hypothesis), we deduce that
$\val^{c} - \val \leq \perturbation + \perturbation \cdot |\play_1| =
\perturbation \cdot |\play|$.  This concludes this case.

\smallskip

Finally, we suppose that
$\play = \play_1 \xrightarrow{\trans, \delay} s$ with $\last(\play_1)$
a real state of \MaxPl and we let $\val$ (resp. $\val^{c}$, $\val_1$,
and $\val_1^{c}$) be the last valuation of $\play$ (resp. of
$\convert_{\wf \to \rmath}^{\robmaxstrategy}(\play) $, $\play_1$, and
$\convert_{\wf \to \rmath}^{\robmaxstrategy}(\play_1) $). In this
case, we distinguish two subcases according to the definition of
$\convert_{\wf \to \rmath}^{\robmaxstrategy}(\play_1)$.  First,
consider that
$\robmaxstrategy(\convert_{\wf \to \rmath}^{\robmaxstrategy}(\play_1))
= (\trans, \delay - \val^c_1 + \val_1)$ and we have
$\convert_{\wf \to \rmath}^{\robmaxstrategy}(\play) = \convert_{\wf
  \to \rmath}^{\robmaxstrategy}(\play_1)
\xrightarrow{\robmaxstrategy(\convert_{\wf \to
    \rmath}^{\robmaxstrategy}(\play_1) } s$.  In particular,
$\val^c - \val = \val_c^1 + (\delay - \val^c_1 + \val_1) - \val_1 -
\delay = 0$.  Thus,
$0 \leq \val^{c} - \val \leq \perturbation \cdot |\play|$.

Otherwise, we have $\convert_{\wf \to \rmath}^{\robmaxstrategy}(\play) =
\convert_{\wf \to \rmath}^{\robmaxstrategy}(\play_1) \xrightarrow{\trans, \delay} s'$.
In particular, $\val^{c} - \val = \val^{c}_1 + \delay - \val_{1} - \delay = \val^c_1 - \val_1$.
Thus, by induction hypothesis applied to $\play_1$, we deduce that
$0 \leq \val^{c} - \val \leq \perturbation \cdot |\play_1| \leq \perturbation \cdot |\play|$.
Thus, the property holds.

\medskip

\noindent
\textbf{Proof of~\cref{itm:copy_wf-to-r_weight}}:
Here again, we reason by induction on the length of $\play$.
First, we suppose that $\play = (\tuple{\loc, 0}, \val)$.
In this case, $\convert_{\wf \to \rmath}^{\robmaxstrategy}(\play) = (\tuple{\loc, 0}, \val) = \play$
and $\weightC(\play) = 0 = \weightC(\convert_{\wf \to \rmath}^{\robmaxstrategy}(\play))$.

Then, we suppose that
$\play = \play_1 \xrightarrow{\trans, \delay} s$, with $s$ a real
state of \MinPl, and we distinguish two cases. By letting $\val$
(resp. $\val_1^c$) be the last valuation of $\play$
(resp. $\convert_{\wf \to \rmath}^{\robmaxstrategy}(\play_1)$), we
first consider that $\val - \val_1^c > 0$, i.e. $\trans$ does not
contain a reset since otherwise $\val = 0$ and
$\val - \val^c_1 \leq 0$. In this case,
$\convert_{\wf \to \rmath}^{\robmaxstrategy}(\play) = \convert_{\wf
  \to \rmath}^{\robmaxstrategy}(\play_1) \xrightarrow{\trans, \val -
  \val_1^c} s' \xrightarrow{\trans, \delay^p} s''$.  In particular, we
have
\begin{align*}
	\weightC(\convert_{\wf \to \rmath}^{\robmaxstrategy}(\play))
	&= \weightC(\convert_{\wf \to \rmath}^{\robmaxstrategy}(\play_1)) +
	(\val - \val_1^c + \delay^{\perturbation}) \cdot \weight(\loc) + \weight(\trans) \\
	&\leq \weightC(\play_1) +
	(\val - \val_1^c + \delay^{\perturbation}) \cdot \weight(\loc) + \weight(\trans) + \maxWeightLoc \cdot \perturbation \cdot |\play_1|^2 \\
	& \hspace{5cm} \text{(by induction hypothesis)}\\
	&= \weightC(\play_1) +
	(\val_1 + \delay - \val_1^c + \delay^{\perturbation}) \cdot \weight(\loc) + \weight(\trans) +
	\maxWeightLoc \cdot \perturbation \cdot |\play_1|^2 \\
	& \hspace{5cm} \text{(since $\trans$ does not contain a reset)}\\
	&= \weightC(\play) +
	(\val_1  - \val_1^c + \delay^{\perturbation}) \cdot \weight(\loc)  +
	\maxWeightLoc \cdot \perturbation \cdot |\play_1|^2
\end{align*}
By~\cref{itm:copy_wf-to-r_last-val}, we note that
$-\perturbation \cdot |\play_1| \leq \val_1 - \val_1^c \leq 0$,
i.e.
$-\perturbation \cdot |\play_1| \leq \val_1 - \val_1^c + \delay^p \leq
\perturbation$.  To conclude this case, we distinguish two subcases
according to the sign of $ \weight(\loc)$.
\begin{itemize}
	\item First, we suppose that $\weight(\loc) \geq 0$, and we deduce
	\begin{align*}
		\weightC(\convert_{\wf \to \rmath}^{\robmaxstrategy}(\play))
		&\leq \weightC(\play) +
		\delay^{\perturbation} \cdot \weight(\loc)  +
		\maxWeightLoc \cdot \perturbation \cdot |\play_1|^2 \\
		&\leq \weightC(\play) +
		\perturbation \cdot \maxWeightLoc  +
		\maxWeightLoc \cdot \perturbation \cdot |\play_1|^2 \\
		&\leq \weightC(\play) + \maxWeightLoc \cdot \perturbation \cdot |\play|^2
	\end{align*}
	as claimed.
	\item Otherwise $\weight(\loc) < 0$, and we have
	\begin{align*}
		\weightC(\convert_{\wf \to \rmath}^{\robmaxstrategy}(\play))
		&\leq \weightC(\play) -
		(\val_1  - \val_1^c + \delay^{\perturbation}) \cdot |\weight(\loc)|  +
		\maxWeightLoc \cdot \perturbation \cdot |\play_1|^2 \\
		&= \weightC(\play) +
		(\val_1^c - \val_1 - \delay^{\perturbation}) \cdot |\weight(\loc)|  +
		\maxWeightLoc \cdot \perturbation \cdot |\play_1|^2 \\
		&\leq \weightC(\play) +
		(\perturbation  \cdot |\play_1| - \perturbation) \cdot |\weight(\loc)|  +
		\maxWeightLoc \cdot \perturbation \cdot |\play_1|^2
		\qquad\quad\text{(by~\cref{itm:copy_wf-to-r_last-val})}\\
		&\leq \weightC(\play) +
		\perturbation \cdot |\play_1|  \cdot |\weight(\loc)|  +
		\maxWeightLoc \cdot \perturbation \cdot |\play_1|^2 \\
		&\leq \weightC(\play) +
		\perturbation \cdot |\play_1|  \cdot \maxWeightLoc  +
		\maxWeightLoc \cdot \perturbation \cdot |\play_1|^2 \\
		&\leq \weightC(\play) + \maxWeightLoc \cdot \perturbation \cdot |\play|^2
	\end{align*}
	as claimed.
\end{itemize}

Now, we suppose that
$\convert_{\wf \to \rmath}^{\robmaxstrategy}(\play) =
\convert_{\wf \to \rmath}^{\robmaxstrategy}(\play_1) \xrightarrow{\trans,\delay}
s' \xrightarrow{\trans, \delay^p} s''$, and we have
\begin{align*}
	\weightC(\convert_{\wf \to \rmath}^{\robmaxstrategy}(\play))
	&= \weightC(\convert_{\wf \to \rmath}^{\robmaxstrategy}(\play_1)) +
	(\delay + \delay^{\perturbation}) \cdot \weight(\loc) + \weight(\trans) \\
	&\leq \weightC(\play_1) +
	(\delay + \delay^{\perturbation}) \cdot \weight(\loc) + \weight(\trans) + \maxWeightLoc \cdot \perturbation \cdot |\play_1|^2 \\
	& \hspace{5cm} \text{(by induction hypothesis)}\\
	&= \weightC(\play_1) + \delay \cdot \weight(\loc) + \weight(\trans) +
	\delay^{\perturbation} \cdot \weight(\loc) +
	\maxWeightLoc \cdot \perturbation \cdot |\play_1|^2 \\
	&= \weightC(\play) +
	\delay^{\perturbation} \cdot \weight(\loc)  +
	\maxWeightLoc \cdot \perturbation \cdot |\play_1|^2 \\
	&\leq \weightC(\play) +
	\delay^{\perturbation} \cdot \maxWeightLoc  +
	\maxWeightLoc \cdot \perturbation \cdot |\play_1|^2 \\
	&\leq \weightC(\play) +
	\perturbation \cdot \maxWeightLoc  +
	\maxWeightLoc \cdot \perturbation \cdot |\play_1|^2 \\
	&\leq \weightC(\play) +
	\maxWeightLoc \cdot \perturbation \cdot |\play|^2
\end{align*}
as claimed.
Thus, the property holds in this case.

\medskip

Finally, we suppose that $\last(\play) \in \ConfMax$ and, as before,
we distinguish two cases according to the definition of
$\convert_{\wf \to \rmath}^{\robmaxstrategy}(\play_1)$.  First,
consider that
$\robmaxstrategy(\convert_{\wf \to \rmath}^{\robmaxstrategy}(\play_1))
= (\trans, \delay - \val^c_1 + \val_1)$ and
$\convert_{\wf \to \rmath}^{\robmaxstrategy}(\play) = \convert_{\wf
  \to \rmath}^{\robmaxstrategy}(\play_1)
\xrightarrow{\robmaxstrategy(\convert_{\wf \to
    \rmath}^{\robmaxstrategy}(\play_1) } s$.  In particular, we have
\begin{align*}
	\weightC(\convert_{\wf \to \rmath}^{\robmaxstrategy}(\play))
	&= \weightC(\convert_{\wf \to \rmath}^{\robmaxstrategy}(\play_1)) +
	\delay^c \cdot \weight(\loc) + \weight(\trans)  \\
	&\leq \weightC(\play_1) +
	\delay^c \cdot \weight(\loc) + \weight(\trans) +
	\maxWeightLoc \cdot \perturbation \cdot |\play_1|^2  \\
	& \hspace{5cm} \text{(by induction hypothesis)}\\
	&= \weightC(\play_1) +
	(\delay - \val_1^c + \val_1) \cdot \weight(\loc) + \weight(\trans) + \maxWeightLoc \cdot \perturbation \cdot |\play_1|^2 \\
	&=  \weightC(\play)   +  (\val_1 - \val_1^c) \cdot \weight(\loc) +
	\maxWeightLoc \cdot \perturbation \cdot |\play_1|^2
\end{align*}
Similar to the case of states of \MinPl,
by~\cref{itm:copy_wf-to-r_last-val}, we get that
$-\perturbation \cdot |\play_1| \leq \val_1 - \val_1^c \leq 0$.  Thus,
to conclude this case, we distinguish two subcases according to the
sign of $ \weight(\loc)$.
\begin{itemize}
	\item First, we suppose that $\weight(\loc) \geq 0$, and we deduce
	\begin{align*}
		\weightC(\convert_{\wf \to \rmath}^{\robmaxstrategy}(\play))
		&\leq \weightC(\play) +
		\maxWeightLoc \cdot \perturbation \cdot |\play_1|^2
		\leq \weightC(\play) + \maxWeightLoc \cdot \perturbation \cdot |\play|^2
	\end{align*}
	as claimed.
	\item Otherwise $\weight(\loc) < 0$, and we have
	\begin{align*}
		\weightC(\convert_{\wf \to \rmath}^{\robmaxstrategy}(\play))
		&\leq \weightC(\play) -
		(\val_1  - \val_1^c) \cdot |\weight(\loc)|  +
		\maxWeightLoc \cdot \perturbation \cdot |\play_1|^2 \\
		&= \weightC(\play) +
		(\val_1^c - \val_1) \cdot |\weight(\loc)|  +
		\maxWeightLoc \cdot \perturbation \cdot |\play_1|^2 \\
		&\leq \weightC(\play) +
		(\perturbation  \cdot |\play_1|) \cdot |\weight(\loc)|  +
		\maxWeightLoc \cdot \perturbation \cdot |\play_1|^2
		\qquad\quad\text{(by~\cref{itm:copy_wf-to-r_last-val})}\\
		&\leq \weightC(\play) +
		\perturbation \cdot |\play_1|  \cdot \maxWeightLoc  +
		\maxWeightLoc \cdot \perturbation \cdot |\play_1|^2 \\
		&\leq \weightC(\play) + \maxWeightLoc \cdot \perturbation \cdot |\play|^2
	\end{align*}
	as claimed.
\end{itemize}

Now, we suppose that
$\convert_{\wf \to \rmath}^{\robmaxstrategy}(\play) =
\convert_{\wf \to \rmath}^{\robmaxstrategy}(\play_1) \xrightarrow{\trans,\delay}
s' $, and we have
\begin{align*}
	\weightC(\convert_{\wf \to \rmath}^{\robmaxstrategy}(\play))
	&= \weightC(\convert_{\wf \to \rmath}^{\robmaxstrategy}(\play_1)) +
	\delay \cdot \weight(\loc) + \weight(\trans) \\
	&\leq \weightC(\play_1) +
	\delay  \cdot \weight(\loc) + \weight(\trans) + \maxWeightLoc \cdot \perturbation \cdot |\play_1|^2 \\
	& \hspace{5cm} \text{(by induction hypothesis)}\\
	&= \weightC(\play) +
	\maxWeightLoc \cdot \perturbation \cdot |\play_1|^2 \\
	&\leq \weightC(\play) +
	\maxWeightLoc \cdot \perturbation \cdot |\play|^2
\end{align*}
as claimed.  Thus, the property also holds in this case.  This
concludes the proof of this lemma. \qed

\newpage
\section[Proof of Lemma 19]{Proof of~\cref{lem:rval-eq-val_to-exact_conf}}
\label{app:rval-eq-val_to-exact_conf}

\eqValConf*

We reason by induction on the length of $\play$.  If
$\play = (\tuple{\loc, 0}, \val)$, then
$\convert_{\wf \to \rmath}^{\robmaxstrategy}(\play) = \play$ and the
property trivially holds.  Otherwise, we write
$\play = \play_1 \xrightarrow{\trans, \delay} (\tuple{\loc, a}, \val)$
and we distinguish two cases depending on the player owning the last
configuration.

If $\last(\play_1) \in \ConfMin$, then
$\last(\convert_{\wf \to \rmath}^{\robmaxstrategy}(\play_1)) \in
\ConfMin$ (by~\cref{itm:copy_wf-to-r_last-loc}
of~\cref{lem:copy_wf-to-r} applied to $\play_1$).  Applying the
induction hypothesis to $\play_1$ and using the fact that
$\convert_{\wf \to \rmath}^{\robmaxstrategy}$ follows
$\robmaxstrategy$ to choose the perturbation, we can conclude.

Otherwise, $\last(\play_1) \in \ConfMax$. Then,
\cref{itm:copy_wf-to-r_last-loc} of~\cref{lem:copy_wf-to-r} applied to
$\play_1$ entails
$\last(\convert_{\wf \to \rmath}^{\robmaxstrategy}(\play_1)) \in
\ConfMax$.  Applying the induction hypothesis to $\play_1$ and using
the fact that $\convert_{\wf \to \rmath}^{\robmaxstrategy}$ follows
$\robmaxstrategy$ to choose the delay $t$ and the transition $\trans$,
we can conclude. \qed